%% file: main.tex
\PassOptionsToPackage{table}{xcolor}
\documentclass[acmsmall,screen,nonacm]{acmart}

\input{packages}
\input{macros}

\begin{document}

\title{A Type System for Communication-Safe Hardware Design}

\newcommand{\nusaffiliation}{
\affiliation{%
  \institution{National University of Singapore}
  \country{Singapore}
}
}

\author{Aditya Ranjan Jha}
\authornote{Equal contribution.}
\email{adityajha@u.nus.edu}
\orcid{0009-0005-3553-0903}
\nusaffiliation{}

\author{Jason Zhijingcheng Yu}
\authornotemark[1]
\authornote{Work conducted while at NUS. Currently a visiting researcher at Jane Street.}
\email{yu.z@u.nus.edu}
\orcid{0000-0001-6013-157X}
\nusaffiliation{}

\author{Umang Mathur}
\email{umathur@nus.edu.sg}
\orcid{0000-0002-7610-0660}
\nusaffiliation{}

\author{Prateek Saxena}
\email{prateeks@comp.nus.edu.sg}
\orcid{0000-0002-1875-8675}
\nusaffiliation{}

\input{sections/abstract}


\maketitle

\input{sections/intro}
\input{sections/overview}

\input{sections/pact}

\input{sections/types}

\input{sections/evaluation}
\input{sections/related}

\input{sections/conclusions}
\input{sections/acknowledgement}
\input{sections/data-availability}

\bibliographystyle{ACM-Reference-Format}
\bibliography{references}

\newpage

\appendix
\input{sections/appendix}

\end{document}

%% file: packages.tex
\usepackage{tcolorbox}
\tcbuselibrary{breakable,skins}
\usepackage{needspace}
\usepackage{ebproof}
\usepackage{subcaption}
\usepackage{listings}
\usepackage{wrapfig}
\usepackage{multirow}
\usepackage{minted}
\usepackage[nointegrals]{wasysym}
\usepackage{syntax}
\usepackage{tikz}

\usepackage{amssymb}
\usetikzlibrary{decorations.pathreplacing,angles,quotes}
\usepackage{adjustbox}
\usepackage{comment}
\usepackage{booktabs}
\usepackage[inkscapearea=page]{svg}
\usepackage{macrodefs}
\usepackage[inline]{enumitem}
\usepackage{pifont}
\usepackage{tipa}
\usepackage{amsmath}
\usepackage{thm-restate}
\usepackage{hyperref}
\usepackage{xcolor}

%% file: macros.tex
\newcommand{\figlabel}[1]{\label{fig:#1}}
\newcommand{\figref}[1]{Fig.~\ref{fig:#1}}

\newcommand{\tabref}[1]{Table~\ref{tab:#1}}
\newcommand{\seclabel}[1]{\label{sec:#1}}
\newcommand{\secref}[1]{Section~\ref{sec:#1}}

\newcommand{\deflabel}[1]{\label{def:#1}}
\newcommand{\defref}[1]{Definition~\ref{def:#1}}

\newcommand{\thmlabel}[1]{\label{thm:#1}}
\newcommand{\thmref}[1]{Theorem~\ref{thm:#1}}
\newcommand{\proplabel}[1]{\label{prop:#1}}
\newcommand{\propref}[1]{Proposition~\ref{prop:#1}}
\newcommand{\lemlabel}[1]{\label{lem:#1}}
\newcommand{\lemref}[1]{Lemma~\ref{lem:#1}}
\newcommand{\corlabel}[1]{\label{cor:#1}}
\newcommand{\corref}[1]{Corollary~\ref{cor:#1}}

\newcommand{\applabel}[1]{\label{app:#1}}
\newcommand{\appref}[1]{Appendix~\ref{app:#1}}

\theoremstyle{definition}
\newtheorem{theorem}{Theorem}[section]
\newtheorem{lemma}{Lemma}[section]
\newtheorem{proposition}{Proposition}[section]

\newtheorem{corollary}{Corollary}[section]

\newtheorem{definition}{Definition}[section]

\newenvironment{myparagraph}[1]{\vspace{0.1in}\noindent{\bf #1.}\enspace\ignorespaces}{}

\definecolor{papayawhip}{rgb}{1.0, 0.94, 0.84}

\setminted{
  fontsize=\tiny,
  breaklines=true,
  frame=lines,
  bgcolor=papayawhip,
  framesep=2mm,
  baselinestretch=1.0
}

\lstdefinelanguage{SystemVerilog}{
  keywords={module, input, output, logic},
  keywordstyle=\color{blue}\bfseries,
  sensitive=true,
  morecomment=[l]{//},
  morestring=[b]"
}

\lstdefinelanguage{Anvil}{
  keywords={let, send, recv, proc,loop, chan, cycle ,spawn},
  keywordstyle=\color{purple}\bfseries,
  sensitive=true,
  morecomment=[l]{//},
  morestring=[b]"
}

\lstdefinestyle{anvilstyle}{
  language=Anvil,
  basicstyle=\tiny\ttfamily,
  breaklines=true,
  frame=lines,
  backgroundcolor=\color{papayawhip},
  framesep=2mm,
  columns=fullflexible,
  showstringspaces=false,
  commentstyle=\color{gray}\itshape,
  stringstyle=\color{brown}
}

\lstnewenvironment{anvilcode}{\lstset{style=anvilstyle}}{}

\newcommand{\anvil}{\textsf{Anvil}}
\DeclareRobustCommand{\codename}{\texorpdfstring{\textcolor{pactnamecolor}{\textsc{Pact}}}{Pact}}
\definecolor{vgreen}{RGB}{104,180,104}
\definecolor{vblue}{RGB}{49,49,255}
\definecolor{vorange}{RGB}{255,143,102}
\definecolor{vtypecompat}{RGB}{150,78,45}
\definecolor{vautstate}{RGB}{10,125,125}
\definecolor{vautenv}{RGB}{122,62,168}
\definecolor{vautact}{RGB}{197,48,138}

\makeatletter
\DeclareRobustCommand*{\escapeus}[1]{%
    \begingroup\@activeus\scantokens{#1 }\endgroup}
\begingroup\lccode`\~=`\_\relax
    \lowercase{\endgroup\def\@activeus{\catcode`\_=\active \let~\_}}
\makeatother

\DeclareRobustCommand{\pkw}[1]{%
    \ifmmode
        \text{\textcolor{vblue}{\texttt{#1}}}%
    \else
        \textcolor{vblue}{\texttt{#1}}%
    \fi
}

\newcommand{\psym}[1]{\textcolor{vblue}{#1}}
\newcommand{\pseq}{\mathbin{\psym{;}}}
\newcommand{\pparop}{\mathbin{\psym{\parallel}}}
\newcommand{\pmu}{\psym{\mu}}

\newcommand{\auxkw}[1]{\textcolor{darkred}{#1}}

\newcommand{\Messages}{\mathcal{M}}
\newcommand{\Pids}{\mathcal{T}}
\newcommand{\Etas}{\mathcal{H}}
\newcommand{\Act}{\text{Act}}

\newcommand{\Logs}{\mathcal{L}}
\newcommand{\Strategy}{\mathcal{S}}

\newcommand{\BoundaryMsgs}{\auxkw{\mathsf{B}}}

\newcommand{\PidSet}[1]{\Pids(#1)}
\newcommand{\EtaSet}[1]{\Etas(#1)}
\newcommand{\ActSet}[1]{\mathsf{A}_{#1}}

\newcommand{\Fut}{\auxkw{\mathsf{Fut}}}

\newcommand{\Imm}{\auxkw{\mathsf{Imm}}}

\newcommand{\Pol}{\auxkw{\mathsf{Pol}}}

\newcommand{\Need}{\auxkw{\mathsf{Need}}}

\newcommand{\Resolved}{\auxkw{\mathsf{Resolved}}}

\newcommand{\Avail}{\auxkw{\mathsf{Avail}}}
\newcommand{\LocalEnv}{\auxkw{\mathsf{LocalEnv}}}

\newcommand{\Sync}{\auxkw{\mathsf{Sync}}}
\newcommand{\SyncCompat}{\textcolor{vblue}{\mathsf{SyncCompat}}}

\newcommand{\Honest}{\textcolor{vblue}{\mathsf{Honest}}}

\newcommand{\Stuck}{\textcolor{vblue}{\mathsf{Stuck}}}
\newcommand{\Reach}{\mathsf{Reach}}
\newcommand{\Fair}{\textcolor{vblue}{\mathsf{Fair}}}
\newcommand{\Live}{\textcolor{vblue}{\mathsf{Live}}}
\newcommand{\ProcJoin}{\textcolor{vblue}{\scalebox{1.5}{$\Join$}}}
\newcommand{\Inf}{\mathsf{Inf}}
\newcommand{\Bare}{\auxkw{\mathsf{Bare}}}
\newcommand{\Sends}{\auxkw{\mathsf{Sends}}}
\newcommand{\Recvs}{\auxkw{\mathsf{Recvs}}}
\newcommand{\SyncedMsgs}{\auxkw{\mathsf{Synced}}}
\newcommand{\BEnv}{\mathcal{O}}
\newcommand{\IntMsgs}{\auxkw{\mathcal{I}\mkern-2mu\mathcal{M}}}
\newcommand{\tick}[1]{#1^{\checkmark}}
\newcommand{\unionI}{\cup_{I}}

\DeclareRobustCommand{\tkw}[1]{%
    \ifmmode
        \text{\textcolor{vtypecompat}{\texttt{#1}}}%
    \else
        \textcolor{vtypecompat}{\texttt{#1}}%
    \fi
}
\newcommand{\tsym}[1]{\textcolor{vtypecompat}{#1}}
\newcommand{\tmu}{\tsym{\mu}}
\newcommand{\tclock}{\tkw{\#}}
\newcommand{\tmarker}[1]{\tkw{\$}_{#1}}
\newcommand{\tseq}{\mathbin{\tkw{;}}}
\newcommand{\tplus}{\mathbin{\tkw{+}}}
\newcommand{\tfairplus}[1]{\mathbin{\textcolor{vtypecompat}{\dotplus}_{#1}}}
\newcommand{\tparop}{\mathbin{\textcolor{vtypecompat}{\parallel}}}
\newcommand{\tpar}[1]{\tparop_{#1}}
\newcommand{\toplus}{\textcolor{vtypecompat}{\oplus}}
\newcommand{\twedge}{\textcolor{vtypecompat}{\wedge}}

\newcommand{\TypeFair}{\textcolor{vtypecompat}{\mathsf{Fair}}}

\newcommand{\TypeCompat}{\textcolor{vtypecompat}{\mathsf{Compat}}}
\newcommand{\TypeWinRegion}{\textcolor{vtypecompat}{\mathsf{WinRegion}}}

\newcommand{\Rel}{\mathcal{R}}
\newcommand{\Aut}{\mathcal{A}}
\newcommand{\erase}{\mathsf{erase}}
\newcommand{\inferpact}[2]{#1 \Rightarrow #2}
\newcommand{\checkdo}[2]{#1 \Rightarrow_{\mathsf{do}} #2}
\newcommand{\qst}[1]{\textcolor{vautstate}{#1}}
\newcommand{\aenv}[1]{\textcolor{vautenv}{#1}}
\newcommand{\aact}[1]{\textcolor{vautact}{#1}}

\newcommand{\atrans}[7]{ \qst{#1} \xrightarrow[{\aenv{\langle#2,#3 \rangle}}]{\aact{#4,\langle#5,#6\rangle}}\qst{#7}}

\newcommand{\lstep}[4]{#1 \vdash #2 \xrightarrow{#3} #4}

\newcommand{\stratstep}[5]{#1 \vdash_{#2} #3 \xrightarrow{#4} #5}
\newcommand{\restr}[2]{#1\mathbin{\upharpoonright}_{#2}}

\newcommand{\Alphabet}{\auxkw{\Omega}}
\newcommand{\Off}{\auxkw{\mathsf{Off}}}
\newcommand{\FutOff}{\auxkw{\mathsf{FutOff}}}
\newcommand{\Hide}[2]{\auxkw{\mathsf{hide}}_{#1}(#2)}
\newcommand{\Complete}{\auxkw{\mathsf{c}}}
\newcommand{\Env}[1]{\mathsf{Env}_{#1}}

\newtcolorbox{semrulebox}[1][]{
    colback=gray!3,
    colframe=black!40,
    boxrule=0.35pt,
    arc=1pt,
    left=3pt,
    right=3pt,
    top=3pt,
    bottom=3pt,
    boxsep=0pt,
    #1
}

\newcommand{\term}[2]{%
 \hypertarget{#1}{\ifstrequal{#1}{def:comm-safety}%
   {\textcolor{commsafetycolor}{\emph{#2}}}{\emph{#2}}}}
\newcommand{\termref}[2]{%
 \ifstrequal{#1}{def:comm-safety}%
   {\hyperlink{#1}{\textcolor{commsafetycolor}{#2}}}{#2}}

\definecolor{pactcolor}{HTML}{2E5C8A}
\definecolor{commsafetycolor}{HTML}{2E5C8A}
\definecolor{pactnamecolor}{HTML}{A34A00}
\definecolor{svcolor}{HTML}{7A7A7A}
\definecolor{pactbg}{HTML}{EAF1FA}
\definecolor{svbg}{HTML}{F2F2F2}
\definecolor{markgood}{HTML}{1E7B3E}
\definecolor{markbad}{HTML}{B3261E}
\definecolor{markunk}{HTML}{8A8A8A}

\newcommand{\okmark}{\textcolor{markgood}{\checkmark}}
\newcommand{\badmark}{\textcolor{markbad}{\ding{55}}}
\newcommand{\unkmark}{\textcolor{markunk}{?}}
\newcommand{\pacthdr}[1]{\textcolor{pactcolor}{\textbf{#1}}}
\newcommand{\svhdr}[1]{\textcolor{svcolor}{\textbf{#1}}}

%% file: sections/abstract.tex
\begin{abstract}
Hardware components communicate through clock-cycle-sensitive protocols.
\termref{def:hdl}{Hardware Description Languages (HDLs)}, however, provide limited support for expressing these protocols or for
checking that independently developed components agree on them.
When components disagree, the composed design can perform unsynchronized
exchanges or stop making progress.
We call these failures \termref{def:comm-safety}{communication-safety} violations.
\termref{def:comm-safety}{Communication safety} is therefore a
global property, determined by interactions between multiple components.
This makes the task of modularly checking this property difficult.

This paper presents \codename{}, a typed extension of \anvil{} for modular
verification of \termref{def:comm-safety}{communication safety} in synchronous hardware.
\codename{} types describe \termref{def:message-passing}{message-passing}, control-flow, and clock-cycle
behaviour.
Each component is checked against its type.
Compatible types then guarantee that the corresponding components compose into
a communication-safe system.

\codename{} also supports incremental design.
It can check a subsystem before the full system context is available.
The check produces a \termref{def:summary-type}{summary type} that captures the communication behaviour left
for future compositions.
To reduce annotation burden, \codename{} infers types by abstracting payload
values while preserving communication, timing, control, and \termref{def:fairness}{fairness}.
These inferred types are optimal for the checker: whenever some
checker-accepted types can prove \termref{def:compatibility}{compatibility}, the inferred types can prove
it as well.
\termref{def:type-inference}{Type inference} and \termref{def:type-checking}{type checking} are decidable.
\termref{def:compatibility}{Compatibility} checking is decided by a sound and complete algorithm.

We implement \codename{} by extending the \anvil{} compiler and evaluate it
along five dimensions: detecting communication errors during incremental design,
expressiveness on ten clock-cycle-accurate translations of open-source hardware
IPs, \termref{def:compatibility}{compatibility}-checking time, generated
hardware overhead, and annotation burden.
The results show that \codename{} expresses realistic protocols and verifies
all ten benchmarks in reasonable time.

\end{abstract}

%% file: sections/intro.tex
\section{Introduction}
\seclabel{intro}

Type systems in modern programming languages are effective because they move some verification obligations into the language itself.
Instead of checking every property after a program is built, the programmer records enough intent for the compiler to rule out whole classes of mistakes during development.
Hardware design has not benefited from this style of reasoning to the same extent.
Conventional \term{def:hdl}{hardware description languages (HDLs)}~\cite{18002017IEEEStandard2018,VHDLRef}, such as SystemVerilog and VHDL,
expose the low-level \term{def:rtl}{register-transfer language (RTL)} abstractions needed for efficient circuits.
Those abstractions make hardware expressive, but they leave many correctness properties to post-hoc verification,
where model checking and pre-silicon verification often face state-space explosion.

Recent safe HDLs~\cite{filamentHDL,christensenWireSortsLanguage2021,anvilhdl2026} reduce this burden by checking some properties at the language level.
They retain RTL concepts such as clock cycles, registers, and wires.
They also add abstractions and type systems that verify structural properties, including absence of combinational loops, pipeline safety, and timing safety.
These structural checks are useful, but they do not cover the behavioural properties that arise when hardware components \termref{def:message-passing}{communicate}.

Such behavioural properties matter because hardware designs are reactive systems.
Their components do not merely compute values; they repeatedly exchange values in particular clock cycles.
If two components disagree about when an exchange happens, the design can read the wrong value or wait forever.
Such instances arise in widely used open-source hardware IPs as protocol bugs
and documented interface constraints~\cite{pulpino182,pulpaxi-changelog,litex572,axisfifo-src,axidemux-src,axixbar-src}.

We focus on two such properties.
The first is \term{def:sync-safety}{synchronization safety}, which requires each sender and receiver to be mutually synchronized before they exchange a value.
The second is \term{def:liveness}{liveness}, which requires every component to eventually make progress.
Together, they define \term{def:comm-safety}{communication safety}.

\myparagraph{Challenge}
\termref{def:comm-safety}{Communication safety} is difficult to verify modularly.
Hardware systems are built from components, and those components are often designed independently.
The property we want, however, is not purely local to each pair of interacting components.
The dining philosophers problem is the standard example: each philosopher--fork pair can appear compatible,
while the composed system can still deadlock~\cite{dining-philosophers}.
The deadlock appears only when all pairs are composed, because the composed system may contain a cyclic wait.

Timing makes this modularity problem sharper in hardware.
A software protocol type may record that one process sends and another receives.
A hardware protocol type must also record when the \termref{def:message-send}{send} and \termref{def:message-receive}{receive} occur.
One cycle of latency can be the difference between a correct exchange and a synchronization bug.
Thus, a useful \termref{def:hdl}{HDL} abstraction for \termref{def:comm-safety}{communication safety} must describe both interaction structure and clock-cycle behaviour.
It must also support checking components as they are composed.

Formal verification for hardware spans post-hoc analysis of RTL state spaces and theorem proving over abstract hardware models.
This paper takes a complementary language-design approach.
We present \codename{}, an \termref{def:hdl}{HDL} that preserves \termref{def:rtl}{RTL} abstractions while adding constructs and types for communication-safe hardware.
\codename{} preserves these abstractions by building on Anvil~\cite{anvilhdl2026}, a recent HDL that lowers to SystemVerilog RTL.
\anvil{} also makes communication between hardware components explicit through \termref{def:message-passing}{message-passing} constructs.
This message-passing structure gives \codename{} a place to attach behavioural types.
A component type in \codename{} describes the component's communication behaviour,
including clock-cycle delays and complex control-flow patterns.

Behavioural type systems, such as session types~\cite{sessiontypes}, describe
communication through types. \codename{} uses these types in a bottom-up
workflow: each component implementation is first checked against its own type.
These component types are then checked for \term{def:compatibility}{compatibility}.
If the implementations satisfy their types and the types are compatible, the composed system is \termref{def:comm-safety}{communication safe}.
This is the \term{def:closed}{closed setting}, where all components are present
and their types are checked together.

Hardware designs, however, are often assembled before every component is present.
To support this incremental workflow, \codename{} also works in an \term{def:open}{open setting},
where some components are known and the rest of the environment is not yet specified.
In this setting, compatibility asks whether some boundary environment can make
the available components \termref{def:comm-safety}{communication safe}, under every resolution of their
internal choices (\defref{compatible-processes-main}).
The result is a \term{def:summary-type}{summary type}, which describes the
observable behaviour at that boundary and can be reused when adding components.
\codename{} checks compatibility in both closed and open settings using a sound and complete algorithm,
based on an automata-theoretic characterization of types.

For usability, \codename{} automatically infers a type from each process;
user-provided protocol types are optional.
Inference erases payload values and expressions, giving a
\term{def:data-oblivious-abstraction}{data-oblivious abstraction}.
These inferred types are optimal for the checker: if any types accepted by the
checker can establish \termref{def:compatibility}{compatibility} for the given
processes, then the inferred types can do so as well.
If inferred types are incompatible, choosing different checker-accepted types
cannot make the check succeed. The designer must change the processes or their
composition, or use value-sensitive reasoning beyond this checker.

Finally, we implement \codename{} as an extension of \anvil{} and evaluate it
by creating cycle-accurate ports of open-source hardware IP cores in \codename{}.
These translations demonstrate \codename{}'s expressiveness in capturing complex communication patterns and verifying their \termref{def:comm-safety}{communication safety}.

\myparagraph{Contributions}
The contribution of this paper is \codename{}, a typed \termref{def:hdl}{HDL} for specifying and verifying
\termref{def:comm-safety}{communication safety} in synchronous hardware.
We give the syntax and semantics of \codename{}, including the clock-cycle behaviour needed to define \termref{def:comm-safety}{communication safety} for process compositions.
We then define \termref{def:type}{types}, type satisfaction, and type \termref{def:compatibility}{compatibility}.
These definitions yield the main soundness theorem: if processes satisfy compatible types, then their composition is \termref{def:comm-safety}{communication safe}.
We give a sound and complete automata-based algorithm for checking \termref{def:compatibility}{compatibility} in both \termref{def:closed}{closed} and \termref{def:open}{open} settings, and an optimal \termref{def:type-inference}{type-inference} algorithm under a \termref{def:data-oblivious-abstraction}{data-oblivious abstraction}.
Finally, we implement \codename{} as an extension of \anvil{} and evaluate its expressiveness and practicality on real-world hardware IP cores.
Our results have also been mechanized in Lean 4, and the accompanying artefacts are open-source (see the Data Availability Statement).

\myparagraph{Outline}
The rest of the paper is organized as follows.
Section~\ref{sec:overview} gives a brief overview of \codename{}'s motivation and design through examples.
Section~\ref{sec:pact} presents the language syntax and semantics in detail.
Section~\ref{sec:types} discusses the type system.
Section~\ref{sec:evaluation} presents our empirical evaluation.
Section~\ref{sec:related} discusses related work, and Section~\ref{sec:conclusions} concludes with the future work.

%% file: sections/overview.tex
\section{Overview of \codename{}}
\seclabel{overview}

In this section, we first walk the readers through a
simple example illustrating \termref{def:comm-safety}{communication safety} issues in hardware design.
We then discuss our language \codename{} and present the core ideas
behind the design of its type system catered for development of
\termref{def:comm-safety}{communication safe} \termref{def:rtl}{RTL} designs.

\input{sections/problem}
\input{sections/overview-anvil-abstractions}
\input{sections/overview-new-constructs}
\input{sections/overview-types}
\input{sections/overview-type-compatibility}
\input{sections/overview-incremental}
\input{sections/overview-optimality-dataoblivious}

%% file: sections/problem.tex
\subsection{The Problem}
\seclabel{problem}

\begin{figure}[t]
\centering
\input{assets/motivating-example}
\caption{Example of two worker nodes contending for two arbiters.
 Each worker connects to both arbiters. On each connection,
\protect\prog{req} and \protect\prog{done} flow from worker to arbiter; \protect\prog{gnt} flows
from arbiter to worker.}
\figlabel{motivatingex}
\end{figure}

Let us illustrate what we mean by \termref{def:comm-safety}{communication safety} through \figref{motivatingex}
which shows a simple \termref{def:rtl}{RTL}-style implementation of
a hardware design with two worker modules and two arbiter modules.
Informally, each arbiter manages some resource and each worker is supposed
to access a resource by liaising with the corresponding arbiter.
The design encodes a coordination protocol that is supposed to ensure that
at any given moment in time, an arbiter grants access to the resource it is managing
to at most one of the two workers.
All four components execute synchronously and share a single global clock.

The clocked pseudocode in \figref{motivatingex} represents the state-machine
structure used in \termref{def:hdl}{HDLs} such as SystemVerilog.
Each worker is implemented as a state machine.
The register \prog{state} records the current control state of the worker, while the register \prog{cyc} counts the number of elapsed clock cycles.
The figure uses clocked pseudocode: the block selected by the current
\prog{state} determines the outputs for that cycle, and assignments to
\prog{state} take effect at the next clock edge.
The counter \prog{cyc} increments once at every
clock edge. All \prog{req} and \prog{done} outputs default to zero in each
cycle unless assigned by the active block; the worker drives these signals,
and the arbiter drives \prog{gnt}.
The workers \termref{def:message-passing}{communicate} with the arbiters through wires.
Each worker-arbiter interface consists of three signals: \prog{req}, \prog{gnt}, and \prog{done}.
The signal \prog{req} is asserted by a worker to request access, \prog{gnt} is asserted by an arbiter to grant access, and \prog{done} is asserted by a worker to release access.

Now consider \prog{Worker 1}. It first asserts \prog{req} to
\prog{Arbiter 1} and then waits for \prog{Arbiter 1} to assert \prog{gnt}.
Once \prog{Worker 1} receives this grant, it makes a local choice based on the value of \prog{cyc}.
If the cycle count is not a multiple of four, the worker
finishes the transaction immediately: it asserts \prog{done} to \prog{Arbiter 1} and returns to its initial state.
If the cycle count is a multiple of four, the worker keeps the first transaction open and begins a second transaction with \prog{Arbiter 2}.
It asserts \prog{req} to \prog{Arbiter 2}, waits for \prog{Arbiter 2} to assert \prog{gnt}, and only then asserts \prog{done} to both arbiters.

\prog{Worker 2} implements a symmetric behaviour. It first initiates a transaction with \prog{Arbiter 2}. After receiving a grant from \prog{Arbiter 2}, it makes a local choice based on the value of \prog{cyc}. Depending on this choice, it either completes the transaction immediately or keeps the transaction open while initiating a second transaction with \prog{Arbiter 1}.
\prog{Worker 2} uses multiples of three rather than four for this choice; the workers
share the same protocol structure, but need not make the same local choice in
each cycle.

This behaviour implicitly encodes several assumptions about the behaviour of the arbiters.
Some assumptions describe the basic ordering of messages: a \prog{req} is followed by a \prog{gnt}, and the transaction is eventually closed by a \prog{done}.
Other assumptions describe the timing of these messages. For instance, a \prog{req} and its corresponding \prog{gnt} are separated by at least one clock cycle. In contrast, \prog{done} may be asserted either in the same cycle as the grant or in a later cycle, depending on whether the worker starts a second transaction.

The ordering and timing of these messages determine whether the components
can communicate safely. However, the component interfaces do not express
these requirements, and the HDLs provide little support for checking them.

In this example, even if each arbiter respects these local assumptions,
the overall system obtained by the composition of the 4 components may still fail to make progress.
Because the workers behave symmetrically, the system can eventually reach a state where \prog{Worker 1} is waiting for \prog{Arbiter 2}, while \prog{Worker 2} is waiting for \prog{Arbiter 1}.
Each worker still holds its first resource because it has not sent \prog{done}.
Each arbiter may therefore wait for that transaction to finish before granting
access to the other worker.
In this state, both workers are waiting for grants that can never be produced, and the system is deadlocked.

As this example illustrates, existing \termref{def:hdl}{HDLs}
make it difficult to reason about the communication behaviour of a system.
This is because such HDLs do not allow hardware designers to succinctly express the
relevant protocol directly in the component interface.
Instead, it has to be encoded indirectly in low-level state machines,
making \termref{def:comm-safety}{communication safety} difficult to specify and check.

%% file: assets/motivating-example.tex
\begin{minipage}[t]{0.48\linewidth}
\textbf{Worker 1}
\begin{lstlisting}[basicstyle=\ttfamily\scriptsize,columns=fullflexible,keepspaces=true]
reg cyc : logic[32]
reg state : logic[3] = SEND_REQ1
// At each clock edge: cyc := cyc + 1
// Each cycle: req = done = 0 by default
state: SEND_REQ1
    arb1.req = 1
    state = WAIT_GRANT1
state: WAIT_GRANT1
    if (arb1.gnt == 1) {
        state = (cyc%4==0) ?
                SEND_REQ2 : SEND_REQ1
        arb1.done = (cyc%4==0) ? 0 : 1
    } else arb1.req = 1
state: SEND_REQ2
    arb2.req = 1
    state = WAIT_GRANT2
state: WAIT_GRANT2
    if (arb2.gnt == 1) state = DONE
    else arb2.req = 1
state: DONE
    arb1.done = 1
    arb2.done = 1
    state = SEND_REQ1
\end{lstlisting}
\end{minipage}\hfill
\begin{minipage}[t]{0.48\linewidth}
\textbf{Worker 2}
\begin{lstlisting}[basicstyle=\ttfamily\scriptsize,columns=fullflexible,keepspaces=true]
reg cyc : logic[32]
reg state : logic[3] = SEND_REQ2
// At each clock edge: cyc := cyc + 1
// Each cycle: req = done = 0 by default
state: SEND_REQ2
    arb2.req = 1
    state = WAIT_GRANT2
state: WAIT_GRANT2
    if (arb2.gnt == 1) {
        state = (cyc%3==0) ?
                SEND_REQ1 : SEND_REQ2
        arb2.done = (cyc%3==0) ? 0 : 1
    } else arb2.req = 1
state: SEND_REQ1
    arb1.req = 1
    state = WAIT_GRANT1
state: WAIT_GRANT1
    if (arb1.gnt == 1) state = DONE
    else arb1.req = 1
state: DONE
    arb2.done = 1
    arb1.done = 1
    state = SEND_REQ2
\end{lstlisting}
\end{minipage}
\smallskip

%% file: sections/overview-anvil-abstractions.tex
\subsection{Background: \anvil{} Abstractions}
\seclabel{overview-anvil-abstractions}

\codename{} builds on top of \anvil~\cite{anvilhdl2026}, a recently introduced HDL
that has taken a step towards capturing intended communication behaviour explicitly.
\anvil{} adopts the \term{def:message-passing}{message-passing} paradigm: each hardware
component is modeled as a \term{def:process}{process}, and processes communicate by exchanging
\termref{def:message}{messages} over channels.
Its type system ensures the safety of values exchanged between components, which
\anvil{}'s authors call \emph{timing safety}. For example, when a process sends a
reference to its internal state, \anvil{} ensures that the state is not mutated
while the receiving process may still access it. It also ensures that the receiver
does not use the value beyond the lifetime declared in the process interface.

Communication in \anvil{} is limited to blocking message exchanges: both
\prog{send} and \prog{recv} wait until the corresponding message is exchanged.
This prevents synchronization errors, but does not rule out liveness errors.
For example, a process that sends \prog{m1} followed by \prog{m2} can deadlock
with a process that receives \prog{m2} followed by \prog{m1}.
Both processes wait for the other to perform its second communication, and
\anvil{}'s type system does not reject this composition.

\myparagraph{Extending \anvil{}}
\codename{} extends \anvil{} in two ways. First, it generalizes the communication
model to include instantaneous message exchanges and two offer constructs:
\prog{offer}, which allows multiple enabled communications to proceed, and
\prog{xoffer}, which selects one enabled alternative. We explain these constructs
in detail later. These extensions retain \anvil{}'s blocking communication while expressing
a broader class of communication behaviours.

Second, \codename{} introduces a type system for
\termref{def:comm-safety}{communication safety}. With the more expressive
communication model, untyped \codename{} programs can exhibit both synchronization
and liveness errors. Our type system rules out both: processes that conform to
compatible types are guaranteed to be \termref{def:comm-safety}{communication safe} when composed.
This guarantee is orthogonal to \anvil{}'s timing safety, which concerns the
validity of the values exchanged. The implementation of \codename{} combines
both type systems to provide both guarantees.

To illustrate the extended communication model,
\figref{overview-processes}(\subref{fig:overview-worker}) shows Worker~1
from the preceding example rewritten in \codename{}.

\input{assets/fig-overview-processes}

The channel \prog{waij} connects Worker~$i$ to Arbiter~$j$. Thus Worker~1
uses \prog{wa11} and \prog{wa12}, while Arbiter~1 uses \prog{wa11} and
\prog{wa21}. This process makes the communication structure explicit. The messages
\prog{wa11.req}, \prog{wa11.done}, and their \prog{wa12} counterparts name the
worker's interactions with the two arbiters. A \term{def:blocking-send}{\prog{send}} is blocking: the
process continues only when the matching communication is available. An
\term{def:immediate-send}{\prog{isend}} is immediate: it must be exchanged in the current clock cycle.
The operator \prog{>>} sequences behaviours, while \prog{;} composes
behaviours in parallel.

At the RTL level, each message exchange is lowered to a payload signal and,
when needed, handshake signals. Here the blocking \prog{send wa11.req(1)}
captures the request/grant handshake: the worker requests access and continues
when the arbiter accepts the request. The grant is therefore part of the
handshake, rather than a separate message in the process code. The immediate
\prog{done} sends describe the release signals. More generally, a send drives
the payload and indicates its availability; the matching receive reads the
payload when the handshake occurs. Offers inspect availability before
committing and use two-way handshakes, such as ready-valid or valid-ack.

\myparagraph{More Constructs Needed}
However, sends and receives alone still fall short of structurally expressing
behaviour important for reasoning about communication protocols. For example,
the worker process above describes one side of the protocol. The arbiter
describes the other side: it must observe which requests are available and then
grant one of them. This pattern is awkward to express using only blocking or
immediate sends and receives.

%% file: assets/fig-overview-processes.tex
\begin{figure}[t]
\centering
\lstset{aboveskip=2pt,belowskip=2pt,
  xleftmargin=2pt,xrightmargin=2pt,
  framexleftmargin=0pt,framexrightmargin=0pt}
\captionsetup{skip=4pt}
\captionsetup[subfigure]{skip=2pt}
\begin{subfigure}[t]{0.52\linewidth}
\begin{lstlisting}
reg cyc : logic[32] // cycle counter
loop {
    send wa11.req(1) >> // blocking send
    if (*cyc % 4 != 0) then {
        isend wa11.done(1) // immediate send
    }
    else {
        send wa12.req(1) >> // blocking send
        cycle 1 >>
        isend wa12.done(1); // immediate send
        isend wa11.done(1) // immediate send
    } >>
    cycle 1
}
loop {
    set cyc := *cyc + 1 // increment cycle counter
}
\end{lstlisting}
\caption{Worker~1}
\label{fig:overview-worker}
\end{subfigure}\hfill
\begin{subfigure}[t]{0.46\linewidth}
\begin{lstlisting}
reg prio : logic
loop {
    xoffer *prio { // current priority
        _ = recv wa11.req : {
            _ = irecv wa11.done >>
            set prio := 1 // switch priority
        }
        _ = recv wa21.req : {
            _ = irecv wa21.done >>
            set prio := 0 // switch priority
        }
    }
    else {
        ()
    } >>
    cycle 1
}
\end{lstlisting}
\caption{Arbiter~1}
\label{fig:overview-arbiter}
\end{subfigure}
\caption{Worker and arbiter processes in \codename{}.}
\figlabel{overview-processes}
\end{figure}

%% file: sections/overview-new-constructs.tex
\subsection{New Language Constructs in \codename{}}
\seclabel{overview-new-constructss}

\codename{} addresses this limitation of \anvil{} by extending it with
two \term{def:offer}{offer constructs}:
\prog{offer B else P} and \prog{xoffer B else P}. A \term{def:conjunctive-offer}{conjunctive offer}
\prog{offer} lets every enabled branch in $B$ progress. An \term{def:exclusive-offer}{exclusive offer}
\prog{xoffer} may observe several enabled branches, but commits to one of
them. Semantically, this choice is nondeterministic; in the current
implementation, it is resolved by a designer-programmed priority. Consider
Arbiter~1 in \figref{overview-processes}(\subref{fig:overview-arbiter}),
which receives requests on \prog{wa11} and \prog{wa21}.

Here \prog{*prio} reads the register that selects the branch priority.
In each branch, \prog{\_ = recv} discards the received payload, and \prog{:}
separates the communication head from its continuation.
This arbiter checks whether either worker is requesting access. If a request is
available, it grants one worker and then expects that worker's \prog{done}
message immediately. After completing the selected branch, or finding no
request, it waits one cycle before checking again. Arbiter~2 has the same
structure with channels \prog{wa12} and \prog{wa22}.
This arbiter requires \prog{done} in the grant cycle, while the worker may
delay it. We examine this mismatch in \secref{overview-incremental}.

%% file: sections/overview-types.tex
\subsection{\codename{} Types}
\seclabel{overview-types}

For \termref{def:comm-safety}{communication safety}, the values carried by \termref{def:message}{messages} are not the important
aspect. What matters is the protocol: which messages may be exchanged, how
control flows between those exchanges, and which timing constraints those
exchanges impose.

Types in \codename{} make this protocol explicit. The full type syntax appears
later in Section~\ref{sec:types}; in this overview, we use a desugared
transition-graph view. This view presents a type as a small transition system
that the \termref{def:compatibility}{compatibility} check can analyze for \termref{def:comm-safety}{communication safety}.

\begin{figure}[t]
  \centering
  \captionsetup[subfigure]{skip=2pt}
  \begin{subfigure}[b]{0.60\linewidth}
    \centering
    \resizebox{\linewidth}{!}{\input{assets/fig-overview-worker-type}}
    \caption{Worker~1}
    \label{fig:workertype}
  \end{subfigure}
  \hfill
  \begin{subfigure}[b]{0.37\linewidth}
    \centering
    \resizebox{\linewidth}{!}{\input{assets/fig-overview-arbiter-type}}
    \caption{Arbiter~1}
    \label{fig:arbitertype}
  \end{subfigure}
  \caption{Worker and arbiter types. $\#$ is a clock tick,
  $\$$ a progress marker, and $\theta$ an internal choice.
  The grouped releases occur in the same cycle.}
  \label{fig:overview-types}
\end{figure}

Consider the types of Worker~1 and Arbiter~1, shown in
Fig.~\ref{fig:overview-types}. The worker type first \termref{def:message-send}{sends}
\prog{wa11.req!}. Since this send is blocking, the type can also wait when the
matching \prog{wa11.req?} action is not available in the current cycle. After
that request, the worker chooses between two branches. One sends
\prog{wa11.done!} immediately. The other sends \prog{wa12.req!}, waits one
cycle ($\#$), and then sends both \prog{wa12.done!} and \prog{wa11.done!}.
The type abstracts away the
branch condition while preserving the communication behaviours that either
branch may perform.

The arbiter type gives the complementary view. It first receives whichever
request is available, such as \prog{wa11.req?} or \prog{wa21.req?}. If both
requests are available, it may choose one non-deterministically. After choosing
a request, it expects the corresponding \prog{done} message immediately.

\term{def:progress-marker}{Progress markers} $\$_{w_1}$ and $\$_{arb_1}$ record when the worker and
arbiter make progress. These markers are used to state \termref{def:liveness}{liveness}. In this
example, liveness requires the clock to tick infinitely often in every
infinite execution, and both markers to appear infinitely often in every fair
infinite execution.

%% file: assets/fig-overview-worker-type.tex
\begin{tikzpicture}[x=1pt,y=1pt,>=stealth,line width=.55pt,
  every node/.style={font=\fontsize{8}{9}\selectfont},
  state/.style={circle,draw=black,minimum size=8pt,inner sep=0pt},
  lab/.style={text=vtypecompat,inner sep=1pt,fill=white}]
\node[state,double,double distance=.8pt] (w0) at (0,0) {};
\node[state] (w1) at (52,0) {};
\node[state] (w2) at (104,0) {};
\node[state] (w3) at (156,0) {};
\node[state] (w4) at (156,-54) {};
\node[state] (w5) at (104,-54) {};
\node[state] (w6) at (52,-54) {};
\node[state] (w7) at (52,-29) {};
\draw[->] (w0) -- node[lab,above] {\texttt{wa11.req!}} (w1);
\draw[->] (w1) -- node[lab,above] {$\theta$} (w2);
\draw[->] (w2) -- node[lab,above] {\texttt{wa12.req!}} (w3);
\draw[->] (w3) -- node[lab,right] {$\#$} (w4);
\draw[->] (w4) -- node[lab,below,align=center]
  {\texttt{wa12.done!}\\\texttt{wa11.done!}} (w5);
\draw[->] (w5) -- node[lab,below] {$\#$} (w6);
\draw[->,rounded corners=3pt] (w6) -- (-31,-54) -- (-31,0) -- (w0);
\node[lab,left] at (-31,-27) {$\$_{w_1}$};
\draw[->] (w1) -- node[lab,right] {$\theta$} (w7);
\draw[->] (w7) -- (w5);
\node[lab] at (78,-34) {\texttt{wa11.done!}};
\draw[->] (w0) to[out=115,in=65,looseness=7] (w0);
\node[lab,above] at (0,15) {$\#\;\neg(\texttt{wa11.req?})$};
\draw[->] (w2) to[out=115,in=65,looseness=7] (w2);
\node[lab,above] at (104,15) {$\#\;\neg(\texttt{wa12.req?})$};
\end{tikzpicture}

%% file: assets/fig-overview-arbiter-type.tex
\begin{tikzpicture}[x=1pt,y=1pt,>=stealth,line width=.55pt,
  every node/.style={font=\fontsize{8}{9}\selectfont},
  state/.style={circle,draw=black,minimum size=8pt,inner sep=0pt},
  lab/.style={text=vtypecompat,inner sep=1pt,fill=white}]
\node[state] (a0) at (0,0) {};
\node[state] (a1) at (-37,-36) {};
\node[state] (a2) at (37,-36) {};
\node[state] (a3) at (0,-72) {};
\node[state] (a4) at (0,-104) {};
\draw[->] (0,14) -- (a0);
\draw[->] (a0) -- (a1);
\node[lab] at (-28,-14) {\texttt{wa11.req?}};
\draw[->] (a0) -- (a2);
\node[lab] at (28,-14) {\texttt{wa21.req?}};
\draw[->] (a1) -- (a3);
\node[lab] at (-28,-59) {\texttt{wa11.done?}};
\draw[->] (a2) -- (a3);
\node[lab] at (28,-59) {\texttt{wa21.done?}};
\draw[->] (a0) -- node[lab,left] {\texttt{else}} (a3);
\draw[->] (a3) -- node[lab,right] {$\#$} (a4);
\draw[->,rounded corners=3pt] (a4) -- (51,-104) -- (51,0) -- (a0);
\node[lab,right] at (51,-78) {$\$_{arb_1}$};
\end{tikzpicture}

%% file: sections/overview-type-compatibility.tex
\subsection{Type Compatibility}
\seclabel{overview-type-compatibility}

Recall that the two workers have symmetric behaviour and therefore symmetric
types, and the two arbiters have the same type. Once all four components have
types, we have a \termref{def:closed}{closed system} at the type level. \termref{def:comm-safety}{Communication safety} can then
be checked over the types rather than over the full \termref{def:rtl}{RTL} state space. This
reduction is justified by \term{def:type-soundness}{type soundness}: the types faithfully abstract process
behaviour. Thus, if the types are compatible, every set of processes that
satisfies those types is \termref{def:comm-safety}{communication safe} as well.

For a \termref{def:closed}{closed system}, this means composing the types of Worker~1, Worker~2,
Arbiter~1, and Arbiter~2 and checking the resulting automaton for
\termref{def:sync-safety}{synchronization safety} and \termref{def:liveness}{liveness}. Synchronization safety rules out stuck
states. Liveness is encoded as an \term{def:acceptance-condition}{acceptance condition} on infinite runs: all
infinite runs must contain infinitely many clock ticks, and all four progress
markers must appear infinitely often on \termref{def:fairness}{fair} runs.

\nparagraph{Fairness}
\term{def:fairness}{Fairness} assumptions are often necessary to establish that a system is \termref{def:liveness}{live}. Without such assumptions, the checker may reject a system because of an execution that indefinitely ignores an enabled branch.
To express these assumptions, \codename{} provides fair-tagged constructs and
corresponding type constructs. For example, a fair-tagged if-else requires both
branches to recur infinitely often if either does. The acceptance condition
checks progress on executions that satisfy these annotations.

The designer specifies these fairness assumptions explicitly. A designer can
initially write a system without fairness annotations and run the compatibility
check. If the check fails, the checker produces a counterexample trace. The
designer can then inspect the trace to determine whether the failure arises
from an execution that should be excluded by a fairness assumption. If so, the
designer can annotate the relevant constructs as fair and rerun the check.

%% file: sections/overview-incremental.tex
\subsection{Incremental Checking}
\seclabel{overview-incremental}

\termref{def:closed}{Closed-system} checking requires types for every component of the system.
Hardware designers, however, often build systems incrementally. They may want
to check whether two components compose before the rest of the system is
complete. This motivates the \termref{def:open}{open setting}: given some components and a shared
interface, we ask whether some \term{def:boundary-environment}{boundary environment} makes their shared
interaction safe.

The key idea is to define both process and type semantics
compositionally. A type takes a step under a \termref{def:one-step-environment}{one-step environment}. That
environment records which complementary actions are available now: receives for
sent messages, and sends for received messages. This local view lets us compose
part of the system while treating the remaining actions as boundary actions.

For example, consider the open composition of Worker~1 and Arbiter~1. The
shared interface contains the messages between them, such as
\prog{wa11.req} and \prog{wa11.done}. The boundary contains the remaining
worker-arbiter messages, such as \prog{wa12.req}, \prog{wa12.done},
\prog{wa21.req}, and \prog{wa21.done}.

This \termref{def:open}{open composition} is not compatible. To see why, suppose the boundary
environment does not provide \prog{wa21.req}. Then the arbiter can grant
Worker~1's request. After that grant, Worker~1 may either send \prog{done}
immediately or delay it. In the delayed branch, the arbiter expects
\prog{done} in the current cycle, so the composition cannot complete that
cycle. Providing
\prog{wa21.req} at the boundary does not remove the problem. The arbiter can
still choose to grant only Worker~1's request, leading to the same delayed
\prog{done} execution. Thus no boundary strategy can ensure both clock
progress and the absence of stuck executions for this composition.

This is the quantifier pattern behind \termref{def:compatibility}{compatibility}. We ask whether there exists an
environment that makes the system safe against all resolutions of the system's
nondeterminism. This alternation is naturally represented as a two-player game
on an automaton~\cite{automatagames}. In our game,
the environment is Player~0: it provides boundary actions that may help the
types win. The system's nondeterminism is Player~1: it resolves the internal
nondeterminism in ways that may violate
\termref{def:comm-safety}{communication safety}. \termref{def:compatibility}{Compatibility} therefore
reduces to determining whether Player~0 has a winning strategy against all
Player~1 strategies. We give the detailed compatibility algorithm in
Section~\ref{sec:typecompat}.

%% file: sections/overview-optimality-dataoblivious.tex
\subsection{Optimal Type Inference}
\seclabel{overview-optimality}

Finally, designers should not have to write every type by hand.
The compiler infers a type $\mathsf{abs}(P)$ from each process $P$
by erasing payload values and local data computations while preserving communication, timing, control
structure, and fair-choice markers.
\termref{def:compatibility}{Compatibility} is then checked on these inferred types.
No user-written behavioural type is required for this check.

The inferred types are optimal for the \termref{def:data-oblivious-type-checking}{data-oblivious checker}:
if some checker-accepted type assignment can establish
\termref{def:compatibility}{compatibility} for the given processes, then the
inferred types can do so as well. Thus choosing different checker-accepted
types cannot repair a failed compatibility check.
The checker does not use invariants about data values to rule out branches.
Such reasoning may prove compatibility beyond this guarantee, as illustrated
in \secref{type-optimality}.

\subsection{Design Workflow}
\seclabel{overview-design-workflow}

In summary, \codename{} supports multiple design workflows, depending on
whether types for the complete system or only part of it are available.
In either case, designers can provide component types as specifications of
the intended protocols or infer them from existing implementations. When an
implementation becomes available for a supplied type, it must be checked
against that type.

\myparagraph{Closed Setting}
When designing a system as a whole, designers can first specify the types of
all components and check whether they are compatible. If they are, designers
can proceed to implement the components against their respective types.
Alternatively, if the implementations already exist, designers can infer
their types and check them for compatibility.

\myparagraph{Open Setting}
When a system is developed incrementally, designers can first check whether
the available component types are compatible while the rest of the system
remains unspecified. If they are, the compatible composition can be
represented by a \termref{def:summary-type}{summary type}
(\secref{summary-types}), which captures its boundary behaviour. This summary
can then be used to check compatibility with additional components as their
types become available, without reconstructing the existing composition
from its individual components.

%% file: sections/pact.tex
\section{Syntax, Semantics and Safety of \codename{} Designs}
\seclabel{pact}

\codename{} models hardware components as \termref{def:process}{processes} that
share a global clock and communicate by passing \termref{def:message}{messages}.
This operational model follows the base language \anvil~\cite{anvilhdl2026}, on which
we build.
A \term{def:message}{message} is a labelled value exchanged between processes.
Each pair of processes exchanges a set of messages.
These messages can be derived from the definitions of the processes.
At the RTL level, messages are refined into bundles of wires carrying data and,
when necessary, handshake signals for synchronisation.
As a result, the core language exposes the timing and \termref{def:message-passing}{communication} structure of a design while abstracting away datapath details.
In the rest of this section, we present the abstract syntax, semantics, and
safety properties of \codename{}.

\input{sections/pact-syntax}
\input{sections/pact-semantics}
\input{sections/pact-safety}

%% file: sections/pact-syntax.tex
\subsection{Syntax of \codename{}}
\seclabel{pact-syntax}

\begin{figure}[t]
\centering
\setlength{\abovedisplayskip}{0pt}
\setlength{\belowdisplayskip}{0pt}
\[
\begin{aligned}
e ::= {} & x \mid v \mid e\ \texttt{bop}\ e
\qquad
v ::= \pkw{true} \mid \pkw{false}
\qquad B ::= \emptyset \mid B, a : P \\
a ::= {} & \pkw{send}\ m\ e \mid x = \pkw{recv}\ m
\\[0.4em]
P ::= {} &
\pkw{skip} \mid a\ \pkw{in}\ P \mid X \mid \pkw{cycle} \mid \pkw{end}_t
\mid P \pparop_I P \mid P \pseq  P \mid \pkw{let}\ x = e\ \pkw{in}\ P \\
& \mid \pkw{if}\ e\ \pkw{then}\ P\ \pkw{else}\ P
\mid \pkw{fair-if}_\eta\ e\ \pkw{then}\ P\ \pkw{else}\ P \\
& \mid \pkw{offer}\ B\ \pkw{else}\ P
\mid \pkw{xoffer}\ B\ \pkw{else}\ P
\mid \pmu X . P
\end{aligned}
\]
	\caption{Abstract syntax of processes in \codename{}. Here $m$ ranges over
	message labels, $I$ over sets of message labels, $t$ over progress
	identifiers, $\eta$ over fair-choice identifiers, and \texttt{bop} over
	binary operators.}
\label{fig:abstract-syntax}
\end{figure}

Fig.~\ref{fig:abstract-syntax} presents the abstract syntax of \codename{} \termref{def:process}{processes}, distinguishing expressions $e$, branch sets $B$, branch actions $a$, and processes $P$.
We omit datapath constructs, such as register assignments, because the safety property depends only on timing and communication behaviour.
The full syntax also includes a message alphabet specifying the messages a process can exchange with other components.
For simplicity, we omit this alphabet from the main text and defer its formal treatment to Appendix~\ref{app:signed-process-syntax}.
The core syntax also uses \pkw{send} and \pkw{recv} for the instantaneous
\prog{isend} and \prog{irecv} of the overview. Blocking communication is
encoded using offers, as shown below. Similarly, sequencing $P\pseq P$
corresponds to \prog{>>} in the overview, and parallel composition
$P\pparop_I P$ corresponds to \prog{;}.

\myparagraph{Core structure}{
The process $\pkw{cycle}$ advances the clock,
while $\pkw{skip}$ terminates the
current sub-process. Sequencing, let-binding, and if-else are standard.
The parallel composition operator $\pparop_I$ composes two processes
synchronizing on the messages in the set $I$. This operator models both
concurrency inside a hardware component, where $I=\emptyset$, and parallel
composition of components in a system.
Recursion $\pmu X.P$ describes the repeating behaviour of hardware components.
Well-formed recursion requires every occurrence of $X$ to be bound. Because
hardware behaviour is clocked, each recursive cycle must also be guarded by
$\pkw{cycle}$. Finally, recursion is tail terminating: along every path, the
recursive call must be the last action in that path.

}

\myparagraph{Communication}{
A \term{def:message-prefix}{message prefix} $a\ \pkw{in}\ P$ performs a \term{def:message-send}{send action} or \term{def:message-receive}{receive action} and then
continues as $P$.
In the core language, these actions are \term{def:instantaneous-action}{instantaneous}: a process can send or receive a
message only if another process is ready to perform the complementary action in
the same clock cycle. They may wait for a complementary action to become
available within that cycle, but cannot defer the exchange across a clock tick.
They correspond to \prog{isend} and \prog{irecv} in the implementation language
used in the overview. This makes it unnatural to model hardware components that
may communicate with several processes or block on a message. To address this,
\codename{} introduces the \termref{def:offer}{offer constructs} $\pkw{offer}$ and $\pkw{xoffer}$.
The offer forms inspect a branch set $B$, where each branch head action $a$ is
a send or receive.
$\pkw{xoffer}$ (\termref{def:exclusive-offer}{exclusive offer}) commits to one enabled branch, and
$\pkw{offer}$ (\termref{def:conjunctive-offer}{conjunctive offer}) admits all enabled branches. Both allow the
else branch to be taken when no branch is enabled.
Offers also let us encode \termref{def:blocking-send}{blocking} send/receive behaviour. For example,
$
\pmu X.\ \pkw{xoffer}\ \{a:\pkw{skip}\}\ \pkw{else}\ (\pkw{cycle}\pseq X)$
waits across cycles until action $a$ is enabled, then performs it and terminates.
This is the blocking \prog{send}/\prog{recv} behaviour used in the overview.
Their precise semantics are
discussed in the next subsection.

}

\myparagraph{Safety markers}{%
\phantomsection\label{sec:safety-specifications}
The construct $\pkw{fair-if}_\eta$ marks fair choices with
identifier $\eta$.
The \termref{def:progress-marker}{marker} $\pkw{end}_{t}$ records progress with identifier $t$. In a
looping hardware process, placing $\pkw{end}_{t}$ at the end of the loop body
records completion of one logical iteration.
Here we include fair-tagged if-else in the core syntax. A fair-tagged exclusive
offer has an analogous specification, but we omit it from the core language.
These markers let designers make their intent explicit. For instance,
fair-tagging an if-else restricts \termref{def:liveness}{liveness} checking to executions that respect
that \termref{def:fairness}{fairness} assumption.
}

%% file: sections/pact-semantics.tex
\subsection{Semantics of \codename{}}
\seclabel{pact-semantics}

\input{assets/fig-semantics}

\begingroup
\raggedbottom
\setlength{\abovedisplayskip}{7pt}
\setlength{\belowdisplayskip}{7pt}
\setlength{\abovedisplayshortskip}{4pt}
\setlength{\belowdisplayshortskip}{6pt}

The process semantics of \codename{} is defined in a local, small-step style.
A judgment
\[
\lstep{E}{P}{\alpha}{P'}
\]
says that, under the \term{def:one-step-environment}{one-step environment} $E$, the process $P$ takes one
atomic step, produces the observation $\alpha$, and continues as $P'$.
The process $P'$ is the residual left after this step.
This local view lets us reason compositionally. It is especially useful for
\termref{def:open}{open systems}, where some components are left unspecified but
we still want to prove \termref{def:comm-safety}{communication safety}. \termref{def:closed}{Closed}
systems are the special case in which the only outside behaviour is the global
clock shared by all processes.

\myparagraph{Actions}
The label $\alpha$ records the set of actions performed by a process in a single local step. This set may contain multiple actions for parallel compositions and conjunctive offers. These actions include \termref{def:message-send}{message sends} ($m\psym{!}v$), \termref{def:message-receive}{message receives} ($m\psym{?}v$), clock ticks ($\#$), fair-branch outcomes, and \termref{def:progress-marker}{progress markers}. A resolved internal exchange of message $m$ is recorded as $\tick{m}$.

We use $\IntMsgs(P)$ to denote the names of internal messages and $\BoundaryMsgs(P)$ to denote the polarized actions visible outside $P$.
Internal messages arise from parallel composition $P_1\pparop_I P_2$, where $I$ denotes the set of internal messages.
The actions of $P$ come from the following set:
\[
\begin{aligned}
\Act_{P} &=
  \underbrace{\{\tick{m}\mid m\in\IntMsgs(P)\}}_{\text{resolved messages}}
  \cup
  \underbrace{\{m\psym{?}v\mid m\psym{?}\in\BoundaryMsgs(P)\}
  \cup\{m\psym{!}v\mid m\psym{!}\in\BoundaryMsgs(P)\}}_{\text{boundary messages}}\\
&\quad \cup
  \underbrace{\{\eta_{\mathsf{true}},\eta_{\mathsf{false}}
  \mid\eta\in\EtaSet{P}\}}_{\text{fair outcomes}}
  \cup
  \underbrace{\{\$_t\mid t\in\PidSet{P}\}}_{\text{progress identifiers}} .
\end{aligned}
\]

Here $\EtaSet{P}$ is the set of fair-choice identifiers $\eta$ occurring in
$P$, and $\PidSet{P}$ is the set of its progress identifiers $t$.
The set of possible labels is $\ActSet{P}=2^{\Act_{P}}\cup\{\{\#\}\}$,
where $\{\#\}$ records a clock tick.
The empty label $\epsilon=\emptyset$ records a silent step. We will omit braces
around singleton labels.

\myparagraph{Environment}
The \termref{def:one-step-environment}{one-step environment} $E$ describes what the processes around $P$ can offer
during the current clock cycle. We will use $E=(\Imm(E),\Fut(E))$ for this environment. The \term{def:immediate-environment}{immediate
part} $\Imm(E)$ contains complementary actions that are available now: receives
for sends, sends for receives, or the clock tick $\#$. The future part
$\Fut(E)$ contains complementary \termref{def:message}{message} actions that are unavailable now but
may become available later in the same clock cycle.
For example, in the combinational assignments $b=c$ and $a=b$, the value of
$a$ depends on $c$ through $b$, although both assignments settle within one
cycle. Message availability can have the same causal structure: one exchange
can enable another before the next clock tick.
The distinction between $\Imm$ and $\Fut$ preserves this structure, allowing the semantics to capture the causal ordering of message actions across modules while reasoning about each module independently.

Fig.~\ref{fig:process-semantics-rules} presents selected rules defining this judgment.
The atomic communication rules use $\Imm(E)$ to determine whether a message action can occur: a send on $m$ can fire when $m\psym{?}$ is immediately available, while a receive on $m$ can fire when some $m\psym{!}v$ is immediately available.
Otherwise, the process may take an $\epsilon$-step without changing its syntactic state.

These waiting steps capture the semantics of synchronous hardware designs, where all combinational changes within a clock cycle are observed together, even though they may be causally ordered.
For example, in the clocked pseudocode from \figref{motivatingex}, \prog{arb1.done} is assigned in the same cycle in which \prog{arb1.gnt} is observed.
However, this assignment depends on the grant being enabled, establishing a causal dependency from \prog{arb1.gnt} to \prog{arb1.done}:
\begin{lstlisting}
if (arb1.gnt == 1) arb1.done = (cyc%4==0)?0:1 // gnt -> done
\end{lstlisting}

Thus, an action should not be ruled out merely because its complement is not immediately available. For instance, $\pkw{send}\ m\ e\ \pkw{in}\ P$ may wait for $m\psym{?}$ to become available within the current cycle. However, once the environment provides $\#$, the cycle boundary has been reached, and the process must take the clock step rather than continue waiting.

\myparagraph{Parallel Composition}
The rule for $P_1\pparop_I P_2$ derives each component's local environment from the outer environment and its sibling's offers.
Each step exposes a pair $(O_C,O_F)$, where $O_C$ contains the current offers and $O_F$ contains the bare actions that may be offered later in the cycle.
These offers describe availability, while the step label records the communications that actually occur.
Let $\Off(E,P,\alpha,P')$ denote the offer pairs associated with a step.

For each offer pair $(O_{C2},O_{F2})$ of $P_2$, $\LocalEnv(E,I,P_1,O_{C2},O_{F2})$ constructs the local environment of $P_1$.
The rule then checks the local steps and offers of both components simultaneously.
Every solution satisfying these mutual constraints contributes a step and its offers to the composition.
The premise $\Sync_I(\alpha_1,\alpha_2)$ ensures that the local observations agree on the internal interface $I$: complementary internal messages must synchronize, and both components must take the clock tick together when it is provided by the environment.

If both components finish, the composition finishes in the same step.
If one has already finished, the other continues using only the messages visible outside $I$ and supplied by the environment.
The full parallel composition rules appear in \appref{app-pact-semantics}, while the definitions of $\LocalEnv$ and $\Sync_I$ appear in \appref{aux-appendix}.

\myparagraph{Example}Consider the following three process states:
\[
\begin{aligned}
P_0 ={}&
(x=\pkw{recv}\ m_1\ \pkw{in}\
  \pkw{send}\ m_2\ v_2\ \pkw{in}\ \pkw{cycle})
\pparop_{\{m_1\}}
(\pkw{send}\ m_1\ v_1\ \pkw{in}\ \pkw{cycle}) ,\\
P_1 ={}&
(\pkw{send}\ m_2\ v_2\ \pkw{in}\ \pkw{cycle})
\pparop_{\{m_1\}}\pkw{cycle} ,\\
P_2 ={}& \pkw{cycle}\pparop_{\{m_1\}}\pkw{cycle} .
\end{aligned}
\]
In $P_0$, the message $m_1$ is internal, while $m_2$ is a boundary message.
Suppose the external environment first promises $m_2\psym{?}$ for later in the
cycle, then makes it immediate. Let
$E_0=(\emptyset,\{m_2\psym{?}\})$ and $E_1=(\{m_2\psym{?}\},\emptyset)$.
At the first step, the left component receives $m_1$ and offers $m_2\psym{!}$ for
later; the right component sends $m_1$ and has no future offer.
These offers give the local environments
\[
 E_L=(\{m_1\psym{!}v_1\},\{m_2\psym{?}\}),\qquad
 E_R=(\{m_1\psym{?}\},\{m_2\psym{?}\}).
\]
The components synchronize on $m_1$. The left component can then send $m_2$
in the same cycle, while the right component waits at \pkw{cycle}:
\[
\lstep{E_0}{P_0}{\tick{m_1}}{P_1}
\quad\text{then}\quad
\lstep{E_1}{P_1}{m_2\psym{!}v_2}{P_2}.
\]
Finally, once the environment provides the clock tick, both components take
the clock step together. In a \termref{def:closed}{closed system}, this clock tick is the only
boundary action the environment can provide.

\myparagraph{Exclusive Offer}
The future component $\Fut(E)$ is most important for offers. An \termref{def:exclusive-offer}{exclusive
offer} $\pkw{xoffer}\ B\ \pkw{else}\ P_0$ commits to one enabled branch in
$B$. If a branch head is immediately enabled, the process may take that branch.
If no branch head is immediately enabled, it may take the else branch. Both
choices are allowed only when the offer is resolved.

The premise $\Resolved(E,B)$ makes this last condition precise. Each branch in
$B$ has a need, written $\Need(B)$, namely the complementary action required to
fire that branch. The offer is resolved only when no branch need can still
become immediate later in the same cycle. Equivalently, the bare branch needs
must be disjoint from the future actions of the environment; the formal
definition is given in \appref{aux-appendix}.

The need for this check comes from hardware semantics. Consider
\[
\begin{aligned}
Q &= \pkw{xoffer}\
      (x=\pkw{recv}\ m_1 : \pkw{skip},\
       \pkw{send}\ m_2\ v_2 : \pkw{skip})\
      \pkw{else}\ \pkw{skip},\\
R &= \pkw{send}\ m_1\ v_1\ \pkw{in}\
      \pkw{xoffer}\
      (y=\pkw{recv}\ m_2 : \pkw{skip})\
      \pkw{else}\ \pkw{skip}.
\end{aligned}
\]
For $Q$, the action $m_1\psym{!}v_1$ is immediate, but $m_2\psym{?}$ is future: $R$
may expose it after the $m_1$ exchange in the same clock cycle. If $Q$
committed using only immediate availability, it could choose the receive branch
on $m_1$ too early. The later receive opportunity on $m_2$ would then be hidden
by that choice, even though the \termref{def:rtl}{RTL} may observe both opportunities within the
same cycle. Such a mismatch is exactly the kind of
\termref{def:sync-safety}{synchronization hazard} that can become undefined
combinational behaviour. Since $m_2\psym{?}$ is future, $\Resolved(E,B)$ fails, so
$Q$ must wait rather than commit prematurely.

\myparagraph{Conjunctive Offer}
The \termref{def:conjunctive-offer}{conjunctive form} $\pkw{offer}\ B\ \pkw{else}\ P_0$ models the dual case:
when several branches are enabled in one cycle, the process can accept all of
them. Operationally, the process immediately communicates every enabled head in one step.
Its residual contains the continuations of those heads, with received payloads
substituted, and singleton exclusive offers for the heads that were not enabled.
This remaining xoffer allows the process to wait for the other branches to become enabled
within the same clock-cycle.

If no branch is immediately enabled,
the else branch is available only after the same resolution check succeeds.
Thus $\pkw{xoffer}$ captures \termref{def:exclusive-offer}{exclusive external choice}, while $\pkw{offer}$
captures a conjunction of all communications that the environment enables in
the cycle.
\par\endgroup

%% file: assets/fig-semantics.tex
\begin{figure}[t]
\centering
\setlength{\topsep}{0pt}
\setlength{\partopsep}{0pt}
\newcommand{\semfigscale}{0.83}
\begin{semrulebox}[fontupper=\scriptsize]
\setlength{\parskip}{2pt}
\textbf{Atomic communication and time}

\begin{center}
\scalebox{\semfigscale}{$
\begin{array}{@{}ccc@{}}
\dfrac{m\psym{?} \in \Imm(E) \quad e\Downarrow v}
      {\lstep{E}{\pkw{send}\ m\ e\ \pkw{in}\ P}{m\psym{!}v}{P}}
\rulelabel{L-Send}
&
\dfrac{\exists v. \ m\psym{!}v \in \Imm(E)}
      {\lstep{E}{x=\pkw{recv}\ m\ \pkw{in}\ P}{m\psym{?}v}{P[x\mapsto v]}}
\rulelabel{L-Recv}
&
\dfrac{\Imm(E)=\{\#\}}
      {\lstep{E}{\pkw{cycle}}{\#}{\pkw{skip}}}
\rulelabel{L-Cycle}
\\[1.3em]
\dfrac{m\psym{?}\notin\Imm(E) \quad \Imm(E)\neq\{\#\}}
      {\lstep{E}{\pkw{send}\ m\ e\ \pkw{in}\ P}{\epsilon}{\pkw{send}\ m\ e\ \pkw{in}\ P}}
\rulelabel{L-SendWait}
&
\dfrac{\forall v.\ m\psym{!}v\notin\Imm(E) \quad \Imm(E)\neq\{\#\}}
      {\lstep{E}{x=\pkw{recv}\ m\ \pkw{in}\ P}{\epsilon}{x=\pkw{recv}\ m\ \pkw{in}\ P}}
\rulelabel{L-RecvWait}
&
\dfrac{\Imm(E)\neq\{\#\}}
      {\lstep{E}{\pkw{cycle}}{\epsilon}{\pkw{cycle}}}
\rulelabel{L-CycleWait}
\end{array}$}
\end{center}

\noindent\makebox[\linewidth]{\color{black!25}\rule{0.96\linewidth}{0.2pt}}

\textbf{Local computation}

\begin{center}
\scalebox{\semfigscale}{$
\begin{array}{@{}ccc@{}}
\dfrac{e\Downarrow\pkw{true} \quad \Imm(E)\neq\{\#\}}
      {\lstep{E}{\pkw{if}\ e\ \pkw{then}\ P_1\ \pkw{else}\ P_2}{\epsilon}{P_1}}
\rulelabel{L-IfTrue}
&
\dfrac{e\Downarrow\pkw{false} \quad \Imm(E)\neq\{\#\}}
      {\lstep{E}{\pkw{if}\ e\ \pkw{then}\ P_1\ \pkw{else}\ P_2}{\epsilon}{P_2}}
\rulelabel{L-IfFalse}
&
\dfrac{e\Downarrow v \quad \Imm(E)\neq\{\#\}}
      {\lstep{E}{\pkw{let}\ x=e\ \pkw{in}\ P}{\epsilon}{P[x\mapsto v]}}
\rulelabel{L-Let}
\end{array}$}
\end{center}

\noindent\makebox[\linewidth]{\color{black!25}\rule{0.96\linewidth}{0.2pt}}

\textbf{Fair branching and structural rules}

\begin{center}
\scalebox{\semfigscale}{$
\begin{array}{@{}ccc@{}}
\dfrac{e\Downarrow\pkw{true} \quad \Imm(E)\neq\{\#\}}
      {\lstep{E}{\pkw{fair-if}_\eta\ e\ \pkw{then}\ P_t\ \pkw{else}\ P_f}
        {\eta_{\mathsf{true}}}{P_t}}
\rulelabel{L-FairTrue}
&
\dfrac{e\Downarrow\pkw{false} \quad \Imm(E)\neq\{\#\}}
      {\lstep{E}{\pkw{fair-if}_\eta\ e\ \pkw{then}\ P_t\ \pkw{else}\ P_f}
        {\eta_{\mathsf{false}}}{P_f}}
\rulelabel{L-FairFalse}
&
\dfrac{\Imm(E)\neq\{\#\}}
      {\lstep{E}{\pmu X.P}{\epsilon}{P[X\mapsto\pmu X.P]}}
\rulelabel{L-Rec}
\end{array}$}
\end{center}

\begin{center}
\scalebox{\semfigscale}{$
\begin{array}{@{}cccc@{}}
\dfrac{\lstep{E}{P_1}{\alpha}{P_1'} \quad P_1'\neq\pkw{skip}}
      {\lstep{E}{P_1\pseq P_2}{\alpha}{P_1'\pseq P_2}}
\rulelabel{L-Seq}
&
\dfrac{\lstep{E}{P_1}{\alpha}{\pkw{skip}}}
      {\lstep{E}{P_1\pseq P_2}{\alpha}{P_2}}
\rulelabel{L-SeqDone}
&
\dfrac{\Imm(E)\neq\{\#\}}
      {\lstep{E}{\pkw{skip}}{\epsilon}{\pkw{skip}}}
\rulelabel{L-Skip}
&
\dfrac{\Imm(E)\neq\{\#\}}
      {\lstep{E}{\pkw{end}_{t}}{\$_t}{\pkw{skip}}}
\rulelabel{L-End}
\end{array}$}
\end{center}

\noindent\makebox[\linewidth]{\color{black!25}\rule{0.96\linewidth}{0.2pt}}

\textbf{Exclusive offer}

\begin{center}
\begin{adjustbox}{max width=\linewidth}
\scalebox{\semfigscale}{$
\begin{array}{@{}cc@{}}
\dfrac{\pkw{send}\ m\ e:P_i\in B \quad m\psym{?}\in\Imm(E) \quad e\Downarrow v \quad \Resolved(E,B)}
      {\lstep{E}{\pkw{xoffer}\ B\ \pkw{else}\ P_0}{m\psym{!}v}{P_i}}
\rulelabel{L-XOfferSend}
&
\dfrac{x=\pkw{recv}\ m\ :P_i\in B \quad m\psym{!}v\in\Imm(E) \quad \Resolved(E,B)}
      {\lstep{E}{\pkw{xoffer}\ B\ \pkw{else}\ P_0}{m\psym{?}v}{P_i[x\mapsto v]}}
\rulelabel{L-XOfferRecv}
\\[1.3em]
\dfrac{\Need(B)\cap\Bare(\Imm(E))=\emptyset \quad \Resolved(E,B) \quad \Imm(E)\neq\{\#\}}
      {\lstep{E}{\pkw{xoffer}\ B\ \pkw{else}\ P_0}{\epsilon}{P_0}}
\rulelabel{L-XOfferElse}
&
\dfrac{\neg\,\Resolved(E,B) \quad \Imm(E)\neq\{\#\}}
      {\lstep{E}{\pkw{xoffer}\ B\ \pkw{else}\ P_0}{\epsilon}{\pkw{xoffer}\ B\ \pkw{else}\ P_0}}
\rulelabel{L-XOfferWait}
\end{array}$}
\end{adjustbox}
\end{center}

\noindent\makebox[\linewidth]{\color{black!25}\rule{0.96\linewidth}{0.2pt}}

\textbf{Parallel composition}

\begin{center}

\begin{adjustbox}{max width=.97\linewidth}
$\dfrac{\begin{gathered}
\lstep{E_1}{P_1}{\alpha_1}{P_1'}\quad \lstep{E_2}{P_2}{\alpha_2}{P_2'}\\
(O_{C1},O_{F1})\in\Off(E_1,P_1,\alpha_1,P_1')\quad (O_{C2},O_{F2})\in\Off(E_2,P_2,\alpha_2,P_2')\\
E_1\in\LocalEnv(E,I,P_1,O_{C2},O_{F2})\quad E_2\in\LocalEnv(E,I,P_2,O_{C1},O_{F1})\quad \Sync_I(\alpha_1,\alpha_2)\\
P_1\neq \pkw{skip}\quad P_2\neq \pkw{skip}\quad (P_1'\neq \pkw{skip}\ \text{or}\ P_2'\neq \pkw{skip})\end{gathered}}{\lstep{E}{P_1\pparop_I P_2}{\alpha_1\unionI\alpha_2}{P_1'\pparop_I P_2'}}\rulelabel{L-Par}$
\end{adjustbox}

\end{center}
\end{semrulebox}
\caption{Selected operational rules for \codename{} processes. The complete rules
and auxiliary functions appear in \appref{app-pact-semantics} and
\appref{aux-appendix}, respectively.}

\label{fig:process-semantics-rules}
\end{figure}

%% file: sections/pact-safety.tex
\subsection{Defining Communication Safety}
\seclabel{pact-safety}

A composite system is built by nesting the parallel operator over component
\termref{def:process}{processes},
\[
Q_0 = P_1\pparop_{I_1}\cdots\pparop_{I_{n-1}}P_n,
\]
for some fixed nesting. Each component is a non-terminating hardware process,
so each $P_i$ has the form $\pmu X.P$. Once $Q_0$ is fixed, its boundary messages
$\BoundaryMsgs(Q_0)$, internal messages $\IntMsgs(Q_0)$, action labels
$\ActSet{Q_0}$, and one-step environments are exactly those of
\secref{pact-semantics}, read for $P:=Q_0$.

\subsubsection{Synchronization Safety}

Intuitively, a system is \termref{def:sync-safety}{synchronization safe} when
any component that intends to perform a message action can complete that action
with its peer in the same clock cycle.

For example, consider
\[
\begin{aligned}
 P_1=\pmu X.\,(x=\pkw{recv}\ m_1\ \pkw{in}\
                    (\pkw{cycle}\pseq \pkw{end}_1\pseq X)) \qquad \text{and} \qquad
 P_2=\pmu Y.\,(\pkw{cycle}\pseq \pkw{end}_2\pseq Y)
\end{aligned}
\]
composed on the shared interface $I=\{m_1\}$. The system $Q=P_1\pparop_I P_2$
cannot complete the current cycle: $P_1$ waits to receive $m_1$, but $P_2$
never sends it. When the environment offers the clock tick, $Q$ has no
transition, because $P_1$ cannot advance past the unmatched receive.
No external environment can supply this internal message, so $Q$ is
\term{def:stuck}{stuck} under that clock environment.
The same notion extends to \termref{def:open}{open systems} by fixing the \termref{def:one-step-environment}{one-step environment}
$E$.

\begin{definition}[Stuck]
A system $Q$ is stuck under a one-step environment $E$ if no transition is
possible from $Q$ under $E$:
\[
\Stuck_E(Q)\triangleq
\forall\alpha\in\ActSet{Q}.\ \nexists Q'.\ \lstep{E}{Q}{\alpha}{Q'}.
\]
\end{definition}

To discuss \termref{def:stuck}{stuck} states, we must fix the multi-step environment in which a
system runs.  A finite execution prefix from $Q_0$ has the form
\[
  h=
  Q_0\langle E_0,\alpha_0\rangle Q_1\cdots
  \langle E_{k-1},\alpha_{k-1}\rangle Q_k,
\]
where each $E_i$ is a well-formed one-step environment for $Q_0$ and
$\lstep{E_i}{Q_i}{\alpha_i}{Q_{i+1}}$ for every $i<k$. Let
$\mathsf{last}(h)=Q_k$ and $h_0=Q_0$. A \term{def:strategy}{history-sensitive strategy} is a
function
\[
  \Strategy_{Q_0}:\mathsf{Pref}(Q_0)\longrightarrow\mathsf{Env}(Q_0)
\]
from such prefixes to one-step environments. Its argument contains the current
residual and every preceding environment and action. In an open system, the
strategy abstracts the missing boundary
components. In a closed system, no message actions remain at the boundary, so
the strategy can offer only the global clock tick or nothing.

Under a strategy $\Strategy_{Q_0}$, prefixes extend as follows:
\[
\begin{aligned}
\stratstep{\Strategy_{Q_0}}{}{h}{\alpha}{h'}
\triangleq{}&
\exists Q'.\
\lstep{\Strategy_{Q_0}(h)}{\mathsf{last}(h)}{\alpha}{Q'}
\\[-1mm]
&\land\
h'=h\cdot
  \langle\Strategy_{Q_0}(h),\alpha\rangle Q'.
\end{aligned}
\]

Every maximal execution is represented by an infinite, \term{def:environment-annotated-log}{environment-annotated
log}
\[
  \ell=\langle E_0,\alpha_0\rangle
       \langle E_1,\alpha_1\rangle\cdots .
\]
An ordinary strategy step from $h_i$ records
$\langle\Strategy_{Q_0}(h_i),\alpha_i\rangle$.  If the selected environment
has no successor, completion records
$\langle\Strategy_{Q_0}(h_i),\epsilon\rangle$, enters an absorbing failure
state, and repeats an empty-environment $\epsilon$-entry forever.  Thus a failed execution has an infinite log with no further clock ticks. Let $\Logs^{\omega}_{\Strategy_{Q_0}}(Q_0)$ contain all ordinary
infinite logs and all such failure completions.

Not every \termref{def:strategy}{strategy} is a faithful model of an environment.  Recall that
$\Fut(E)$ records what may still become available later in the current cycle,
and define availability without the clock by
\[
  \Avail(E)=\Bare\bigl(\Imm(E)\setminus\{\#\}\bigr)\cup\Fut(E).
\]
A log is \term{def:honest}{honest} when each non-clock step respects the
promise made by the preceding environment; a clock step starts a new cycle and
therefore resets the condition.  A strategy is honest exactly when all of its
completed logs are honest:
\[
\begin{aligned}
\Honest(\ell)
&\triangleq
\forall i\geq 0.\
  \#\in\alpha_i\ \lor\
  \Avail(E_{i+1})\subseteq\Fut(E_i),
\\
\Honest(\Strategy_{Q_0})
&\triangleq
\forall\ell\in\Logs^{\omega}_{\Strategy_{Q_0}}(Q_0).\
  \Honest(\ell).
\end{aligned}
\]
Thus an honest strategy may withdraw future actions, but before the next clock
it cannot expose an action already ruled out. Since immediate and future
actions are disjoint, an immediate offer also cannot reappear in that cycle.  The adjacent condition implies the corresponding condition between
any two positions in the same cycle.  This is exactly what makes
$\Resolved(E,B)$ (\secref{pact-semantics}) a safe point at which to commit to
an else branch: once a branch need drops out of $\Fut$, it cannot resurface
before the next $\#$.

With a fixed \termref{def:honest}{honest} strategy, \term{def:synchronization-compatibility}{synchronization compatibility} becomes a
reachability condition. Let $\Reach_{\Strategy_{Q_0}}(Q_0)$ contain the prefixes
reachable from $h_0$ by the strategy-step relation.  No such prefix may be
stuck under the environment selected for it.

\begin{definition}[Sync-Compatible]
A system $Q_0=P_1\pparop_{I_1}\cdots\pparop_{I_{n-1}}P_n$ is
sync-compatible under $\Strategy_{Q_0}$ if the residual at no reachable prefix
is stuck under the environment selected by the strategy:
\[
\SyncCompat_{\Strategy_{Q_0}}(P_1,\dots,P_n)
\triangleq
\forall h\in\Reach_{\Strategy_{Q_0}}(Q_0).\
\neg\Stuck_{\Strategy_{Q_0}(h)}\bigl(\mathsf{last}(h)\bigr).
\]
\end{definition}

Synchronization safety closes this condition existentially over \termref{def:honest}{honest}
strategies. It asks whether \emph{some} honest strategy keeps the system out of
\termref{def:stuck}{stuck} states. This condition alone allows a strategy to
withhold the clock indefinitely; liveness below also requires clock progress.

\begin{definition}[Synchronization Safety]
A system is sync-safe if there exists an honest strategy under which it is
sync-compatible:
\[
\exists\Strategy_{Q_0}.\ \Honest(\Strategy_{Q_0})\land\SyncCompat_{\Strategy_{Q_0}}(P_1,\dots,P_n).
\]
\end{definition}

\subsubsection{Liveness}

\termref{def:sync-safety}{Synchronization safety} only rules out \termref{def:stuck}{stuck} states. It does not require the
system to keep making \emph{progress}. \termref{def:liveness}{Liveness} adds that requirement.

For example, consider a process that waits across cycles for $m_1$ and
records progress only after receiving it:
\[
 P=\pmu X.\,\bigl(
   (\pmu Y.\,\pkw{xoffer}\ \{x=\pkw{recv}\ m_1:\pkw{skip}\}
       \ \pkw{else}\ (\pkw{cycle}\pseq Y))\pseq \pkw{end}_1\pseq X\bigr).
\]
A strategy that never offers $m_1$ can keep the process running through its
else branch, but the process never reaches progress marker $1$.
Liveness rules out this strategy for the progress set $T=\{1\}$.
Compatibility asks whether some strategy meets both requirements, under every
resolution of the process choices.

The completed logs above let us state these requirements for all executions
under a strategy.

In practice, \termref{def:liveness}{liveness} often requires \termref{def:fairness}{fairness} assumptions. In
\codename{}, when a designer uses a fair-tagged construct such as
$\pkw{fair-if}_\eta$, the specification should only consider executions where
both branches of that construct are taken infinitely often, or neither branch
is taken infinitely often.

\begin{definition}[Fairness]
For a log $\ell$, let $\Inf(\ell)$ inspect the action component and collect the
actions that occur infinitely often. A log is fair when each fair branch resolves both ways
infinitely often or neither way infinitely often:
\[
\Fair(\ell)\triangleq
\forall\eta\in\EtaSet{Q_0}.\
\eta_{\mathsf{true}}\in\Inf(\ell)\Longleftrightarrow\eta_{\mathsf{false}}\in\Inf(\ell).
\]
\end{definition}

With logs and \termref{def:fairness}{fairness} in hand, \termref{def:liveness}{liveness} has a direct form. Every infinite log
must contain infinitely many clock ticks. In addition, every fair infinite log
must revisit each progress marker in $T$ infinitely often.

\begin{definition}[Liveness]
For $T\subseteq\PidSet{Q_0}$, the progress identifiers of interest, the system
$Q_0$ is live under $\Strategy_{Q_0}$ with progress set $T$ if every fair
infinite log visits each progress marker in $T$ infinitely often and every
infinite log contains infinitely many clock ticks:
\[
\Live_{\Strategy_{Q_0},T}(Q_0)\triangleq
\forall\ell\in\Logs^{\omega}_{\Strategy_{Q_0}}(Q_0).\
\#\in\Inf(\ell)
\ \text{and}\
\bigl(\Fair(\ell)\Rightarrow\forall t\in T.\ \$_t\in\Inf(\ell)\bigr).
\]
\end{definition}

\subsection{Compatible Processes}

A system is \termref{def:compatibility}{compatible} under a given \termref{def:strategy}{strategy} when it is both sync-compatible
and live under that strategy. Process compatibility, or
\termref{def:comm-safety}{communication safety}, asks whether such an honest
strategy exists.

\begin{definition}[Compatible Processes]
\deflabel{compatible-processes-main}
For a fixed composition $Q_0=P_1\pparop_{I_1}\cdots\pparop_{I_{n-1}}P_n$,
processes are compatible under $\Strategy_{Q_0}$ with progress set $T$ when
\[
\ProcJoin_{\Strategy_{Q_0},T}(P_1,\dots,P_n)
\triangleq
\SyncCompat_{\Strategy_{Q_0}}(P_1,\dots,P_n)
\ \text{and}\
\Live_{\Strategy_{Q_0},T}(Q_0),
\]
and the system is compatible when some honest strategy achieves this:
\[
\exists\Strategy_{Q_0}.\ \Honest(\Strategy_{Q_0})\land\ProcJoin_{\Strategy_{Q_0},T}(P_1,\dots,P_n).
\]
\end{definition}

\termref{def:closed}{Closed compatibility} is the special case $\BoundaryMsgs(Q_0)=\emptyset$, where
the strategy can only offer the clock tick or nothing. The reachable prefixes
and closed environments are given explicitly in
\appref{app-pact-safety}.

%% file: sections/types.tex
\section{The \codename{} Type System}
\seclabel{types}

This section presents the \term{def:type-system}{\codename{} type system}. We first introduce the
syntax and local semantics of \codename{} types. We then define \termref{def:type-inference}{type inference}
and \termref{def:type-checking}{type checking}: given a \termref{def:process}{process}, the compiler infers a type that captures
its communication, timing, and control behaviour, and checks user-provided
contracts by comparing them against the inferred type. Next, we define
\termref{def:compatibility}{compatibility} for collections of types and show how to decide it by reducing it to
questions on games on finite graphs.
We then characterize the optimality guarantee of our \termref{def:type-inference}{type inference}
algorithm. Finally, we describe how summary types allow compatible
compositions to be reused in incremental checks.

\input{sections/type-syntax-semantics}
\input{sections/typechecking-inference}
\input{sections/type-compatibility-checking}
\input{sections/type-optimality}
\input{sections/summary-types}

%% file: sections/type-syntax-semantics.tex
\subsection{Type Syntax and Semantics}
\seclabel{type-syntax-semantics}

\codename{} \term{def:type}{types} describe the value-erased communication structure exposed by
a \termref{def:process}{process}: message order, branch outcomes, clock ticks, and \termref{def:progress-marker}{progress markers}.
Similar to processes, types carry a finite alphabet of polarized message
actions, which we omit from the main text for simplicity. The full syntax and
well-formedness conditions are given in \appref{type-semantics-appendix}.
The type syntax is:
\begin{align*}
\tau ::= {}& a \mid \tmarker{t} \mid \tclock \mid \tkw{skip}
        \mid \tau \tplus \tau \mid \tau\tseq \tau \mid \tau \tpar{I} \tau
        \mid X \mid \tmu X. \tau \\
        &\mid \tau \tfairplus{\eta} \tau
        \mid \toplus (\tau, a : \tau, \cdots)
        \mid \twedge (\tau, a : \tau, \cdots) \\
        a ::= {}& m\tsym{!} \mid m\tsym{?}
\end{align*}
The atomic types $m\tsym{!}$ and $m\tsym{?}$ describe \termref{def:message-send}{sending} and \termref{def:message-receive}{receiving} message $m$.
The type $\tclock$ describes a clock tick, and $\tmarker{t}$ records that
\termref{def:progress-marker}{progress marker} $t$ has been reached. The type $\tkw{skip}$ is completed
behaviour. Sequencing $\tau_1\tseq\tau_2$ performs $\tau_1$ and then $\tau_2$.
Parallel composition $\tau_1\tpar{I}\tau_2$ synchronizes the two types on the
interface $I$. Recursion $\tmu X.\tau$ gives regular, repeating protocols.
We consider well-formed, closed, guarded \term{def:regular-type}{regular types}: when unfoldings of each recursive
binder are shared, only finitely many control residuals are generated.  This is
the finiteness condition required by the compatibility algorithm.

The remaining constructors describe choices. The sum $\tau_1\tplus\tau_2$ is an
internal nondeterministic choice; $\tau_1\tfairplus{\eta}\tau_2$ is the fair-tagged
version, whose branch outcomes are recorded using the identifier $\eta$. The
\termref{def:exclusive-offer}{exclusive offer} $\toplus(\tau_0,a_1:\tau_1,\ldots)$ observes which branch
needs are available and commits to one enabled branch, or to the default type
$\tau_0$ when no branch can still become enabled in the current cycle.
The default is the first argument of a type offer; in process syntax it
follows \pkw{else}. The
\termref{def:conjunctive-offer}{conjunctive offer} $\twedge(\tau_0,a_1:\tau_1,\ldots)$ admits all enabled
branches in the current cycle, again falling through to $\tau_0$ when the
branch set is resolved and no branch is enabled.

\termref{def:type}{Types} abstract away payload values. A process action such as $m\psym{!}v$ or $m\psym{?}v$
is represented by the corresponding message polarity $m\tsym{!}$ or $m\tsym{?}$. Likewise,
data-dependent conditionals in processes become type-level choice, either
$\tplus$ or $\tfairplus{\eta}$.

The operational semantics follows the process semantics after this erasure.
A type step has the same local form:
\[
  \lstep{E_\tau}{\tau}{\alpha}{\tau'}.
\]
Here $E_\tau$ is a \termref{def:one-step-environment}{one-step environment} for the type, $\alpha$ is a type
observation, and $\tau'$ is the \term{def:type-residual}{residual of $\tau$} after that step.
This semantic residual is distinct from a summary type, which describes the
boundary behaviour of a compatible composition.
As for processes, the environment records which complementary boundary actions
are available now and which may still become available later in the same clock
cycle. The full erased semantic rules are given in
\appref{type-semantics-appendix}. They are structured so that a type can be
executed under a boundary strategy in the same way as the process it abstracts.

%% file: sections/typechecking-inference.tex
\subsection{Type Inference, Refinement, and Checking}
\seclabel{type-checking-inference}

Our \term{def:type-inference}{type inference} algorithm erases payload values
and expressions while preserving communication and timing.
\term{def:type-checking}{Type checking} compares this inferred type with a
claimed type using \term{def:type-refinement}{type refinement}.
To justify these checks, we first relate process steps to type steps after
erasing their payload values.

The \term{def:value-erasure}{erasure} $\erase$ maps process observations to
type observations by removing payloads from boundary messages and keeping
other observations unchanged:
\[
\erase(m\psym{!}v)=m\tsym{!},
\qquad
\erase(m\psym{?}v)=m\tsym{?},
\qquad
\erase(\#)=\tclock,
\qquad
\erase(\$_t)=\tmarker{t}.
\]
It also preserves fair-branch observations,
$\erase(\eta_{\mathsf{true}})=\eta_{\mathsf{true}}$ and
$\erase(\eta_{\mathsf{false}})=\eta_{\mathsf{false}}$, and maps $\epsilon$ to
$\epsilon$. We also lift $\erase$ to sets of actions, one-step
environments, and histories and for a process environment
$E_P=(U,F)$ the corresponding type
environment is $\erase(E_P)=(\erase(U),\erase(F))$.

\myparagraph{Type conformance}
To justify checking \termref{def:comm-safety}{communication safety} on types, we first define when a
type describes the behaviour of a process.
\begin{definition}[Process-type conformance]
\deflabel{process-type-conformance}
A process $P$ \term{def:process-type-conformance}{conforms} to a type $\tau$, written $P \models \tau$,
if there exists a relation $\Rel$ between processes and types such that
$(P,\tau)\in\Rel$ and, for every $(P_0,\tau_0)\in\Rel$ and concrete
one-step environment $E_P$:
\begin{enumerate}
 \item \textbf{Availability reflection.} If $\tau_0$ can take a step under $\erase(E_P)$, then $P_0$
       can take a step under $E_P$.
 \item \textbf{Step matching.} Every process step $\lstep{E_P}{P_0}{\alpha}{P_1}$ has a matching
       type step $\lstep{\erase(E_P)}{\tau_0}{\erase(\alpha)}{\tau_1}$
       with $(P_1,\tau_1)\in\Rel$.
\end{enumerate}
\end{definition}

The relation $\Rel$ is a witness for \termref{def:process-type-conformance}{conformance}. It records the process/type
pairs that remain matched as execution proceeds. The first clause says that
if the process cannot step, neither can the type under the erased
environment. The second clause requires the type to match each process step
after erasing payload values.

\myparagraph{Lifting from types to processes under value-oblivious objectives}
Process-type conformance is useful because it lets us transfer winning
environment strategies from types back to processes. Informally, if
$P\models\tau$, then any game that the type $\tau$ can win can also be won by
the process $P$, provided the winning condition does not depend on payload
values, i.e. it is \term{def:value-oblivious}{value-oblivious}. Formally, an objective $\Phi$
over infinite environment-annotated logs is value-oblivious if
\[
\erase(\ell)=\erase(\ell')
\quad\Longrightarrow\quad
\ell\in\Phi \Longleftrightarrow \ell'\in\Phi.
\]
Erasure acts pointwise on both the environment and action components.
A strategy is only a faithful model of an environment if it is
\termref{def:honest}{honest} in the sense of \secref{pact-safety}: it must never
expose, immediately or as a future promise, an action it had already ruled out
earlier in the same clock cycle. We therefore only ever transfer honest
strategies. We can now state the formal guarantee offered by type conformance
(proof in \appref{sec4-type-proofs}):
\begin{restatable}[Strategy transfer]{lemma}{StrategyTransfer}
\lemlabel{strategy-transfer}
Let $\Phi$ be a \termref{def:value-oblivious}{value-oblivious} objective over infinite environment-annotated
logs. If $P\models\tau$ and an honest type-level environment strategy $\Strategy_\tau$
ensures $\Phi$ for $\tau$, then there exists an honest process-level
environment strategy $\Strategy_P$ that ensures $\Phi$ for $P$.
\end{restatable}

\myparagraph{Inference}
The inferred type follows the process syntax. The abstraction erases
payloads and expressions while retaining communication and control steps:
{\small
\[
\begingroup
\renewcommand{\arraystretch}{1.15}
\begin{array}{@{}l@{\hspace{2.4em}}l@{}}
\multicolumn{2}{@{}c@{}}{
  \begin{array}{@{}l@{\hspace{1.1em}}l@{\hspace{1.1em}}l@{\hspace{1.1em}}l@{}}
  \mathsf{abs}(\pkw{skip}) \triangleq \tkw{skip},
  &
  \mathsf{abs}(\pkw{cycle}) \triangleq \tclock,
  &
  \mathsf{abs}(\pkw{end}_{t}) \triangleq \tmarker{t},
  &
  \mathsf{abs}(X) \triangleq X,
  \end{array}
}\\[2pt]
\mathsf{abs}(\pkw{if}\ e\ \pkw{then}\ P\ \pkw{else}\ Q)
\triangleq \mathsf{abs}(P)\tplus\mathsf{abs}(Q),
&
\mathsf{abs}(\pkw{send}\ m\ e\ \pkw{in}\ P)
\triangleq m\tsym{!}\tseq\mathsf{abs}(P),
\\
\mathsf{abs}(\pkw{fair-if}_\eta\ e\ \pkw{then}\ P\ \pkw{else}\ Q)
\triangleq \mathsf{abs}(P)\tfairplus{\eta}\mathsf{abs}(Q),
&
\mathsf{abs}(x=\pkw{recv}\ m\ \pkw{in}\ P)
\triangleq m\tsym{?}\tseq\mathsf{abs}(P),
\\[2pt]
\mathsf{abs}(\pkw{let}\ x=e\ \pkw{in}\ P)
\triangleq \tkw{skip}\tseq\mathsf{abs}(P),
&
\mathsf{abs}(\pmu X.P)
\triangleq \tmu X.\mathsf{abs}(P),
\\
\mathsf{abs}(P\pseq Q)
\triangleq \mathsf{abs}(P)\tseq\mathsf{abs}(Q),
&
\mathsf{abs}(P\pparop_I Q)
\triangleq \mathsf{abs}(P)\tpar{I}\mathsf{abs}(Q),
\\
\mathsf{abs}(\pkw{xoffer}\ B\ \pkw{else}\ P)
\triangleq \toplus(\mathsf{abs}(P),\mathsf{abs}_B(B)),
&
\mathsf{abs}(\pkw{offer}\ B\ \pkw{else}\ P)
\triangleq \twedge(\mathsf{abs}(P),\mathsf{abs}_B(B)).
\end{array}
\endgroup
\]
}
Here $\mathsf{abs}_B$ abstracts branch sets by erasing each branch head to
its message polarity and recursively abstracting the branch continuation:
{\small
\[
\begingroup
\renewcommand{\arraystretch}{1.15}
\begin{array}{@{}l@{\hspace{2.4em}}l@{}}
\mathsf{abs}_B(\emptyset)
\triangleq \emptyset,
&
\mathsf{abs}_B(B,\ \pkw{send}\ m\ e : P)
\triangleq \mathsf{abs}_B(B),\ m\tsym{!}:\mathsf{abs}(P),
\\
\multicolumn{2}{@{}l@{}}{
  \mathsf{abs}_B(B,\ x=\pkw{recv}\ m : P)
  \triangleq \mathsf{abs}_B(B),\ m\tsym{?}:\mathsf{abs}(P).
}
\end{array}
\endgroup
\]
}
Thus payloads and expressions are erased, while each local binding retains
its silent step. Communication
structure, timing actions, progress markers, recursion, and fair-choice
identifiers are preserved.

The inference judgment associates each process with this abstraction:
\[
  \inferpact{P}{\tau}
  \quad\triangleq\quad
  \tau = \mathsf{abs}(P).
\]

\begin{restatable}[Soundness of inferred types]{lemma}{DataAbsSound}
\lemlabel{data-abs-sound}
For every well-formed process $P$, $P\models\mathsf{abs}(P)$.
\end{restatable}

The proof is in \appref{sec4-type-proofs}; the process well-formedness
requirements are in Appendix~\ref{app:process-well-formedness}.

\myparagraph{Refinement}
To check a user-provided type against the inferred one, we compare types using
a \term{def:precision-preorder}{precision preorder}. We will use
$\tau\sqsubseteq\tau'$ to mean that $\tau$ refines $\tau'$.
\begin{definition}[Type refinement]
\deflabel{type-refinement}
A type $\tau$ refines a type $\tau'$ if there exists a relation $\Rel$
between types at the same fixed boundary such that $(\tau,\tau')\in\Rel$
and, for every $(\tau_0,\tau'_0)\in\Rel$ and one-step environment $E$:
\begin{enumerate}
 \item \textbf{Availability.} Every type step $\lstep{E}{\tau'_0}{\alpha}{\tau'_1}$
       with $(O_C,O_F)\in\Off(E,\tau'_0,\alpha,\tau'_1)$ has a matching
       type step $\lstep{E}{\tau_0}{\alpha}{\tau_1}$
       with the same offer pair.
 \item \textbf{Response.} Every type step $\lstep{E}{\tau_0}{\alpha}{\tau_1}$
       with $(O_C,O_F)\in\Off(E,\tau_0,\alpha,\tau_1)$ has a matching
       type step $\lstep{E}{\tau'_0}{\alpha}{\tau'_1}$
       with the same offer pair and $(\tau_1,\tau'_1)\in\Rel$.
 \item \textbf{Completion.} $\tau_0=\tkw{skip}$ iff $\tau'_0=\tkw{skip}$.
\end{enumerate}
\end{definition}

Here $\tau$ is the more precise type. Availability and response require
related types to admit the same labels and offer pairs under each environment.
Only response requires the resulting types to remain related.
Equal offer pairs preserve the local environments used in parallel composition.
The finite refinement game used by the checker is given in
\appref{typeinference-optimality}.

\begin{restatable}[Conformance is closed under refinement]{lemma}{ConformanceRefinement}
\lemlabel{conformance-refinement}
If $P\models\tau$ and $\tau\sqsubseteq\tau'$, then $P\models\tau'$.
\end{restatable}

Proof in \appref{sec4-type-proofs}.

\myparagraph{Data-oblivious type checking}
The abstraction $\mathsf{abs}(P)$ erases payload values and expressions. For example, when a process branches on a data expression,
the inferred type keeps both branches rather than trying to prove which branch
is feasible. This is the sense in which the checking procedure is
\term{def:data-oblivious-type-checking}{data-oblivious}: it validates contracts using the communication and timing
structure of the process, but it does not use value-sensitive invariants to
remove behaviours.

This abstraction makes inference automatic and keeps checking finite-state. The
\termref{def:data-oblivious-type-checking}{data-oblivious type checking} judgment $\checkdo{P}{\tau}$ asks whether a
claimed type is a further abstraction of the type synthesized by inference.

\begin{definition}[Data-oblivious type checking]
\deflabel{data-oblivious-type-checking}
\[
  \checkdo{P}{\tau}
  \quad\triangleq\quad
  \exists \tau_0.\
    \inferpact{P}{\tau_0}
    \land
    \tau_0 \sqsubseteq \tau .
\]
Equivalently, $\checkdo{P}{\tau}$ iff $\mathsf{abs}(P)\sqsubseteq\tau$.
\end{definition}

\begin{restatable}[Soundness of type checking]{lemma}{TypeCheckingSound}
\lemlabel{type-checking-sound}
If $P$ is well formed and $\checkdo{P}{\tau}$, then $P\models\tau$.
\end{restatable}

Proof in \appref{sec4-type-proofs}.

We will call a process or type composition
\term{def:valid-composition}{valid} if it is well formed. For a fixed process
composition $Q_0$, a type composition is valid for $Q_0$ if it is valid and
has the same grouping and interfaces as $Q_0$. Corresponding components must
have the same boundary actions, fair-choice identifiers, and progress
identifiers, and compatibility checks use the same set $T$ of required
progress markers. The precise conditions are given in \appref{valid-compositions}.

\begin{corollary}[Compositional conformance of checked types]
\corlabel{compositional-conformance}
Let $P_1\pparop_I P_2$ be a valid process composition and $\tau_1\tpar{I}\tau_2$ a
valid type composition for it.
If $\checkdo{P_i}{\tau_i}$ for $i=1,2$, then
\[
  P_1\pparop_I P_2\models\tau_1\tpar{I}\tau_2.
\]
\end{corollary}

The full statement and proof are given in
Appendix~\ref{app:compositional-conformance-proof}.

Applying this corollary at each parallel composition establishes conformance
for any fixed nesting of individually checked components. This justifies
checking implementations locally before checking compatibility of their types.

\begin{restatable}[Decidability and complexity of type inference and checking]{theorem}{TypeCheckingDecidable}
\thmlabel{type-checking-decidable}
For every well-formed process $P$ and well-formed regular type $\tau$, the judgments
$\inferpact{P}{\tau}$ and $\checkdo{P}{\tau}$ are decidable. Moreover,
computing the inferred type takes time linear in $|P|$. Given indexed
explicit finite type automata (\secref{typecompat}) of total sizes $m$ and $n$, including states
and transitions, type checking runs in time $O(mn)$.
\end{restatable}

Proof in \appref{sec4-type-proofs}; the refinement-game construction is given in
\appref{typeinference-optimality}.

%% file: sections/type-compatibility-checking.tex
\subsection{Type Compatibility Checking}
\seclabel{typecompat}

\termref{def:compatibility}{Compatibility} of a tuple of types can be decided by solving a finite
\term{def:two-player-graph-game}{two-player graph game}.  Player~0 represents the \termref{def:boundary-environment}{boundary environment}: in each
round it chooses which boundary actions are available now and which may become
available later in the current clock cycle.  Player~1 represents the composed
system: it resolves the internal nondeterminism of the types.  Player~0 wins
only when the \term{def:communication-objective}{communication objective} holds for every resolution chosen by
Player~1.

We derive this reduction in three steps.  First, we present executions as
infinite \termref{def:environment-annotated-log}{logs} that record both environment choices and system observations,
and state the winning condition on such logs.  Second, we compile types to
finite \term{def:type-automaton}{automata} without changing their behaviour.  Finally, we compose these
automata with finite monitors for honesty, fairness, and progress.  The result
is a finite \term{def:rabin-game}{Rabin game}.

Fix a valid type composition
\[
  \Theta_0=\tau_1\tpar{I_1}\cdots\tpar{I_{n-1}}\tau_n,
  \qquad I=I_1\cup\cdots\cup I_{n-1},
\]
and let \(T\subseteq\Pids(\Theta_0)\) be the \termref{def:progress-marker}{progress markers} of interest.  We
use \(\TypeCompat_{I,T}(\tau_1,\ldots,\tau_n)\) for its compatibility judgment,
with the fixed parenthesization and individual interfaces left implicit.

\subsubsection{Infinite Environment-Annotated Logs}
\label{subsec:tc-objective}

For types, we will use the definitions of strategies, completed logs, honesty,
and fairness from \secref{pact-safety}, with payloads erased.
For a boundary strategy \(\Strategy_\Theta\), let
\(\Logs^\omega_{\Strategy_\Theta}(\Theta_0)\) contain all completed logs of
\(\Theta_0\) generated by that strategy, over every resolution of the internal
choices. The communication objective combines honesty with clock and
component progress:
\[
\begin{aligned}
\Psi_T(\ell)\triangleq{}&
\Honest(\ell)
\ \land\ \tclock\in\Inf(\ell)
\\[-1mm]
&{}\land
\left(
  \TypeFair(\ell)
  \Longrightarrow
  \forall t\in T.\ \tmarker{t}\in\Inf(\ell)
\right).
\end{aligned}
\]
Here fairness ranges over the fair-choice identifiers of $\Theta_0$.
Every winning log is \termref{def:honest}{honest} and makes clock progress; if it is fair, it
also revisits every required progress marker.  A failure-completed log has only
\(\epsilon\)-labels from the sink onward and therefore loses the clock
condition.

\begin{proposition}[Regularity of the log objective]
\proplabel{type-objective-regular}
For the fixed valid composition $\Theta_0$, $\Psi_T$ is $\omega$-regular.
\end{proposition}

The full statement and proof are in \appref{type-compatibility-construction}.

The log formulation now fixes what it means to win.  To obtain an algorithm,
it remains to represent the type dynamics by a finite graph.

\subsubsection{From Types and Processes to Finite Automata}
\label{subsec:tc-automata}

The type semantics provides the edges of this graph. Its states are the
reachable residuals, with repeated recursive syntax shared to obtain a finite
control representation. This gives the canonical semantic automaton
\[
\begin{aligned}
 \Aut_\tau&=(Q_\tau,\Sigma_\tau,\delta_\tau,i_\tau,F_\tau),\\
 Q_\tau&=\{\lfloor\rho\rfloor_\tau\mid\rho\text{ is reachable from }\tau\},
 \qquad i_\tau=\lfloor\tau\rfloor_\tau.
\end{aligned}
\]
Here $\lfloor\rho\rfloor_\tau$ is the control state representing residual
$\rho$. Recursive unfoldings share binder headers, while their unfolding
steps remain explicit. The initial state represents $\tau$, and the final
states represent completed $\tkw{skip}$ residuals.

Each semantic step determines an edge that retains its environment,
observation, and current/future offers:
\[
\begin{aligned}
 \delta_\tau=\bigl\{\,
 &\atrans{\lfloor\rho\rfloor_\tau}{E_C}{E_F}{\alpha}{O_C}{O_F}
          {\lfloor\rho'\rfloor_\tau}\ \bigm|\ e=(E_C,E_F),\\
 &\lfloor\rho\rfloor_\tau\in Q_\tau,\quad
   \lstep{e}{\rho}{\alpha}{\rho'},\quad
   (O_C,O_F)\in\Off(e,\rho,\alpha,\rho')\,\bigr\}.
\end{aligned}
\]
The environment and observation become one log entry. The offer fields are
also retained because they determine the local environments when components
synchronize. Their finite alphabets, and hence $\Sigma_\tau$, are determined
by the full type syntax in \appref{type-semantics-appendix}.

The canonical definition describes the automaton through semantic steps.
To obtain a finite representation from syntax, we also give an equivalent automaton
construction from type syntax in Appendix~\ref{app:constructor-automaton} and prove its
correspondence with semantics.
The two constructions preserve the same environments, observations, and
offer pairs on their reachable steps. The constructor correspondence is also
mechanized in Lean.

\begin{lemma}[Type--automaton correspondence]
\label{lem:type-automata-bisimulation}
For every well-formed regular type $\tau$, $\Aut_\tau$ is finite.
There is a relation between residuals of $\tau$ and states of $\Aut_\tau$
containing $(\tau,i_\tau)$. At every related pair, semantic steps and automaton
edges match in both directions with the same environment, observation, and
current/future offers. The resulting residual and state remain related.
\end{lemma}

The full statement and proof are given in
Appendix~\ref{app:semantic-automaton-correspondence}.

The conformance relation from
\defref{process-type-conformance} connects implementations to the same finite
representation.

\begin{lemma}[Process--automaton correspondence]
\label{lem:process-automaton-correspondence}
Suppose $P\models\tau$ and $\tau$ is well formed. There is a relation
between process residuals and states of $\Aut_\tau$ containing $(P,i_\tau)$.
At every related pair, each process step has a matching automaton edge with
the same erased environment and observation, and the resulting process and
state remain related. The process has a step under $E$ exactly when the
automaton state has an edge under $\erase(E)$.
\end{lemma}

The full statement and proof are given in
Appendix~\ref{app:process-automaton-correspondence}.
For a fixed composition, \corref{compositional-conformance} establishes the
required composite conformance from the component checks.

\subsubsection{The Finite Compatibility Game}
\label{subsec:tc-algorithm}
\label{subsec:tc-product}
\label{subsec:tc-arena}
\label{subsec:tc-winning}
\label{subsec:tc-check}

Let $\Aut_I=\Aut(\Theta_0)=(S,\Sigma_I,\Delta_I,q_0,F_I)$ be the automaton
of the fixed composition.
At a state \(q\), Player~0 first chooses a well-formed boundary environment
\(e\).  Player~1 then chooses any product edge carrying \(e\), thereby
resolving all internal and type nondeterminism:
\[
  q^P\xrightarrow{e}q^e\xrightarrow{\alpha}q'^P
  \qquad\text{when}\qquad
  \atrans{q}{E_C}{E_F}{\alpha}{O_C}{O_F}{q'}\in\Delta_I,
  \quad e=(E_C,E_F).
\]
The two half-moves contribute the single log entry
\(\langle e,\alpha\rangle\).  If the chosen environment labels no outgoing
edge, the play enters an absorbing failure sink whose empty-environment
\(\epsilon\)-loop keeps the play and its log infinite.

To enforce honesty, the game remembers the preceding future set until a
clock tick resets it. A boundary choice that violates the same-cycle promise
enters the losing sink. Finite monitors track clock ticks, fair-branch
outcomes, and the progress markers in \(T\). Their acceptance condition
requires clock progress on every play and component progress on every fair
play, as specified by \(\Psi_T\).

The type automaton, boundary environment set, and monitor memories are finite.
Their product therefore gives a finite \termref{def:rabin-game}{Rabin game}
\(G_{I,T}(\tau_1,\ldots,\tau_n)\) with initial vertex \(s_0\).
The full monitor construction is given in Appendix~\ref{sec:tc-game-appendix}.

\subsubsection{Winning, Compatibility, and Consequences}
\label{subsec:tc-proofs}

\begin{definition}[Type compatibility]
\deflabel{type-compatibility}
The types \(\tau_1,\ldots,\tau_n\) are compatible with respect to \(I\) and
\(T\) when Player~0 wins the finite game from its initial vertex:
\[
\TypeCompat_{I,T}(\tau_1,\ldots,\tau_n)
\triangleq
s_0\in\TypeWinRegion\bigl(G_{I,T}(\tau_1,\ldots,\tau_n)\bigr).
\]
\end{definition}

The game definition has the intended semantic reading.

\begin{theorem}[Correctness of the game reduction]
\label{thm:type-compatibility-soundness-completeness}
For the fixed valid composition $\Theta_0$,
\[
\TypeCompat_{I,T}(\tau_1,\ldots,\tau_n)
\Longleftrightarrow
\exists\Strategy_\Theta.\
\forall\ell\in\Logs^\omega_{\Strategy_\Theta}(\Theta_0).\
\Psi_T(\ell).
\]
\end{theorem}

The full statement and proof are given in Appendix~\ref{sec:tc-game-appendix}.

The implementation result below uses conformance of the fixed process
composition to the fixed type composition. By
\corref{compositional-conformance}, checking each component against its type
is sufficient to establish this premise, for any fixed nesting of the
parallel compositions.

\begin{theorem}[Soundness of type compatibility]
\label{thm:type-compatibility-process-soundness}
Let $Q_0$ be a process composition and $\Theta_0$ a valid type composition
for it. If $Q_0\models\Theta_0$ and the component types are compatible, then
$Q_0$ is \termref{def:comm-safety}{communication safe}.
\end{theorem}

The full statement and proof are given in
Appendix~\ref{app:type-compatibility-process-soundness}.

\begin{theorem}[Decidability and complexity of compatibility]
\thmlabel{type-compatibility-decidable}
Suppose the monitored game has \(N\) vertices, \(M\) edges, and
\(K=1+2|\Etas(\Theta_0)|\) Rabin pairs.  Then \(\TypeCompat_{I,T}(\tau_1,\ldots,\tau_n)\) is decidable in
\(\widetilde{O}(MN(K!)^{1+o(1)})\) time.
\end{theorem}

The proof is given in Appendix~\ref{app:type-compatibility-complexity}.

%% file: sections/type-optimality.tex
\subsection{Optimality}
\seclabel{type-optimality}

\codename{} allows both inference and checking of types.
Checking is useful when the designer has a candidate type in mind,
while inference is useful when the designer wants to avoid the burden of
writing a type. In either case, if the compatibility check fails,
the designer might search for a different valid type assignment.
Our optimality result shows that this search is unnecessary for types inferred
by \codename{}: if any tuple of types accepted by the checker establishes
\termref{def:compatibility}{compatibility}, then the tuple of inferred types does
as well.

This guarantee is relative to the precision of the checker.
Our checker is \termref{def:data-oblivious-type-checking}{data-oblivious}: it records the possible
communication actions of each process but does not use invariants about data
values to eliminate branches. The following example illustrates the resulting
loss of precision. Let $P$ be a conditional that sends $m$ when $b$ is true
and sends $n$ otherwise, then terminates.
Suppose that a value invariant establishes that $b$ is always true whenever
this conditional is reached. Then $P$ can only send $m$, and
$\tau=\tkw{skip}\tseq m\tsym{!}$ is a valid semantic type for $P$: that is,
$P\models\tau$. The initial $\tkw{skip}$ matches the conditional's silent
selection step. This precision can affect compatibility with a peer that
does not support $n$.

Establishing $P \models\tkw{skip}\tseq m\tsym{!}$, however, requires reasoning about the value of
$b$. The \termref{def:data-oblivious-type-checking}{data-oblivious} checker does not perform such reasoning and must retain
both branches, yielding the inferred type $\mathsf{abs}(P)=m\tsym{!}\tplus n\tsym{!}.$
Thus, a value-sensitive type may establish \termref{def:compatibility}{compatibility} when the inferred
\termref{def:data-oblivious-type-checking}{data-oblivious} type cannot. Our result does not claim optimality with respect
to such value-sensitive types. Rather, it establishes optimality among exactly
the types admitted by the \termref{def:data-oblivious-type-checking}{data-oblivious} judgment $\checkdo{P}{\tau}$.

We now make this claim precise. By \lemref{data-abs-sound}, $P \models \mathsf{abs}(P)$,
so the inferred type is sound for $P$. Moreover, by definition of type
checking, $\checkdo{P}{\tau}$ implies that $\mathsf{abs}(P)\sqsubseteq\tau$.
Thus, every type accepted by the checker is a further abstraction of the
inferred type. Consequently, $\mathsf{abs}(P)$ is the most precise type
available within our \termref{def:data-oblivious-abstraction}{data-oblivious abstraction}.

To lift this per-module precision result to \termref{def:compatibility}{compatibility},
we need one monotonicity property. Since checked types may be coarser than inferred types,
we show that refining every component of a compatible tuple preserves
\termref{def:compatibility}{compatibility}.

\begin{restatable}[Compatibility is monotone under type refinement]{lemma}{TypeCompatRefinementMonotone}
\lemlabel{type-compat-refinement-monotone}
Consider two valid type compositions for the same system.
If $\tau_i\sqsubseteq\tau_i'$ for every $i$, then
\[
\TypeCompat_{I,T}(\tau_1',\ldots,\tau_n')
\Longrightarrow
\TypeCompat_{I,T}(\tau_1,\ldots,\tau_n).
\]
\end{restatable}

A coarser type admits at least all the behaviours admitted by a more precise
one. Hence, if an environment strategy wins against the coarser tuple, the
strategy transfers to the refined tuple by maintaining a matching abstract
history. The full proof in \appref{type-optimality} establishes the lemma
as the type-to-type instance of the strategy-transfer principle used in
\lemref{strategy-transfer}.

We can now state the optimality result.

\Needspace{8\baselineskip}
\begin{restatable}[Optimality of inferred types]{theorem}{TypeInferenceOptimality}
\thmlabel{type-inference-optimality}
Fix a valid process composition of $P_1,\ldots,P_n$.
The inferred types are compatible exactly when some valid type composition
for it is compatible and all its component types are accepted by the checker:
\begin{align*}
\begin{array}{rcl}
\TypeCompat_{I,T}
  (\mathsf{abs}(P_1),\ldots,\mathsf{abs}(P_n))
  &
  \iff
  &
\exists \tau_1,\ldots,\tau_n.\
\left(\bigwedge_i \checkdo{P_i}{\tau_i}\right)
\wedge
\TypeCompat_{I,T}(\tau_1,\ldots,\tau_n).
\end{array}
\end{align*}
\end{restatable}

For the forward direction, choose $\tau_i=\mathsf{abs}(P_i)$ for every $i$;
these types are accepted by reflexivity of $\sqsubseteq$. For the backward
direction, each $\checkdo{P_i}{\tau_i}$ gives
$\mathsf{abs}(P_i)\sqsubseteq\tau_i$.
\termref{def:compatibility}{Compatibility} of the checked tuple implies compatibility of the
inferred tuple by \lemref{type-compat-refinement-monotone}.
The full proof is given in \appref{type-optimality}.
The contrapositive gives the operational consequence of optimality.

\begin{restatable}[Failure rules out data-oblivious proofs]{corollary}{TypeInferenceFailure}
\corlabel{type-inference-failure}
For a valid process composition, if the inferred types are not compatible,
then there is no valid type composition for it that is both compatible and
accepted by the checker for every component.
\end{restatable}

Therefore, when \termref{def:compatibility}{compatibility} fails for the inferred types, choosing different
types accepted by the \termref{def:data-oblivious-type-checking}{data-oblivious} checker cannot make the check succeed.
One must instead change the processes or their composition, or use
value-sensitive reasoning beyond the abstraction considered here.

%% file: sections/summary-types.tex
\subsection{Summary Types}
\seclabel{summary-types}

When an open composition is compatible, we can summarize its boundary behaviour in a single \termref{def:summary-type}{summary type}. We construct this summary by composing the component type automata on their shared interfaces and hiding the internal message observations. The resulting automaton retains the boundary observations and the current and future offers needed for subsequent composition.

However, the resulting summary automaton may not have a natural representation in the type syntax. Consider two types composed on $\{m_1\}$:
$\sigma_1=\toplus(\tkw{skip},m_1\tsym{!}:\tau_1,m_2\tsym{!}:\tau_2)$
and
$\sigma_2=\toplus(\tkw{skip},m_1\tsym{?}:\tkw{skip}).$
Once the offers are resolved, the composition can synchronize on $m_1$ and select the continuation $\tau_1$, even when the boundary does not offer $m_2\tsym{?}$. Hiding this internal exchange leaves a silent step. Alternatively, if $m_2\tsym{?}$ is immediately available, the composition can send $m_2$ and select $\tau_2$, while synchronization on $m_1$ remains possible. Thus, the summary must represent a silent step that competes with an enabled boundary communication, a behaviour that does not have a natural representation in the current type syntax.
In these cases, the checker can use the summary automaton directly to check compatibility with additional components.

%% file: sections/evaluation.tex
\section{Evaluation}
\seclabel{evaluation}

We evaluate \codename{} through five research questions.

\noindent \textbf{RQ1: Utility.} Can incremental \termref{def:compatibility}{compatibility checking} expose communication errors during design?\\
\textbf{RQ2: Expressiveness.} Can \codename{} express realistic \termref{def:rtl}{RTL} protocols and prove \termref{def:compatibility}{compatibility}?\\
\textbf{RQ3: Checking time.} How does \codename{} \termref{def:compatibility}{compatibility} checking compare with SystemVerilog post-hoc verification?\\
\textbf{RQ4: Circuit overhead.} What area and timing overhead does \codename{} introduce after synthesis?\\
\textbf{RQ5: Annotation burden.} How much code and annotation do \codename{} designs need compared with SV designs and their verification harnesses?

We answer these questions with a prototype built by extending the OCaml
implementation of the \anvil{} compiler~\cite{anvilhdl2026}. The extension adds
the \termref{def:offer}{offer constructs} \prog{offer} and \prog{xoffer}, \termref{def:fairness}{fairness} annotations on
\prog{if}-expressions and \prog{xoffer}, and the \termref{def:type-system}{type system} with its
automata-based \termref{def:compatibility}{compatibility} checker. The backend also adds
ready-valid handshakes, where the receiver asserts ready first, alongside
\anvil{}'s valid-ack handshakes, where the sender asserts valid first.
We reduce code generation overhead by eliminating unnecessary control
state in the generated RTL (\appref{appendix-evaluation}).

\subsection{Utility on Real-World Communication Errors}
\seclabel{evaluation-utility}
An engineer can use \codename{} in two complementary ways. With an explicit
boundary type, the compiler checks both the implementation and the components
already present against the documented protocol. Without such a type, inference
summarizes the implementation's observable communication behaviour. In either
case, checking is incremental: an incomplete subsystem can be checked before
its environment is fully implemented.

To evaluate this workflow, we reproduced communication errors and documented
interface constraints from open-source \termref{def:rtl}{RTL}. We focus here on three AXI cases:
the filed PULPino handshake-dependency issue, an address/data ordering dependency in
PULP AXI's test infrastructure, and the corresponding dependency in its
demultiplexer~\cite{pulpino182,pulpaxi-changelog,axidemux-src}. The complete
set of cases is given in \tabref{utility-cases}.

Each case retains upstream RTL at the boundary relevant to its ordering
constraint. In the table, Filed denotes an upstream issue; Repository fix
reproduces behaviour corrected upstream. Documented uses a stated component
constraint, while Class applies a documented error pattern to the component's
boundary.

\begin{table}[t]
\centering
\captionsetup{skip=3pt}
\scriptsize
\caption{Completed AXI communication-error case studies. The setting shows whether
the mismatch is visible pairwise or requires an open environment;
\badmark{} means \codename{} reports the reproduced composition as incompatible.}
\label{tab:utility-cases}
\renewcommand{\arraystretch}{1.12}
\setlength{\tabcolsep}{2pt}
\begin{tabular}{@{}>{\raggedright\arraybackslash}p{2.0cm} >{\raggedright\arraybackslash}p{1.55cm} >{\raggedright\arraybackslash}p{1.1cm} >{\raggedright\arraybackslash}p{6.9cm} c@{}}
\toprule
Component & Source & Setting & Reproduced mismatch & \shortstack{Compat.\\result} \\
\midrule
AXI debug BIU~\cite{pulpino182}
  & Filed & Pairwise
  & The sender waits for address-channel readiness before offering write data, coupling two independent channels.
  & \badmark \\
AXI test infrastructure~\cite{pulpaxi-changelog}
  & Repository fix & Pairwise
  & The implementation requires the address before write data, although the protocol allows either order or both together.
  & \badmark \\
AXI demux~\cite{pulpaxi-changelog,axidemux-src}
  & Documented & Open
  & Routing and buffering restrict write-channel message order.
  & \badmark \\
AXI per-ID monitor~\cite{pulp-axi}
  & Class & Open
  & An ordering requirement for each transaction ID conflicts with interleavings of independent channel events.
  & \badmark \\
AXI demux/mux pipeline~\cite{axidemux-src,aximux-src}
  & Class & Open
  & Composing router stages imposes a message order that no individual interface specifies.
  & \badmark \\
AXI mux~\cite{aximux-src}
  & Documented & Open
  & Arbitration restricts the order of clients' write-channel messages.
  & \badmark \\
\bottomrule
\end{tabular}
\end{table}

All three cases expose the same underlying mismatch. An AXI write transaction
has an address channel $\mathit{AW}$ and a data channel $\mathit{W}$ whose
handshakes are independent~\cite{AMBAAXIProtocol2023}. A peer may therefore
produce either $\mathit{AW}?;\mathit{W}?$ or
$\mathit{W}?;\mathit{AW}?$. The implementations we studied accept only one
order. In the router, this serialization is not a single local mistake: it
emerges from the composition of routing logic and internal FIFOs.

Incremental checking localizes this mismatch before the complete router is
assembled. We modelled a client that produces $\mathit{W}$ before
$\mathit{AW}$ and composed the router pipeline one participant at a time.
\codename{} reported incompatibility after only 4 of the 10 participants were
present: the client and 3 of the router's 5 pipeline components. This partial
composition has no winning boundary strategy, so adding the remaining
participants cannot make it compatible.

\subsection{Expressiveness}
Our benchmarks are end-to-end systems adapted from open-source IPs: PULP common
cells~\cite{commoncellsgithub}, PULP AXI Lite routers~\cite{pulp-axi}, a PULP
ACE cache-coherence controller~\cite{pulp-ace}, and a Chip Alliance TileLink
component~\cite{chipalliance-tilelink}. All are SystemVerilog except TileLink,
which is Chisel. For each, we implemented a \codename{} version and checked
\termref{def:compatibility}{compatibility}. Since these systems are intended to
be \termref{def:comm-safety}{communication safe}, they test whether the
\termref{def:type-system}{type system} can describe realistic protocols and
establish their compatibility. During implementation, these types also served
as executable documentation and helped maintain
\termref{def:process-type-conformance}{conformance}.

The benchmarks range from FIFOs, arbiters, stream FIFOs, and spill registers
to larger AXI Lite designs. The AXI Demux and Mux require less direct type
abstractions.

\begin{table*}[t]
\centering
\captionsetup{skip=3pt}
\scriptsize
\caption{Compatibility-checking time and circuit synthesis results.
Part.\ is the number of participants; Ratio is \codename{} area divided by SV
area; Cycle Lat.\ Overhead counts additional clock cycles relative to SV.
\okmark{} denotes a successful proof/check or a met timing target;
\unkmark{} denotes an unproved property or timeout.
Overall gives success counts, geometric-mean areas, and mean cycle-latency overhead.}
\label{tab:eval-results}
\label{tab:compatibility-runtime}
\label{tab:synthesis-summary}
\renewcommand{\arraystretch}{1.05}
\begin{minipage}[t]{0.50\textwidth}
\centering
\setlength{\tabcolsep}{1.4pt}
\caption*{Compatibility checking}
\begin{tabular}{l c *{2}{>{\columncolor{svbg}}c} r *{3}{>{\columncolor{pactbg}}c}}
\toprule
\multirow{2}{*}{Benchmark} & \multirow{2}{*}{Part.}
  & \multicolumn{3}{c}{\svhdr{SV Model Checking}}
  & \multicolumn{3}{c}{\pacthdr{\codename{}}} \\
\cmidrule(lr){3-5}\cmidrule(lr){6-8}
& & Safety & Live & Time (s) & Result & 1 thr. (s) & 8 thr. (s) \\
\midrule
FIFO         & 3  & \okmark  & \okmark  & 3354.69 & \cellcolor{pactbg}\okmark & \cellcolor{pactbg}0.02 & \cellcolor{pactbg}0.03 \\
Stream FIFO  & 3  & \unkmark & \unkmark & 7200    & \cellcolor{pactbg}\okmark & \cellcolor{pactbg}0.02 & \cellcolor{pactbg}0.03 \\
Spill reg.   & 3  & \unkmark & \unkmark & 7200    & \cellcolor{pactbg}\okmark & \cellcolor{pactbg}0.02 & \cellcolor{pactbg}0.03 \\
Arbiter      & 6  & \okmark  & \okmark  & 3306.85 & \cellcolor{pactbg}\okmark & \cellcolor{pactbg}0.11 & \cellcolor{pactbg}0.06 \\
ACE          & 9  & \unkmark & \unkmark & 7200    & \cellcolor{pactbg}\okmark & \cellcolor{pactbg}130.69 & \cellcolor{pactbg}23.45 \\
TileLink     & 10 & \okmark  & \okmark  & 2839.34 & \cellcolor{pactbg}\okmark & \cellcolor{pactbg}9.44 & \cellcolor{pactbg}1.87 \\
AXI demux R  & 8  & \unkmark & \unkmark & 313.82  & \cellcolor{pactbg}\okmark & \cellcolor{pactbg}32.49 & \cellcolor{pactbg}7.53 \\
AXI demux W  & 10 & \unkmark & \unkmark & 515.96  & \cellcolor{pactbg}\okmark & \cellcolor{pactbg}627.15 & \cellcolor{pactbg}96.51 \\
AXI mux R    & 9  & \unkmark & \unkmark & 7200    & \cellcolor{pactbg}\okmark & \cellcolor{pactbg}112.08 & \cellcolor{pactbg}17.90 \\
AXI mux W    & 11 & \unkmark & \unkmark & 7200    & \cellcolor{pactbg}\okmark & \cellcolor{pactbg}1243.97 & \cellcolor{pactbg}181.06 \\
\midrule
Overall      & -- & 3/10 & 3/10 & -- & \cellcolor{pactbg}10/10 & \cellcolor{pactbg}-- & \cellcolor{pactbg}-- \\
\bottomrule
\end{tabular}
\end{minipage}
\hfill
\begin{minipage}[t]{0.47\textwidth}
\centering
\setlength{\tabcolsep}{1.5pt}
\caption*{Synthesis results}
\begin{tabular}{l *{2}{>{\columncolor{pactbg}}r} r c c}
\toprule
\multirow{2}{*}{Benchmark} & \multicolumn{3}{c}{Area ($\mu\mathrm{m}^2$)} & \multirow{2}{*}{\shortstack{Meets\\Timing}} & \multirow{2}{*}{\shortstack{Cycle Lat.\\Overhead}} \\
\cmidrule(lr){2-4}
& \pacthdr{\codename{}} & \svhdr{SV} & Ratio & & \\
\midrule
FIFO         & \cellcolor{pactbg}11,659 & \cellcolor{svbg}11,616 & 1.00$\times$ & \okmark & 0 \\
Stream FIFO  & \cellcolor{pactbg}11,671 & \cellcolor{svbg}11,486 & 1.02$\times$ & \okmark & 0 \\
Spill reg.   & \cellcolor{pactbg}3,363  & \cellcolor{svbg}2,770  & 1.21$\times$ & \okmark & 0 \\
Arbiter      & \cellcolor{pactbg}1,045  & \cellcolor{svbg}1,386  & 0.75$\times$ & \okmark & 0 \\
ACE          & \cellcolor{pactbg}4,792  & \cellcolor{svbg}8,423  & 0.57$\times$ & \okmark & 0 \\
TileLink     & \cellcolor{pactbg}1,942  & \cellcolor{svbg}2,280  & 0.85$\times$ & \okmark & 0 \\
AXI demux R  & \cellcolor{pactbg}2,948  & \cellcolor{svbg}1,978  & 1.49$\times$ & \okmark & 0 \\
AXI demux W  & \cellcolor{pactbg}4,462  & \cellcolor{svbg}2,726  & 1.64$\times$ & \okmark & 0 \\
AXI mux R    & \cellcolor{pactbg}3,490  & \cellcolor{svbg}2,603  & 1.34$\times$ & \okmark & 0 \\
AXI mux W    & \cellcolor{pactbg}6,036  & \cellcolor{svbg}4,217  & 1.43$\times$ & \okmark & 0 \\
\midrule
Overall      & \cellcolor{pactbg}4,065  & \cellcolor{svbg}3,771  & 1.08$\times$ & 10/10 & 0 \\
\bottomrule
\end{tabular}
\end{minipage}
\end{table*}

\subsection{Compatibility Checking}
The \codename{} checking times in \tabref{eval-results} are medians of five
runs on the full type composition, using one and eight threads. The SystemVerilog
baseline uses SymbiYosys proof mode with a maximum k-induction depth of 500.

The SV harness checks \termref{def:sync-safety}{synchronization safety} and
encodes bounded response, such as receiving a response within $N$ cycles,
as a safety assertion. The response bound $N$ is distinct from the induction
depth: a successful induction proves the assertion for all executions, rather
than only for the first 500 cycles. We use bounded-response assertions because
the Yosys \prog{read\_slang} frontend used for these designs does not support
the \prog{eventually} constructs needed to express recurring progress directly.
The SV \emph{Live} column therefore reports bounded response, whereas
\codename{} checks the recurring progress condition of \secref{pact-safety}.
Safety and bounded response are
checked separately for the two result columns; the reported time is for an
additional run containing all assertions.

SV proves all assertions for three benchmarks: the FIFO, the arbiter,
and TileLink, each taking 47 to 56 minutes. For the other seven, induction
leaves the properties unproved at depth 500 or the run times out after two
hours.

\codename{} proves \termref{def:compatibility}{compatibility} for all ten benchmarks. With one thread,
six of ten checks finish within 33 seconds, and AXI mux read and ACE within
about two minutes. AXI demux write takes 10 minutes, and the largest case, AXI
mux write, takes 21 minutes. With eight threads, these four checks run 5 to 7
times faster, and every check finishes within about three minutes.

Data-obliviousness is an important reason for this difference. The SV model
checker operates on the RTL state space, which includes payload data and local
computations. \codename{} exposes communication through message passing and
offers, and its inferred types erase data while retaining communication-relevant
control and timing. Model checking an explicit data abstraction could also
reduce the SV state space, but would additionally require establishing that
the RTL implementation refines that abstraction.

\subsection{Circuit Overhead}
To measure circuit overhead, we synthesized each \codename{} component and its
SystemVerilog counterpart with Yosys at a 10ns clock period (100MHz), using the
SkyWater 130nm high-density standard-cell library. A design meets timing if its worst-case setup slack at this period is positive.

\myparagraph{Timing and latency}
All designs meet the timing target, and \codename{} adds no clock-cycle latency.
These results support our design goal of preserving cycle-accurate behaviour without adding cycles to enforce \termref{def:comm-safety}{communication safety}.

\myparagraph{Area}
Across ten benchmarks, \codename{} incurs a geometric-mean area overhead of 7.8\% over SV.
This overhead comes from control state introduced during code generation.
Unlike the SV designs, which implement much of their behaviour using combinational logic, the \codename{} designs are written as communicating processes that compile into state machines spanning multiple clock cycles.

For example, in the AXI routers, the \codename{} designs declare 25--50 flip-flops, comparable to the 26--60 flip-flops in the SV routers.
Code generation adds another 21--43 control flip-flops per router, of which 15--27 implement one-cycle delays at clock-cycle boundaries in process threads.
In contrast, the SV routers have no corresponding control state machines: apart from their transaction FIFOs, they retain only a one-bit lock in the write routers and arbiter priority in the muxes.

\subsection{Annotation Burden}
\seclabel{evaluation-burden}
\begin{wraptable}{r}{0.41\textwidth}
\centering
\captionsetup{skip=3pt}
\scriptsize
\caption{Benchmark size in LOC.}
\label{tab:benchmark-size}
\renewcommand{\arraystretch}{1.15}
\setlength{\tabcolsep}{1.75pt}
\begin{tabular}{l *{3}{>{\columncolor{svbg}}r >{\columncolor{pactbg}}r}}
\toprule
\multirow{2}{*}{Benchmark} & \multicolumn{2}{c}{Design} & \multicolumn{2}{c}{Full} & \multicolumn{2}{c}{Verification} \\
\cmidrule(lr){2-3}\cmidrule(lr){4-5}\cmidrule(lr){6-7}
& \svhdr{SV} & \pacthdr{\codename{}} & \svhdr{SV} & \pacthdr{\codename{}} & \svhdr{SV} & \pacthdr{\codename{}} \\
\midrule
FIFO         & 99 & 49 & 345 & 484 & 152 & 19 \\
Stream FIFO  & 67 & 38 & 221 & 183 & 100 & 3 \\
Spill reg.   & 91 & 83 & 213 & 182 & 82 & 7 \\
Arbiter      & 244 & 31 & 375 & 128 & 63 & 9 \\
ACE          & 557 & 65 & 1,807 & 402 & 350 & 39 \\
TileLink   & 437 & 93 & 1,198 & 355 & 174 & 58 \\
AXI demux R  & 316 & 205 & 794 & 416 & 117 & 62 \\
AXI demux W  & 384 & 278 & 907 & 528 & 145 & 79 \\
AXI mux R    & 554 & 238 & 1,045 & 445 & 117 & 60 \\
AXI mux W    & 621 & 300 & 1,158 & 567 & 145 & 78 \\
\midrule
Total        & 3,370 & 1,380 & 8,063 & 3,690 & 1,445 & 414 \\
\bottomrule
\end{tabular}
\end{wraptable}
\tabref{benchmark-size} compares benchmark sizes in lines of code (LOC), excluding blank lines and comments.
\emph{Design} counts the synthesized components: the open-source SV IPs and their cycle-accurate \codename{} translations.
\emph{Full} counts all sources compiled for cycle-level equivalence checking: the main component and those it communicates with.
\emph{Verification} counts the code needed to prove \termref{def:comm-safety}{communication safety}.
For SV, this includes communication-safety assertions and the signals exposed by testbench stubs.
\codename{} infers types from implementations and requires no type annotations;
we declare types for all participants for comparison.
For TileLink, we count the Chisel sources for its crossbar and arbiter,
excluding channel bundles, rather than generated Verilog.

Every \codename{} design is smaller than its SV baseline, totaling 1,380 LOC compared to 3,370 LOC across ten benchmarks.
Even with explicitly declared types for all participants, \codename{} requires only 414 LOC for verification, compared to 1,445 LOC for the SV harness.

%% file: sections/related.tex
\section{Related Work}
\seclabel{related}

\myparagraph{HDLs for Safety}{%
The common \termref{def:rtl}{RTL} design flow uses low-level
\termref{def:hdl}{HDLs} such as SystemVerilog and VHDL before applying
post-design verification.
Prior safety-focused \termref{def:hdl}{HDLs} improve this baseline, but mostly
target low-level properties of wire use or register assignments, such as
combinational loop freedom~\cite{christensenWireSortsLanguage2021,basejump},
structural hazard freedom~\cite{filamentHDL,jangModularHardwareDesign2024},
information flow control~\cite{liCaissonHardwareDescription2011,zhangHardwareDesignLanguage2015},
and timing safety~\cite{anvilhdl2026}.
Some also focus on specific use cases such as pipelined CPU
designs~\cite{zagieboyloPDLHighlevelHardware2022,skarmanSpadeExpressionBasedHDL2023}.
Others, including high-level synthesis languages~\cite{BluespecSystemVerilogReference2008,Systemcorg},
provide safety by abstracting away details such as clock cycles and register assignments, or by
relying on latency-insensitive interfaces~\cite{hanShakeFlowFunctionalHardware2023,jangModularHardwareDesign2024}
and run-time scheduling logic
that incur significant overhead~\cite{BluespecSystemVerilogReference2008}.
Filament~\cite{filamentHDL} has a timing type system similar to Anvil's,
with static timelines for operations that rule out synchronization hazards
by restricting expressiveness, making control flow awkward to express.
\codename{} instead provides a richer communication model than \anvil{}
and a type system that ensures \termref{def:comm-safety}{communication safety}.

}

\myparagraph{RTL Verification}{%
Post-design RTL verification checks completed RTL against user-written
properties, often expressed in temporal logic or SystemVerilog assertions
(SVA)~\cite{pnueliTemporalLogicPrograms1977,antoniocauIntervalTemporalLogic2023,moszkowskiExecutingTemporalLogic2000,moszkowskiTemporalLogicMultilevel1985,18002017IEEEStandard2018}.
These checks use testing, model checking, fuzzing, or static
analysis~\cite{Cocotb,IEEEStandardUniversal2020,YosysHQYosys2025,JasperFormalVerification,witharanaSurveyAssertionbasedHardware2022,oleksenkoEnterExitPage2025,rfuzz,fuzz-hw-sw,efficient-fuzz,milesan,cweat,qihe,chisa}.
They are general, but testing and bounded methods can miss bugs, while
unbounded model checking faces state explosion. \codename{} instead checks
communication protocols in the \termref{def:hdl}{HDL}, before lowering to RTL,
and targets \termref{def:comm-safety}{communication safety} rather than
arbitrary temporal properties.
}

\myparagraph{Behavioural Types}{%
Graph types~\cite{mullerStaticPredictionParallel2022,rinaldiPipelinesBeyondGraph2024}
support deadlock and cost analysis of asynchronous parallel computations with
futures.
\codename{} instead checks \termref{def:comm-safety}{communication safety} of synchronous, cycle-accurate
protocols. Cost analysis is a possible direction for future work.
Session types describe communication protocols using
top-down~\cite{topdown1,topdown2,topdown3} or
bottom-up~\cite{bottomup2,bottomUpHybrid,bottomupkmc,bottomupsmc} approaches.
Unlike top-down approaches, where global contracts dictate local behaviour, \codename{}
follows the bottom-up approach, prioritizing modularity over efficiency~\cite{topdownBottomUpcomplexity}.
Bottom-up multiparty compatibility, from Synchronous Multiparty Compatibility
(SMC)~\cite{bottomupsmc} to asynchronous, bounded-buffer, and hybrid
variants~\cite{bottomupkmc,multipartysesiontype,bottomUpHybrid}, targets closed
automata systems over channels and lacks discrete time. \codename{} instead
uses a shared clock, cycle-accurate types, and open compatibility via
environment strategies and summary types, enabling incremental composition.
Temporal and timed session types~\cite{tempsessiontypes,timedsession} add time
to mostly asynchronous settings; \codename{} uses one global tick. Hardware
applications remain rare: H-Calculus~\cite{hcalculus} models cycle-level timing,
but targets HLS design-space exploration rather than designer-facing RTL.
}

%% file: sections/conclusions.tex
\section{Conclusions and Future Work}
\seclabel{conclusions}

This paper presented \codename{}, a typed extension of Anvil for specifying and
modularly verifying \termref{def:comm-safety}{communication safety}.
Types in \codename{}  capture clock-cycle latency and control
flow. We characterize these types as automata and use that
characterization to give a sound and complete \termref{def:compatibility}{compatibility}
algorithm.
The same type system also supports incremental design. In the
\termref{def:open}{open setting}, designers can check partial systems before
every component is available, and the checker produces the
\termref{def:summary-type}{summary of the behaviour} needed at the boundary. Finally,
\codename{} infers types optimally within its
\termref{def:data-oblivious-abstraction}{data-oblivious abstraction}: if some
checker-accepted type assignment can establish \termref{def:compatibility}{compatibility}, then
the inferred types can establish \termref{def:compatibility}{compatibility} as
well.
We implemented \codename{} as an extension of the Anvil compiler and evaluated
it by designing and checking open-source hardware IPs.

Our work opens several avenues for future research. On the design side,
incremental checking would benefit from more expressive process and type
representations. One example is value-sensitive refinement, which could recover
precision lost by our data-oblivious abstraction (\secref{type-optimality}). Another is a
finer-grained model of environment strategies, which could make open-system
checking more flexible.
On the practical side, we plan to improve both analysis and compilation.
Symbolic representations could make type-compatibility checking faster. Better
code generation for the Anvil backend could also reduce the area overhead
observed.

%% file: sections/acknowledgement.tex
\begin{acks}
We thank the anonymous reviewers for their feedback.
This research is supported by a Singapore Ministry of Education (MOE)
Tier~2 grant MOE-T2EP20124-0007.
Additionally, Umang Mathur is supported by a Ministry of Defence Temasek PYP grant.
\end{acks}

%% file: sections/data-availability.tex
\section*{Data Availability Statement}
\label{data-availability}

The \codename{} compiler, benchmarks, and Lean 4 mechanization of the proofs are archived on \mbox{Zenodo}~\cite{pactArtifact2026} and available at \url{https://github.com/AnvilHDL}. The archived artefact provides a stable version of the implementation and proofs presented in this paper, while we will continue to actively maintain and improve these artefacts.

\section*{Disclosure of Use of AI in Research}

The authors developed and are responsible for the research ideas, results, and proofs presented in this paper.
AI tools were used to polish the writing and identify potential mistakes.
The design of the compiler extensions and the implementation of the benchmarks were carried out
primarily by the authors, with AI tools providing interactive assistance during debugging.
AI tools also assisted in the Lean 4 formalization, guided by handwritten proofs developed by the authors.
Finally, AI tools were used to help prepare the reproducible artefact and automate benchmark evaluation and data collection.

%% file: sections/appendix.tex
\input{sections/appendix/pact}

\input{sections/appendix/types}

\input{sections/appendix/evaluation}

%% file: sections/appendix/pact.tex
\section{Complete Syntax, Semantics, and Safety of \codename{} Designs}
\applabel{appendix-pact}

We first give the process syntax with explicit message alphabets, followed by
its operational rules and auxiliary functions. We then give the reachability
relation and the environment restriction used for closed systems.

\input{sections/appendix/pact/syntax-of-pact}
\input{sections/appendix/pact/semantics-of-pact}

\input{sections/appendix/pact/defining-communication-safety}

%% file: sections/appendix/pact/syntax-of-pact.tex
\subsection{Complete Syntax of \codename{}}
\label{app:signed-process-syntax}

The process constructors are those of \figref{abstract-syntax}. To retain
message declarations as a process executes, each process also carries a
finite alphabet of polarized message actions:
\[
 P=(\Omega,R),\qquad
 \Omega\subseteq\mathsf{ActM}
   =\{m\psym{!},m\psym{?}\mid m\in\Messages\}.
\]
Here $\Omega$ is the boundary alphabet, denoted by $\Alphabet(P)$, and $R$
is a process body built from the productions for $P$ in
\figref{abstract-syntax}. Each process child in $R$ is itself an annotated
pair. We call such a pair a signed process. These alphabets
retain declarations that an executed residual may no longer use, and each
contains at most one polarity of a message. Examples and rules suppress
annotations when the surrounding process determines them.

\myparagraph{Alphabet well-formedness}
At a parallel composition, the parent exposes only boundary actions, while
the children also expose complementary actions on their shared interface.
Their alphabets must therefore satisfy
\[
\begin{aligned}
\mathsf{ParCompat}_I(\Omega,\Omega_1,\Omega_2)\quad\Longleftrightarrow\quad
 &\restr{\Omega}{I}=\emptyset,
 \quad \restr{\Omega_1}{I}=\overline{\restr{\Omega_2}{I}},\\
 &\{\mathsf{Msg}(a)\mid a\in\restr{\Omega_1}{I}\}
 =\{\mathsf{Msg}(a)\mid a\in\restr{\Omega_2}{I}\}=I,\\
 &\Omega_1\setminus I=\Omega=\Omega_2\setminus I.
\end{aligned}
\]
Here $\Omega_i=\Alphabet(P_i)$ is the alphabet of child $P_i$.
Complementation exchanges $m\psym{!}$ with $m\psym{?}$ and acts pointwise on sets.
Restriction $\restr{A}{I}$ keeps the actions whose message names lie in $I$,
while $A\setminus I$ removes both polarities of those messages. The first
line keeps interface actions internal to the parent and gives them
complementary polarities in the children. The second requires both children
to declare every message in $I$. The last gives both children the complete
parent alphabet outside $I$, since both children receive the complete parent
environment (\appref{aux-appendix}). Every other constructor gives each child the same alphabet as its
parent. In particular, every default and continuation of an offer carries the
complete surrounding alphabet. A process satisfies $\mathsf{AlphaWF}(P)$
when its annotations satisfy these conditions recursively, including the
single-polarity condition.

\myparagraph{Re-signing a residual}
When parallel composition finishes on one side, the surviving process takes
the parent's alphabet. This replacement propagates through the residual,
while preserving the interface actions of nested parallel children:

\[
\begin{aligned}
 \Pi_I(\Omega',\Omega_i)&=\restr{\Omega_i}{I}\cup\Omega',\\
 (\Omega,P_1\pparop_I P_2)[\Omega']
   &=(\Omega',P_1[\Pi_I(\Omega',\Alphabet(P_1))]
        \pparop_I P_2[\Pi_I(\Omega',\Alphabet(P_2))]).
\end{aligned}
\]
Here $P[\Omega']$ replaces the alphabet annotation; it is not substitution
for a process variable. The function $\Pi_I$ combines the incoming alphabet
with a child's existing actions on $I$. For a non-parallel body $R$, the
replacement passes the incoming alphabet directly to its children:

\[
\begin{aligned}
 (\Omega,R)[\Omega']&=(\Omega',R[\Omega'])\quad(R\text{ non-parallel}),\\
 R[\Omega']&=R\quad
   (R\in\{\pkw{skip},X,\pkw{cycle},\pkw{end}_t\}),\\
 (a\ \pkw{in}\ P)[\Omega']&=a\ \pkw{in}\ P[\Omega'],\\
 (P_1\pseq P_2)[\Omega']&=P_1[\Omega']\pseq P_2[\Omega'],\\
 (\pmu X.P)[\Omega']&=\pmu X.P[\Omega'],\\
 (\pkw{let}\ x=e\ \pkw{in}\ P)[\Omega']
   &=\pkw{let}\ x=e\ \pkw{in}\ P[\Omega'],\\
 (\pkw{if}\ e\ \pkw{then}\ P_1\ \pkw{else}\ P_2)[\Omega']
   &=\pkw{if}\ e\ \pkw{then}\ P_1[\Omega']\ \pkw{else}\ P_2[\Omega'],\\
 (\pkw{fair-if}_\eta\ e\ \pkw{then}\ P_1\ \pkw{else}\ P_2)[\Omega']
   &=\pkw{fair-if}_\eta\ e\ \pkw{then}\ P_1[\Omega']\ \pkw{else}\ P_2[\Omega'],\\
 (Z\ B\ \pkw{else}\ P_0)[\Omega']&=Z\ B[\Omega']\ \pkw{else}\ P_0[\Omega'],\\
 B[\Omega']&=\{a_i:P_i[\Omega']\mid a_i:P_i\in B\}.
\end{aligned}
\]
Here $Z$ ranges over $\pkw{xoffer}$ and $\pkw{offer}$. Re-signing changes no
expression, head action, or binder. Hence an old interface head may remain in
the syntax after completion, but an admissible environment can no longer
enable it.

%% file: sections/appendix/pact/semantics-of-pact.tex
\subsection{Complete Operational Semantics of \codename{}}
\applabel{app-pact-semantics}

The rules below complete the process semantics of \secref{pact-semantics},
including the message alphabets omitted there.

\myparagraph{Preliminaries}
Let $\Messages$ be the finite set of all messages, $\Pids$ be the finite set of
progress identifiers, and $\Etas$ be the finite set of fair-choice identifiers.
Every occurrence of $\pkw{end}_t$ in a process is annotated with some
$t\in\Pids$. Similarly, every occurrence of a fair construct, such as
$\pkw{fair-if}_\eta$,
is annotated with some
$\eta\in\Etas$. Distinct occurrences must use distinct annotations.

Let $\PidSet{P}$ and $\EtaSet{P}$ collect the progress and fair-choice
identifiers occurring in $P$. The boundary actions $\BoundaryMsgs(P)$ come
from its declared alphabet; the internal message names $\IntMsgs(P)$ come
from its parallel interfaces. Their definitions, together with the functions
on branch heads and offers, are given in \appref{aux-appendix}.
The polarities $!$ and $?$ denote sending and receiving from the point of
view of $P$.

Transition labels come from $\ActSet{P}$ as defined in
\secref{pact-semantics}. We will continue to omit braces around singleton
labels and use $\epsilon$ for the empty label.

\nparagraph{Environment}
The offers available to the environment have polarities complementary to
the boundary actions of the process:
\[
\BEnv_P = \{m\psym{?}\mid m\psym{!}\in\BoundaryMsgs(P)\}\cup\{m\psym{!}v\mid m\psym{?}\in\BoundaryMsgs(P)\}.
\]
A one-step environment selects immediate and future offers from this set:
\[
\begin{aligned}
E &= (U,F)\quad\text{where}\quad
U\in 2^{\BEnv_P}\cup\{\{\#\}\},\quad
F\in 2^{\Bare(\BEnv_P)},\\
&\qquad U=\{\#\}\Rightarrow F=\emptyset,\quad
U\neq\{\#\}\Rightarrow \Bare(U)\cap F=\emptyset.
\end{aligned}
\]
Here $U$ contains immediate offers and $F$ contains bare future offers.
The function $\Bare$ erases values, so $\Bare(\{m\psym{!}v\})=\{m\psym{!}\}$
and $\Bare(\{m\psym{?}v\})=\{m\psym{?}\}$. When $U\neq\{\#\}$, the immediate and future
sets are disjoint after erasing values. When $U=\{\#\}$, the clock boundary
has arrived and there is no later point in the same cycle. Let $\Imm(E)=U$
and $\Fut(E)=F$ denote the two components. History-sensitive strategies over
execution prefixes are defined in
\secref{pact-safety}.

\nparagraph{Semantic Rules}
The semantic rules for processes are defined below. The auxiliary functions
used in the rules are defined in full afterward.
Let $\Env{P}$ denote this set of environments. Every rule has the premise
$E\in\Env{P}$ for its annotated source $P$; each component step likewise uses an
environment in its child's environment set. When a rule retains a sequence,
parallel composition, or offer, that constructor keeps its source alphabet.
When a rule selects a child, the child keeps its annotation. The displayed
rules omit these environment premises and alphabet annotations.

\begin{center}
\begin{semrulebox}[fontupper=\scriptsize]
\begin{adjustbox}{max width=.99\linewidth}
$\begin{array}{@{}ccc@{}}
\dfrac{m\psym{?}\in\Imm(E) \quad e\Downarrow v}
      {\lstep{E}{\pkw{send}\ m\ e\ \pkw{in}\ P}{m\psym{!}v}{P}}
\rulelabel{L-Send}
&
\dfrac{m\psym{!}v\in\Imm(E)}
      {\lstep{E}{x=\pkw{recv}\ m\ \pkw{in}\ P}{m\psym{?}v}{P[x\mapsto v]}}
\rulelabel{L-Recv}
&
\dfrac{\Imm(E)=\{\#\}}
      {\lstep{E}{\pkw{cycle}}{\#}{\pkw{skip}}}
\rulelabel{L-Cycle}
\\[1.15em]
\dfrac{m\psym{?}\notin\Imm(E) \quad \Imm(E)\neq\{\#\}}
      {\lstep{E}{\pkw{send}\ m\ e\ \pkw{in}\ P}{\epsilon}{\pkw{send}\ m\ e\ \pkw{in}\ P}}
\rulelabel{L-SendWait}
&
\dfrac{\forall v.\ m\psym{!}v\notin\Imm(E) \quad \Imm(E)\neq\{\#\}}
      {\lstep{E}{x=\pkw{recv}\ m\ \pkw{in}\ P}{\epsilon}{x=\pkw{recv}\ m\ \pkw{in}\ P}}
\rulelabel{L-RecvWait}
&
\dfrac{\Imm(E)\neq\{\#\}}
      {\lstep{E}{\pkw{cycle}}{\epsilon}{\pkw{cycle}}}
\rulelabel{L-CycleWait}
\\[1.15em]
\dfrac{\lstep{E}{P_1}{\alpha}{\pkw{skip}}}
      {\lstep{E}{P_1\pseq P_2}{\alpha}{P_2}}
\rulelabel{L-SeqDone}
&
\dfrac{\Imm(E)\neq\{\#\}}
      {\lstep{E}{\pkw{skip}}{\epsilon}{\pkw{skip}}}
\rulelabel{L-Skip}
&
\dfrac{\Imm(E)\neq\{\#\}}
      {\lstep{E}{\pkw{end}_{t}}{\$_t}{\pkw{skip}}}
\rulelabel{L-End}
\\[1.15em]
\multicolumn{3}{c}{
\dfrac{\lstep{E}{P_1}{\alpha}{P_1'} \quad P_1'\neq\pkw{skip}}
      {\lstep{E}{P_1\pseq P_2}{\alpha}{P_1'\pseq P_2}}
\rulelabel{L-Seq}
}
\end{array}$
\end{adjustbox}
\end{semrulebox}
\end{center}

\begin{center}
\begin{semrulebox}[fontupper=\scriptsize]
\begin{adjustbox}{max width=.99\linewidth}
$\begin{array}{@{}ccc@{}}
\dfrac{e\Downarrow\pkw{true} \quad \Imm(E)\neq\{\#\}}
      {\lstep{E}{\pkw{if}\ e\ \pkw{then}\ P_1\ \pkw{else}\ P_2}{\epsilon}{P_1}}
\rulelabel{L-IfTrue}
&
\dfrac{e\Downarrow\pkw{false} \quad \Imm(E)\neq\{\#\}}
      {\lstep{E}{\pkw{if}\ e\ \pkw{then}\ P_1\ \pkw{else}\ P_2}{\epsilon}{P_2}}
\rulelabel{L-IfFalse}
&
\dfrac{e\Downarrow v \quad \Imm(E)\neq\{\#\}}
      {\lstep{E}{\pkw{let}\ x=e\ \pkw{in}\ P}{\epsilon}{P[x\mapsto v]}}
\rulelabel{L-Let}
\\[1.15em]
\dfrac{e\Downarrow\pkw{true} \quad \Imm(E)\neq\{\#\}}
      {\lstep{E}{\pkw{fair-if}_\eta\ e\ \pkw{then}\ P_t\ \pkw{else}\ P_f}
        {\eta_{\mathsf{true}}}{P_t}}
\rulelabel{L-FairTrue}
&
\dfrac{e\Downarrow\pkw{false} \quad \Imm(E)\neq\{\#\}}
      {\lstep{E}{\pkw{fair-if}_\eta\ e\ \pkw{then}\ P_t\ \pkw{else}\ P_f}
        {\eta_{\mathsf{false}}}{P_f}}
\rulelabel{L-FairFalse}
&
\dfrac{\Imm(E)\neq\{\#\}}
      {\lstep{E}{\pmu X.P}{\epsilon}{P[X\mapsto\pmu X.P]}}
\rulelabel{L-Rec}
\end{array}$
\end{adjustbox}
\end{semrulebox}
\end{center}

The sequence rules distinguish a step that finishes the first component
from one that retains it. The wait rules require a non-clock environment;
once the clock is offered, waiting cannot postpone the cycle boundary.

\begin{center}
\begin{semrulebox}[fontupper=\scriptsize]
\begin{adjustbox}{max width=.99\linewidth}
$\begin{array}{@{}c@{}}
\dfrac{\pkw{send}\ m\ e:P_i\in B \quad m\psym{?}\in\Imm(E) \quad e\Downarrow v \quad \Resolved(E,B)}
      {\lstep{E}{\pkw{xoffer}\ B\ \pkw{else}\ P_0}{m\psym{!}v}{P_i}}
\rulelabel{L-XOfferSend}
\\[1.15em]
\dfrac{x=\pkw{recv}\ m:P_i\in B \quad m\psym{!}v\in\Imm(E) \quad \Resolved(E,B)}
      {\lstep{E}{\pkw{xoffer}\ B\ \pkw{else}\ P_0}{m\psym{?}v}{P_i[x\mapsto v]}}
\rulelabel{L-XOfferRecv}
\\[1.15em]
\dfrac{\Need(B)\cap\Bare(\Imm(E))=\emptyset \quad \Resolved(E,B) \quad \Imm(E)\neq\{\#\}}
      {\lstep{E}{\pkw{xoffer}\ B\ \pkw{else}\ P_0}{\epsilon}{P_0}}
\rulelabel{L-XOfferElse}
\qquad
\dfrac{\neg\,\Resolved(E,B) \quad \Imm(E)\neq\{\#\}}
      {\lstep{E}{\pkw{xoffer}\ B\ \pkw{else}\ P_0}{\epsilon}{\pkw{xoffer}\ B\ \pkw{else}\ P_0}}
\rulelabel{L-XOfferWait}
\end{array}$
\end{adjustbox}
\end{semrulebox}
\end{center}

\nparagraph{Offers}
The exclusive-offer rules use $\Resolved(E,B)$ from \appref{aux-appendix}.
For conjunctive offers, the rules below use the set of immediately enabled
heads to construct the combined label and residual.
The continuations share no messages: the ownership
requirements of \appref{type-semantics-appendix} separate their message names.
They therefore run in parallel with empty interfaces.

\begin{center}
\begin{semrulebox}[fontupper=\scriptsize]
\begin{adjustbox}{max width=.99\linewidth}
$\begin{array}{@{}c@{}}
\dfrac{\begin{gathered}
 K=K(E,B)\neq\emptyset\\
 \forall i\in K.\ \lstep{E}{a_i\ \pkw{in}\ P_i}{\beta_i(v_i)}{P_i^{\mathbf v}}
 \end{gathered}}
      {\lstep{E}{\pkw{offer}\ B\ \pkw{else}\ P_0}{\beta_K(\mathbf v)}{C^P_{B,K,\mathbf v}}}
\rulelabel{L-OfferSome}
\\[1.15em]
\dfrac{\Need(B)\cap\Bare(\Imm(E))=\emptyset \quad \Resolved(E,B) \quad \Imm(E)\neq\{\#\}}
      {\lstep{E}{\pkw{offer}\ B\ \pkw{else}\ P_0}{\epsilon}{P_0}}
\rulelabel{L-OfferElse}
\qquad
\dfrac{\Need(B)\cap\Bare(\Imm(E))=\emptyset \quad \neg\,\Resolved(E,B) \quad \Imm(E)\neq\{\#\}}
      {\lstep{E}{\pkw{offer}\ B\ \pkw{else}\ P_0}{\epsilon}{\pkw{offer}\ B\ \pkw{else}\ P_0}}
\rulelabel{L-OfferWait}
\end{array}$
\end{adjustbox}
\end{semrulebox}
\end{center}

For $B=\{a_i:P_i\mid i\in J_B\}$, where $J_B=\{1,\ldots,n\}$, define
\[
 K(E,B)=\{i\in J_B\mid\Need(a_i)\in\Bare(\Imm(E))\}.
\]
For each $i\in K$, a send evaluates its expression to $v_i$, and a receive
chooses a payload $v_i$ supplied by $\Imm(E)$. Let
$\mathbf v=(v_i)_{i\in K}$ denote these payloads. The prefix step in \rulelabel{L-OfferSome}
checks this communication, whose label and continuation are
\[
\begin{aligned}
 \beta_i(v_i)&=\begin{cases}
                 m\psym{!}v_i,&a_i=\pkw{send}\ m\ e,\\
                 m\psym{?}v_i,&a_i=(x=\pkw{recv}\ m),
                \end{cases}\\
 P_i^{\mathbf v}&=\begin{cases}
                 P_i,&a_i=\pkw{send}\ m\ e,\\
                 P_i[x\mapsto v_i],&a_i=(x=\pkw{recv}\ m).
                \end{cases}
\end{aligned}
\]
The combined label is $\beta_K(\mathbf v)=\{\beta_i(v_i)\mid i\in K\}$.
The target retains the continuation of every branch that communicated:
\[
\begin{aligned}
 R_i^{K,\mathbf v}&=\begin{cases}
 P_i^{\mathbf v}[\Omega],&i\in K,\\
 \pkw{xoffer}\ a_i:P_i[\Omega]\ \pkw{else}\ \pkw{skip},&i\notin K,
 \end{cases}\\
 C^P_{B,K,\mathbf v}&=R_1^{K,\mathbf v}\pparop_{\emptyset}\cdots
                       \pparop_{\emptyset}R_n^{K,\mathbf v}.
\end{aligned}
\]
Here $\Omega$ is the source alphabet. The product is right-associated, and a
single factor is returned unchanged. Every generated node carries $\Omega$,
and the continuations are re-signed by $\Omega$. The continuations take no step
in this round. The other heads remain available through their singleton
exclusive offers. If no branch is immediately enabled, resolution determines
whether the original offer takes its default or waits.

\begin{center}
\begin{semrulebox}[fontupper=\scriptsize]
\begin{adjustbox}{max width=.97\linewidth}
$\dfrac{\begin{gathered}
\lstep{E_1}{P_1}{\alpha_1}{\pkw{skip}}\quad \lstep{E_2}{P_2}{\alpha_2}{\pkw{skip}}\\
(O_{C1},O_{F1})\in\Off(E_1,P_1,\alpha_1,\pkw{skip})\quad (O_{C2},O_{F2})\in\Off(E_2,P_2,\alpha_2,\pkw{skip})\\
E_1\in\LocalEnv(E,I,P_1,O_{C2},O_{F2})\quad E_2\in\LocalEnv(E,I,P_2,O_{C1},O_{F1})\quad \Sync_I(\alpha_1,\alpha_2)\end{gathered}}{\lstep{E}{P_1\pparop_I P_2}{\alpha_1\unionI\alpha_2}{\pkw{skip}}}\rulelabel{L-ParFuse}$
\end{adjustbox}
\end{semrulebox}
\end{center}
\begin{center}
\begin{semrulebox}[fontupper=\scriptsize]
\begin{adjustbox}{max width=.97\linewidth}
$\dfrac{\begin{gathered}
\lstep{E_1}{P_1}{\alpha_1}{P_1'}\quad \lstep{E_2}{P_2}{\alpha_2}{P_2'}\\
(O_{C1},O_{F1})\in\Off(E_1,P_1,\alpha_1,P_1')\quad (O_{C2},O_{F2})\in\Off(E_2,P_2,\alpha_2,P_2')\\
E_1\in\LocalEnv(E,I,P_1,O_{C2},O_{F2})\quad E_2\in\LocalEnv(E,I,P_2,O_{C1},O_{F1})\quad \Sync_I(\alpha_1,\alpha_2)\\
P_1\neq \pkw{skip}\quad P_2\neq \pkw{skip}\quad (P_1'\neq \pkw{skip}\ \text{or}\ P_2'\neq \pkw{skip})\end{gathered}}{\lstep{E}{P_1\pparop_I P_2}{\alpha_1\unionI\alpha_2}{P_1'\pparop_I P_2'}}\rulelabel{L-Par}$
\end{adjustbox}
\end{semrulebox}
\end{center}
\begin{center}
\begin{semrulebox}[fontupper=\scriptsize]
\begin{adjustbox}{max width=.97\linewidth}
$\dfrac{\begin{gathered}
\lstep{E_1}{\pkw{skip}}{\epsilon}{\pkw{skip}}\quad \lstep{E_2}{P_2}{\alpha_2}{P_2'}\\
(O_{C1},O_{F1})\in\Off(E_1,\pkw{skip},\epsilon,\pkw{skip})\quad (O_{C2},O_{F2})\in\Off(E_2,P_2,\alpha_2,P_2')\\
E_1\in\LocalEnv(E,I,\pkw{skip},O_{C2},O_{F2})\quad E_2\in\LocalEnv(E,I,P_2,O_{C1},O_{F1})\quad \Sync_I(\epsilon,\alpha_2)\\P_2\neq \pkw{skip}\quad P_2'\neq \pkw{skip}\end{gathered}}{\lstep{E}{\pkw{skip}\pparop_I P_2}{\alpha_2}{P_2'[\Alphabet(P)]}}\rulelabel{L-ParLeft}$
\end{adjustbox}
\end{semrulebox}
\end{center}
\begin{center}
\begin{semrulebox}[fontupper=\scriptsize]
\begin{adjustbox}{max width=.97\linewidth}
$\dfrac{\begin{gathered}
\lstep{E_1}{P_1}{\alpha_1}{P_1'}\quad \lstep{E_2}{\pkw{skip}}{\epsilon}{\pkw{skip}}\\
(O_{C1},O_{F1})\in\Off(E_1,P_1,\alpha_1,P_1')\quad (O_{C2},O_{F2})\in\Off(E_2,\pkw{skip},\epsilon,\pkw{skip})\\
E_1\in\LocalEnv(E,I,P_1,O_{C2},O_{F2})\quad E_2\in\LocalEnv(E,I,\pkw{skip},O_{C1},O_{F1})\quad \Sync_I(\alpha_1,\epsilon)\\P_1\neq \pkw{skip}\quad P_1'\neq \pkw{skip}\end{gathered}}{\lstep{E}{P_1\pparop_I \pkw{skip}}{\alpha_1}{P_1'[\Alphabet(P)]}}\rulelabel{L-ParRight}$
\end{adjustbox}
\end{semrulebox}
\end{center}

The rules \rulelabel{L-Par}, \rulelabel{L-ParFuse}, and
\rulelabel{L-ParLeft}/\rulelabel{L-ParRight} respectively retain both
components, finish both, or continue with the unfinished component.
Here $P$ denotes the annotated source of the parallel step. In the last case,
the finished side at $\pkw{skip}$ takes a silent step and exposes no offers.
It has no clock step, so the surviving process must first leave the product
in a non-clock step. A surviving $\pkw{cycle}$ waits in this step and takes
its clock step afterward.
Re-signing removes the old interface from the survivor's boundary and preserves
its nested interfaces as defined in Appendix~\ref{app:signed-process-syntax}.

\input{sections/appendix/pact/auxiliary-functions}

%% file: sections/appendix/pact/auxiliary-functions.tex
\subsubsection{Auxiliary Functions}
\applabel{aux-appendix}

The following functions define the branch tests, offers, and synchronization
conditions used by the operational rules.

\myparagraph{Messages and labels}
For a branch head $a$, its bare polarized label is
\[
 \Pol(\pkw{send}\ m\ e)=m\psym{!},\qquad \Pol(x=\pkw{recv}\ m)=m\psym{?}.
\]
On an observation, $\mathsf{Msg}(m\psym{!}v)=\mathsf{Msg}(m\psym{?}v)
=\mathsf{Msg}(m\psym{!})=\mathsf{Msg}(m\psym{?})=m$ returns the message name. The remaining
functions distinguish boundary messages from internal ones and erase
payloads from sets of offers:
\[
\begin{aligned}
 \BoundaryMsgs(P)&=\Alphabet(P),\\
 \IntMsgs(P)&=\bigcup\{I\mid P_1\pparop_I P_2\text{ occurs in }P\},\\
 \Bare(\{m\psym{!}v\})&=\{m\psym{!}\},\qquad \Bare(\{m\psym{?}v\})=\{m\psym{?}\},\\
 \Bare(\{b\})&=\{b\}\quad\text{for a bare action }b,\qquad
 \Bare(S)=\bigcup_{b\in S}\Bare(\{b\}).
\end{aligned}
\]
Restriction to an interface keeps message actions with names in that interface;
$\Hide{I}{S}$ removes those actions. Clock and marker observations are not
message offers.

\myparagraph{Branch needs}
To determine whether an offer can commit, we collect its branch heads and
the complementary actions needed to enable them:
\[
\begin{aligned}
 \Need(m\psym{!})&=m\psym{?},\qquad \Need(m\psym{?})=m\psym{!},\qquad
 \Need(a)=\Need(\Pol(a)),\\
 \Need(B)&=\{\Need(a_i)\mid i\in J\},\qquad
 \Resolved(E,B)\ \Longleftrightarrow\ \Need(B)\cap\Fut(E)=\emptyset.
\end{aligned}
\]
Here $B=\{a_i:P_i\mid i\in J\}$ is the branch set, and $\Need$ gives the
complement of each head. The last clause makes
the offer resolved when none of those needs remains in the future
environment. Let $\widehat a$ denote a head's concrete offer: $m\psym{!}v$ for a
send whose expression evaluates to $v$, and $m\psym{?}$ for a receive.

\myparagraph{Future offers}
$\FutOff(P)$ is an approximation to the messages that $P$ may offer before the
next clock. In a sequence, it includes offers from the right component when
the left component passes an optimistic completion test, defined below.
The approximation is computed from syntax as follows:
\[
\begin{aligned}
 \FutOff(\pkw{skip})&=\FutOff(X)=\FutOff(\pkw{cycle})
   =\FutOff(\pkw{end}_t)=\emptyset,\\
 \FutOff((\Omega,a\ \pkw{in}\ P))
   &=(\{\Pol(a)\}\cap\Omega)\cup\FutOff(P),\\
 \FutOff(\pkw{let}\ x=e\ \pkw{in}\ P)&=\FutOff(P),\qquad
 \FutOff(\pmu X.P)=\FutOff(P),\\
 \FutOff(\pkw{if}\ e\ \pkw{then}\ P_1\ \pkw{else}\ P_2)
   &=\FutOff(P_1)\cup\FutOff(P_2),\\
 \FutOff(\pkw{fair-if}_\eta\ e\ \pkw{then}\ P_1\ \pkw{else}\ P_2)
   &=\FutOff(P_1)\cup\FutOff(P_2),\\
 \FutOff(P_1\pseq P_2)&=
 \begin{cases}
  \FutOff(P_1)\cup\FutOff(P_2),&\Complete(P_1),\\
  \FutOff(P_1),&\text{otherwise},
 \end{cases}\\
 \FutOff(P_1\pparop_I P_2)&=\Hide{I}{\FutOff(P_1)\cup\FutOff(P_2)},\\
 \FutOff((\Omega,Z\ B\ \pkw{else}\ P_0))
   &=\Omega\cap\Bigl(\FutOff(P_0)\cup
      \bigcup_{i\in J}(\{\Pol(a_i)\}\cup\FutOff(P_i))\Bigr).
\end{aligned}
\]
Here $Z$ ranges over $\pkw{xoffer}$ and $\pkw{offer}$. An undeclared head
contributes no action, but its continuation still contributes its declared
future offers. The completion test used for sequences is
\[
\begin{aligned}
 \Complete(\pkw{skip})&=\Complete(\pkw{end}_t)=\mathsf{true},\qquad
 \Complete(\pkw{cycle})=\Complete(X)=\mathsf{false},\\
 \Complete(a\ \pkw{in}\ P)&=\Complete(P),\quad
 \Complete(\pkw{let}\ x=e\ \pkw{in}\ P)=\Complete(P),\quad
 \Complete(\pmu X.P)=\Complete(P),\\
 \Complete(\pkw{if}\ e\ \pkw{then}\ P_1\ \pkw{else}\ P_2)
   &=\Complete(P_1)\lor\Complete(P_2),\\
 \Complete(\pkw{fair-if}_\eta\ e\ \pkw{then}\ P_1\ \pkw{else}\ P_2)
   &=\Complete(P_1)\lor\Complete(P_2),\\
 \Complete(P_1\pseq P_2)&=\Complete(P_1\pparop_I P_2)
   =\Complete(P_1)\land\Complete(P_2),\\
 \Complete(Z\ B\ \pkw{else}\ P_0)&=\Complete(P_0)\lor
      \bigvee_{i\in J}\Complete(P_i).
\end{aligned}
\]
For this test, a prefix passes whenever its continuation does, even if the
head cannot fire. Parallel components are tested separately.

\myparagraph{Current and future offers of a step}
$\Off(E,P,\alpha,P')$ is the set of pairs $(O_C,O_F)$ exposed by the selected
step $\lstep{E}{P}{\alpha}{P'}$. Current sends retain their payloads; future
actions are bare. All current offers are projected to the source alphabet:
\[
\begin{aligned}
 \mathsf{Proj}_\Omega(S)&=\{b\in S\mid\Bare(\{b\})\subseteq\Omega\},\\
 D_a&=\mathsf{Proj}_{\Alphabet(P)}(\{\widehat a\}),\qquad
 A_B=\{\widehat a_i\mid i\in J\},\qquad
 D_B=\mathsf{Proj}_{\Alphabet(P)}(A_B),\\
 W(E,B)&=\{\widehat a_i\mid i\in J,\ \Need(a_i)\in\Fut(E)\}.
\end{aligned}
\]
Here $D_a$ is the declared offer of one head, $A_B$ contains all concrete
branch offers, and $D_B$ keeps the declared ones. The waiting set $W(E,B)$
keeps heads whose complements are still promised. The source $P$ in $D_a$
and $D_B$ is the source of the displayed step.
For a firing prefix, $\alpha\neq\epsilon$ is its label and $P'$ is its continuation,
with the received payload substituted when needed. Every clock step, at every constructor,
exposes $(\emptyset,\emptyset)$. For non-clock steps, the direct clauses are
\[
\begin{aligned}
 \Off(E,a\ \pkw{in}\ Q,\alpha,P')
   &=\{(D_a,\FutOff(P')\setminus\Bare(D_a))\},\\
 \Off(E,a\ \pkw{in}\ Q,\epsilon,a\ \pkw{in}\ Q)
   &=\{(D_a,\FutOff(a\ \pkw{in}\ Q)\setminus\Bare(D_a))\},\\
 \Off(E,\pkw{cycle},\epsilon,\pkw{cycle})&=\{(\emptyset,\emptyset)\},\\
 \Off(E,\pkw{skip},\epsilon,\pkw{skip})&=\{(\emptyset,\emptyset)\},\\
 \Off(E,\pkw{end}_t,\$_t,\pkw{skip})&=\{(\emptyset,\emptyset)\},\\
 \Off(E,\pkw{if}\ e\ \pkw{then}\ P_1\ \pkw{else}\ P_2,\epsilon,P_i)
   &=\{(\emptyset,\FutOff(P_1)\cup\FutOff(P_2))\},\\
 \Off(E,\pkw{fair-if}_\eta\ e\ \pkw{then}\ P_1\ \pkw{else}\ P_2,\eta_i,P_i)
   &=\{(\emptyset,\FutOff(P_1)\cup\FutOff(P_2))\},\\
 \Off(E,\pkw{let}\ x=e\ \pkw{in}\ Q,\epsilon,P')
   &=\{(\emptyset,\FutOff(P'))\},\\
 \Off(E,\pmu X.Q,\epsilon,P')&=\{(\emptyset,\FutOff(P'))\}.
\end{aligned}
\]
Here $i\in\{1,2\}$ selects the true or false branch, respectively, with
$\eta_1=\eta_{\mathsf{true}}$ and $\eta_2=\eta_{\mathsf{false}}$.
Both conditional outcomes publish the source menu. Thus the future information
used by a sibling does not depend on which data-dependent branch was selected.
For $Z=\pkw{xoffer}\ B\ \pkw{else}\ P_0$ and
$Y=\pkw{offer}\ B\ \pkw{else}\ P_0$, the offer clauses are given below.
In the first clause, $\alpha_i$ and $P_i'$ are the label and substituted
continuation of branch $a_i:P_i$.
\[
\begin{aligned}
 \Off(E,Z,\alpha_i,P_i')
   &=\{(\{\widehat a_i\},\FutOff(Z)\setminus\Bare(A_B))\},\\
 \Off(E,Z,\epsilon,P_0)&=\{(D_B,\FutOff(Z)\setminus\Bare(A_B))\},\\
 \Off(E,Z,\epsilon,Z)&=\{(W(E,B),\FutOff(Z)\setminus\Bare(W(E,B)))\},\\
 \Off(E,Y,\beta_K(\mathbf v),C^P_{B,K,\mathbf v})
   &=\{(D_B,\FutOff(C^P_{B,K,\mathbf v})\setminus\Bare(D_B))\},\\
 \Off(E,Y,\epsilon,P_0)&=\{(D_B,\FutOff(P_0)\setminus\Bare(D_B))\},\\
 \Off(E,Y,\epsilon,Y)&=\{(D_B,\FutOff(Y)\setminus\Bare(D_B))\}.
\end{aligned}
\]
Here each clause is read with the premises of its corresponding operational
rule. Admissibility ensures that a firing head is declared and that $W(E,B)$
is already within the source alphabet. An undeclared head can wait with empty
current offers.
Exclusive commitment removes all branch heads from its future menu, but
retains the source's continuation menus. Conjunctive firing instead uses the
residual product, which retains the heads that have not communicated.
Both defaults advertise $D_B$ even though no head communicates: current offers
record availability, while labels record the chosen communications. The
exclusive default retains the source continuation menus; the conjunctive
default uses its selected continuation's menu.
A waiting exclusive offer advertises only $W(E,B)$; a waiting conjunctive
offer advertises all declared heads, since it can retain the others after a
head communicates.

\myparagraph{Compound offer clauses}
For each $(O_C,O_F)\in\Off(E,P_1,\alpha,P_1')$ with $\#\notin\alpha$,
sequencing publishes
\[
 \left(O_C,(O_F\cup F_2)\setminus\Bare(O_C)\right),\qquad
 F_2=\begin{cases}\FutOff(P_2),&\Complete(P_1),\\
                  \emptyset,&\text{otherwise}.
     \end{cases}
\]
This is the pair for either sequence rule: finishing changes the target,
while the future field is computed from the same left source and chosen pair.
For a parallel step, range over the component steps and pairs in its premises.
The product and simultaneous-completion rules publish
\[
 \left(\Hide{I}{O_{C1}\cup O_{C2}},
       \Hide{I}{O_{F1}\cup O_{F2}}\right).
\]
A survivor rule publishes the surviving component's pair with $I$ hidden in
both fields. The finished side still supplies its silent $\pkw{skip}$ step and
empty pair. $\Off$ collects a pair for every solution of these mutual local
environment constraints. In particular, each derivable step has a decoration.

\myparagraph{Local environments}
$\LocalEnv(E,I,P_i,O_C,O_F)$ gives the environment of child $P_i$ using the
sibling's chosen offer pair. It is the singleton $\{(U_L,F_L)\}$ where
\[
\begin{aligned}
 U_L&=\begin{cases}
        \{\#\},&\Imm(E)=\{\#\},\\
        \Imm(E)\cup\restr{O_C}{I},&\text{otherwise},
       \end{cases}\\
 F_L&=\begin{cases}
        \emptyset,&\Imm(E)=\{\#\},\\
        (\Fut(E)\cup\restr{O_F}{I})\setminus\Bare(U_L),&\text{otherwise}.
       \end{cases}
\end{aligned}
\]
Both children receive the complete parent environment. Alphabet compatibility
ensures that its actions are declared at each child. The parallel rule also
requires each resulting local environment to belong to its child's
environment set. The argument $P_i$ identifies that child; the displayed
formula for the pair itself does not depend on it.

\myparagraph{Synchronization}
A label identifies the message names sent and received during its step:
\[
 \Sends(\alpha)=\{m\mid\exists v.\ m\psym{!}v\in\alpha\},\qquad
 \Recvs(\alpha)=\{m\mid\exists v.\ m\psym{?}v\in\alpha\}.
\]
Two component labels are synchronized on
interface $I$ when every $I$-tagged send on one side is matched by an
$I$-tagged receive on the other. Clock ticks also synchronize: one side may
produce $\#$ exactly when the other side does. The local environment
contains the sibling's valued send offer, and the receive rule takes its
payload from that environment. Thus the payload agreement is checked before
$\Sync_I$ compares message names.
\[
\begin{aligned}
\Sync_I(\alpha_1,\alpha_2) \triangleq{} &
I\cap\Sends(\alpha_1) = I\cap\Recvs(\alpha_2)
\ \text{and}\ I\cap\Recvs(\alpha_1) = I\cap\Sends(\alpha_2) \\
& \text{and}\ (\alpha_1=\{\#\}\Leftrightarrow\alpha_2=\{\#\}).
\end{aligned}
\]
Given $\Sync_I(\alpha_1,\alpha_2)$, the messages resolved internally in this
step are
\[
\SyncedMsgs(\alpha_1,\alpha_2) = (I\cap\Sends(\alpha_1))\cup(I\cap\Sends(\alpha_2))
\]
(equivalently $(I\cap\Recvs(\alpha_2))\cup(I\cap\Recvs(\alpha_1))$, by
$\Sync_I$). Combining labels replaces every raw send/receive occurrence of a
synced message by a single $\tick{\cdot}$ observation, leaving every other
occurrence unchanged. These include boundary messages, internal exchanges
already resolved in a deeper composition, $\#$, fair outcomes, and progress
markers:
\[
\alpha_1\unionI\alpha_2 =
\{\tick{m}\mid m\in\SyncedMsgs(\alpha_1,\alpha_2)\}
\cup\bigl((\alpha_1\cup\alpha_2)\setminus\{m\psym{!}v,m\psym{?}v\mid m\in\SyncedMsgs(\alpha_1,\alpha_2),\,v\}\bigr).
\]

%% file: sections/appendix/pact/defining-communication-safety.tex
\subsection{Reachability and Closed Environments}
\applabel{app-pact-safety}

We give the reachable prefixes used by the safety definitions in
\secref{pact-safety}, followed by the environment restriction for closed
systems. For the initial prefix $h_0=Q_0$, the reachable set is
\[
\Reach_{\Strategy_{Q_0}}(Q_0) =
\{h\mid \stratstep{\Strategy_{Q_0}}{}{h_0}{*}{h}\}.
\]
Here $*$ denotes the reflexive-transitive closure of the strategy-step
relation. Thus the initial prefix is reachable, as is every prefix obtained
by repeatedly taking a step under the environment selected by the strategy.

In a closed system, $\BoundaryMsgs(Q_0)=\emptyset$, so the strategy has no
boundary messages to offer. For every prefix $h$, its environment satisfies
\[
\Imm(\Strategy_{Q_0}(h))\in\{\emptyset,\{\#\}\}
\quad\text{and}\quad
\Fut(\Strategy_{Q_0}(h))=\emptyset.
\]
The definitions of honesty, synchronization compatibility, and liveness apply
with this restriction; closed compatibility still requires one honest
strategy to ensure both synchronization compatibility and liveness.

%% file: sections/appendix/types.tex
\section{Complete Definitions and Proofs for the \codename{} Type System}
\applabel{appendix-types}

We first give the full type semantics and automaton constructions. These
provide the basis for the proofs and algorithms for conformance, refinement,
and compatibility that follow. Finally, we prove that the inferred types are
optimal for data-oblivious checking.

\input{sections/appendix/types/type-syntax-and-semantics}

\input{sections/appendix/types/type-inference-refinement-and-checking}

\input{sections/appendix/types/type-compatibility-checking}
\input{sections/appendix/types/optimality}

%% file: sections/appendix/types/type-syntax-and-semantics.tex
\subsection{Complete Type Semantics and Automaton Construction}
\applabel{type-semantics-appendix}

We first make the message alphabets explicit in the type syntax and state
the well-formedness requirements. We then give the operational rules and offer
functions, followed by the canonical and constructor-based automata and their
correspondence proofs.

\subsubsection{Signed Types}
The full type syntax retains the boundary alphabet omitted in the main text.
Every child of a compound type also carries an alphabet:

\[
\begin{aligned}
 \tau::={}&(\Omega,\nu),\qquad\Omega\subseteq\mathsf{ActM},\\
 \nu::={}&\tkw{skip}\mid a\mid\tclock\mid\tmarker{t}\mid X\mid\tmu X.\tau
       \mid\tau\tseq \tau\mid\tau\tparop_I\tau\\
       &\mid\tau\tplus\tau\mid\tau\tsym{\dotplus}_\eta\tau
       \mid\toplus(\tau_0,a_1:\tau_1,\ldots,a_n:\tau_n)
       \mid\twedge(\tau_0,a_1:\tau_1,\ldots,a_n:\tau_n),\\
 a::={}&m\tsym{!}\mid m\tsym{?}.
\end{aligned}
\]
Here $\Omega$ is the finite alphabet, denoted by $\Alphabet(\tau)$, and
$\nu$ is the constructor body, denoted by $\mathsf{Body}(\tau)$. The polarized message
actions $\mathsf{ActM}$ and alphabet well-formedness are defined in
Appendix~\ref{app:signed-process-syntax}: sequential and choice children carry
the parent alphabet, while parallel children satisfy $\mathsf{ParCompat}_I$.
These conditions also apply to every offer continuation and default.

Parallel completion changes the surviving type's alphabet just as it changes
a process alphabet. This re-signing propagates through type constructors:

\[
\begin{aligned}
 (\Omega,\nu)[\Omega']&=(\Omega',\nu[\Omega']),\\
 \nu[\Omega']&=\nu\quad
   (\nu\in\{\tkw{skip},a,\tclock,\tmarker{t},X\}),\\
 (\tau_1\circ\tau_2)[\Omega']&=\tau_1[\Omega']\circ\tau_2[\Omega']
       \quad(\circ\in\{\tseq ,+,\tsym{\dotplus}_\eta\}),\\
 (\tmu X.\tau_0)[\Omega']&=\tmu X.\tau_0[\Omega'],\\
 \toplus(\tau_0,\{a_i:\tau_i\}_{i\in J})[\Omega']
   &=\toplus(\tau_0[\Omega'],\{a_i:\tau_i[\Omega']\}_{i\in J}),\\
 \twedge(\tau_0,\{a_i:\tau_i\}_{i\in J})[\Omega']
   &=\twedge(\tau_0[\Omega'],\{a_i:\tau_i[\Omega']\}_{i\in J}),\\
 (\tau_1\tparop_I\tau_2)[\Omega']
   &=\tau_1[\Pi_I(\Omega',\Alphabet(\tau_1))]
       \tparop_I\tau_2[\Pi_I(\Omega',\Alphabet(\tau_2))].
\end{aligned}
\]
Here $\Pi_I$ is the child-alphabet projection defined for processes. A
parallel child retains its own interface polarities, while every other
constructor passes the incoming alphabet to all its children.

\input{sections/appendix/types/type-well-formedness}

\subsubsection{Type Semantics}

The semantics of types is the value-erased analogue of the process semantics
of \secref{pact-semantics}. A type transition has the form
\[
\lstep{E_\tau}{\tau}{\alpha}{\tau'},
\]
where $E_\tau=(U,F)$ is the current/future environment and $\alpha$ is the
observation. For a source $\tau=(\Omega,\nu)$, admissible environments are
\[
\begin{aligned}
 \Env{\tau}=\{(U,F)\mid&U\in2^{\overline\Omega}\cup\{\{\#\}\},\quad
 F\in2^{\overline\Omega},\\
 &U=\{\#\}\Rightarrow F=\emptyset,\quad
 U\neq\{\#\}\Rightarrow U\cap F=\emptyset\}.
\end{aligned}
\]
Here $U$ and $F$ contain bare complementary actions; $U$ may instead contain
the clock tick, in which case $F$ is empty. Every rule has the implicit
premise $E_\tau\in\Env{\tau}$, and each child
step has its corresponding admissibility premise. Rules suppress annotations:
selecting a child returns that signed child, while rebuilding a constructor
keeps the source alphabet.

Payload erasure applies to observations, environments, and auxiliary
functions. The rules below follow the corresponding process rules, with
data-dependent conditionals replaced by nondeterministic choices.

The atomic action $a$ fires when the environment contains its complement:
\begin{center}
\begin{semrulebox}[fontupper=\scriptsize]
\begin{adjustbox}{max width=.99\linewidth}
$\begin{array}{@{}ccc@{}}
\dfrac{\Need(a)\in\Imm(E_\tau)}
      {\lstep{E_\tau}{a}{a}{\tkw{skip}}}
\rulelabel{T-Action}
&
\dfrac{\Imm(E_\tau)\neq\{\tclock\}}
      {\lstep{E_\tau}{\tau_1 + \tau_2}{\epsilon}{\tau_i}}
\rulelabel{T-Choice}
&
\dfrac{\Imm(E_\tau)\neq\{\tclock\}}
      {\lstep{E_\tau}{\tau_1 \tsym{\dotplus}_\eta \tau_2}{\eta_i}{\tau_i}}
\rulelabel{T-FairChoice}
\\[1.15em]
\multicolumn{3}{c}{
\dfrac{\Need(a)\notin\Imm(E_\tau) \quad \Imm(E_\tau)\neq\{\tclock\}}
      {\lstep{E_\tau}{a}{\epsilon}{a}}
\rulelabel{T-ActionWait}}
\end{array}$
\end{adjustbox}
\end{semrulebox}
\end{center}
Here $i\in\{1,2\}$, with $\eta_1=\eta_{\mathsf{true}}$ and
$\eta_2=\eta_{\mathsf{false}}$.
Sequencing and recursion retain their process rules, and the clock and
progress rules carry over unchanged:
\begin{center}
\begin{semrulebox}[fontupper=\scriptsize]
\begin{adjustbox}{max width=.99\linewidth}
$\begin{array}{@{}ccc@{}}
\dfrac{\Imm(E_\tau)=\{\tclock\}}
      {\lstep{E_\tau}{\tclock}{\tclock}{\tkw{skip}}}
\rulelabel{T-Clock}
&
\dfrac{\Imm(E_\tau)\neq\{\tclock\}}
      {\lstep{E_\tau}{\tclock}{\epsilon}{\tclock}}
\rulelabel{T-ClockWait}
&
\dfrac{\Imm(E_\tau)\neq\{\tclock\}}
      {\lstep{E_\tau}{\tkw{skip}}{\epsilon}{\tkw{skip}}}
\rulelabel{T-Skip}
\\[1.15em]
\dfrac{\lstep{E_\tau}{\tau_1}{\alpha}{\tkw{skip}}}
      {\lstep{E_\tau}{\tau_1\tseq\tau_2}{\alpha}{\tau_2}}
\rulelabel{T-SeqDone}
&
\dfrac{\Imm(E_\tau)\neq\{\tclock\}}
      {\lstep{E_\tau}{\tmarker{t}}{\tmarker{t}}{\tkw{skip}}}
\rulelabel{T-Marker}
&
\dfrac{\Imm(E_\tau)\neq\{\tclock\}}
      {\lstep{E_\tau}{\tmu X.\tau}{\epsilon}{\tau[X\mapsto\tmu X.\tau]}}
\rulelabel{T-Rec}
\\[1.15em]
\multicolumn{3}{c}{
\dfrac{\lstep{E_\tau}{\tau_1}{\alpha}{\tau_1'} \quad \tau_1'\neq\tkw{skip}}
      {\lstep{E_\tau}{\tau_1\tseq\tau_2}{\alpha}{\tau_1'\tseq\tau_2}}
\rulelabel{T-Seq}}
\end{array}$
\end{adjustbox}
\end{semrulebox}
\end{center}

An exclusive offer selects one immediately enabled branch after the branch
set is resolved. It waits while resolution is pending, or takes its default
when resolved with no enabled branch:
\begin{center}
\begin{semrulebox}[fontupper=\scriptsize]
\begin{adjustbox}{max width=.99\linewidth}
$\begin{array}{@{}c@{}}
\dfrac{a_i : \tau_i \in B \quad \Need(a_i) \in \Imm(E_\tau) \quad \Resolved(E_\tau,B)}
      {\lstep{E_\tau}{\toplus(\tau_0, B)}{a_i}{\tau_i}}
\rulelabel{T-XChoice}
\\[1.15em]
\dfrac{\Need(B) \cap \Imm(E_\tau) = \emptyset \quad \Resolved(E_\tau,B) \quad \Imm(E_\tau)\neq\{\tclock\}}
      {\lstep{E_\tau}{\toplus(\tau_0, B)}{\epsilon}{\tau_0}}
\rulelabel{T-XChoiceElse}
\qquad
\dfrac{\neg\,\Resolved(E_\tau,B) \quad \Imm(E_\tau)\neq\{\tclock\}}
      {\lstep{E_\tau}{\toplus(\tau_0, B)}{\epsilon}{\toplus(\tau_0, B)}}
\rulelabel{T-XChoiceWait}
\end{array}$
\end{adjustbox}
\end{semrulebox}
\end{center}
The conjunctive offer communicates every immediately enabled head in the
same step. For $B=\{a_i:\tau_i\mid i\in J_B\}$, where
$J_B=\{1,\ldots,n\}$, define
\[
 K(E_\tau,B)=\{i\in J_B\mid\Need(a_i)\in\Imm(E_\tau)\},\qquad
 A_K=\{a_i\mid i\in K\}.
\]
After these communications, the target contains the enabled branches'
continuations and an optional singleton offer for every remaining head:
\[
\begin{aligned}
 R_i^K&=\begin{cases}
         \tau_i[\Omega],&i\in K,\\
         \toplus(\tkw{skip},a_i:\tau_i[\Omega]),&i\notin K,
        \end{cases}\\
 C_{B,K}&=R_1^K\tparop_{\emptyset}\cdots\tparop_{\emptyset}R_n^K.
\end{aligned}
\]
The product is right-associated, with a single factor returned unchanged.
Every generated node carries the source alphabet $\Omega$. Re-signing each
continuation is the identity for an alphabet-well-formed source. The
continuations take no step in this round.
\begin{center}
\begin{semrulebox}[fontupper=\scriptsize]
\begin{adjustbox}{max width=.99\linewidth}
$\begin{array}{@{}c@{}}
\dfrac{K=K(E_\tau,B)\quad K\neq\emptyset}
      {\lstep{E_\tau}{\twedge(\tau_0, B)}{A_K}{C_{B,K}}}
\rulelabel{T-ConjChoice}
\\[1.15em]
\dfrac{\forall a_i:\tau_i\in B.\ \Need(a_i)\notin\Imm(E_\tau) \quad \Resolved(E_\tau,B) \quad \Imm(E_\tau)\neq\{\tclock\}}
      {\lstep{E_\tau}{\twedge(\tau_0, B)}{\epsilon}{\tau_0}}
\rulelabel{T-ConjChoiceElse}
\qquad
\dfrac{\forall a_i:\tau_i\in B.\ \Need(a_i)\notin\Imm(E_\tau) \quad
       \neg\,\Resolved(E_\tau,B) \quad \Imm(E_\tau)\neq\{\tclock\}}
      {\lstep{E_\tau}{\twedge(\tau_0, B)}{\epsilon}{\twedge(\tau_0, B)}}
\rulelabel{T-ConjChoiceWait}
\end{array}$
\end{adjustbox}
\end{semrulebox}
\end{center}
Parallel composition uses the same four cases as the process semantics.
Both component steps and their offer pairs determine the local environments.

\begin{center}
\begin{semrulebox}[fontupper=\scriptsize]
\begin{adjustbox}{max width=.97\linewidth}
$\dfrac{\begin{gathered}
\lstep{E_1}{\tau_1}{\alpha_1}{\tkw{skip}}\quad \lstep{E_2}{\tau_2}{\alpha_2}{\tkw{skip}}\\
(O_{C1},O_{F1})\in\Off(E_1,\tau_1,\alpha_1,\tkw{skip})\quad (O_{C2},O_{F2})\in\Off(E_2,\tau_2,\alpha_2,\tkw{skip})\\
E_1\in\LocalEnv(E_\tau,I,\tau_1,O_{C2},O_{F2})\quad E_2\in\LocalEnv(E_\tau,I,\tau_2,O_{C1},O_{F1})\quad \Sync_I(\alpha_1,\alpha_2)\end{gathered}}{\lstep{E_\tau}{\tau_1\tparop_I \tau_2}{\alpha_1\unionI\alpha_2}{\tkw{skip}}}\rulelabel{T-ParFuse}$
\end{adjustbox}
\end{semrulebox}
\end{center}
\begin{center}
\begin{semrulebox}[fontupper=\scriptsize]
\begin{adjustbox}{max width=.97\linewidth}
$\dfrac{\begin{gathered}
\lstep{E_1}{\tau_1}{\alpha_1}{\tau_1'}\quad \lstep{E_2}{\tau_2}{\alpha_2}{\tau_2'}\\
(O_{C1},O_{F1})\in\Off(E_1,\tau_1,\alpha_1,\tau_1')\quad (O_{C2},O_{F2})\in\Off(E_2,\tau_2,\alpha_2,\tau_2')\\
E_1\in\LocalEnv(E_\tau,I,\tau_1,O_{C2},O_{F2})\quad E_2\in\LocalEnv(E_\tau,I,\tau_2,O_{C1},O_{F1})\quad \Sync_I(\alpha_1,\alpha_2)\\
\tau_1\neq \tkw{skip}\quad \tau_2\neq \tkw{skip}\quad (\tau_1'\neq \tkw{skip}\ \text{or}\ \tau_2'\neq \tkw{skip})\end{gathered}}{\lstep{E_\tau}{\tau_1\tparop_I \tau_2}{\alpha_1\unionI\alpha_2}{\tau_1'\tparop_I \tau_2'}}\rulelabel{T-Par}$
\end{adjustbox}
\end{semrulebox}
\end{center}
\begin{center}
\begin{semrulebox}[fontupper=\scriptsize]
\begin{adjustbox}{max width=.97\linewidth}
$\dfrac{\begin{gathered}
\lstep{E_1}{\tkw{skip}}{\epsilon}{\tkw{skip}}\quad \lstep{E_2}{\tau_2}{\alpha_2}{\tau_2'}\\
(O_{C1},O_{F1})\in\Off(E_1,\tkw{skip},\epsilon,\tkw{skip})\quad (O_{C2},O_{F2})\in\Off(E_2,\tau_2,\alpha_2,\tau_2')\\
E_1\in\LocalEnv(E_\tau,I,\tkw{skip},O_{C2},O_{F2})\quad E_2\in\LocalEnv(E_\tau,I,\tau_2,O_{C1},O_{F1})\quad \Sync_I(\epsilon,\alpha_2)\\\tau_2\neq \tkw{skip}\quad \tau_2'\neq \tkw{skip}\end{gathered}}{\lstep{E_\tau}{\tkw{skip}\tparop_I \tau_2}{\alpha_2}{\tau_2'[\Alphabet(\tau)]}}\rulelabel{T-ParLeft}$
\end{adjustbox}
\end{semrulebox}
\end{center}
\begin{center}
\begin{semrulebox}[fontupper=\scriptsize]
\begin{adjustbox}{max width=.97\linewidth}
$\dfrac{\begin{gathered}
\lstep{E_1}{\tau_1}{\alpha_1}{\tau_1'}\quad \lstep{E_2}{\tkw{skip}}{\epsilon}{\tkw{skip}}\\
(O_{C1},O_{F1})\in\Off(E_1,\tau_1,\alpha_1,\tau_1')\quad (O_{C2},O_{F2})\in\Off(E_2,\tkw{skip},\epsilon,\tkw{skip})\\
E_1\in\LocalEnv(E_\tau,I,\tau_1,O_{C2},O_{F2})\quad E_2\in\LocalEnv(E_\tau,I,\tkw{skip},O_{C1},O_{F1})\quad \Sync_I(\alpha_1,\epsilon)\\\tau_1\neq \tkw{skip}\quad \tau_1'\neq \tkw{skip}\end{gathered}}{\lstep{E_\tau}{\tau_1\tparop_I \tkw{skip}}{\alpha_1}{\tau_1'[\Alphabet(\tau)]}}\rulelabel{T-ParRight}$
\end{adjustbox}
\end{semrulebox}
\end{center}
The auxiliary functions $\Imm$, $\Fut$, $\Need$, $\Resolved$,
$\mathsf{Off}$, $\LocalEnv$, and $\Sync_I$ used above are the erased
readings of their process counterparts. For type labels,
$\Sends(\alpha)=\{m\mid m\tsym{!}\in\alpha\}$ and
$\Recvs(\alpha)=\{m\mid m\tsym{?}\in\alpha\}$. Synchronization then uses the same
two message-equality clauses as process synchronization, together with the
same clock agreement condition:
\[
\alpha_1=\{\tclock\}\Leftrightarrow\alpha_2=\{\tclock\}.
\]
Together with the displayed rules, these erased auxiliary functions define
the transition relation used below.

\input{sections/appendix/types/type-offer-functions}
\input{sections/appendix/types/semantic-automaton}
\input{sections/appendix/types/constructor-automaton}

%% file: sections/appendix/types/type-well-formedness.tex
\myparagraph{Well-formedness}
Types must satisfy the following well-formedness requirements on alphabets,
recursion, message use, and marker identifiers:
\[
\begin{aligned}
 \mathsf{TypeWF}(\tau)\quad\Longleftrightarrow\quad
 &\mathsf{AlphaWF}(\tau)\land\mathsf{Closed}(\tau)\land\mathsf{Hygienic}(\tau)\\
 &{}\land\mathsf{ClockGuarded}(\tau)\land\mathsf{TailRec}(\tau)
       \land\mathsf{SameCycleLinear}(\tau)\\
 &{}\land\mathsf{IfaceScoped}(\tau)\land\mathsf{ParaMsgOwn}(\tau)
       \land\mathsf{ConjMsgOwn}(\tau)\\
 &{}\land\mathsf{ProgLinear}(\tau)\land\mathsf{FairLinear}(\tau).
\end{aligned}
\]
Clock guarding and message linearity are required to prevent hardware hazards. Without clock guarding, recursive unfolding within a cycle can create a combinational loop. Similarly, reusing a message within the same cycle can cause distinct exchanges to compete for the same handshake signals, leading to multi-driver issues. The alphabet, scope, and ownership conditions prevent conflicting uses of statically declared interfaces. Finally, tail recursion prevents recursive calls from accumulating continuations or parallel components.
Together with the other requirements, this ensures a finite control representation.

Here $\mathsf{AlphaWF}$ is the alphabet condition from
Appendix~\ref{app:signed-process-syntax}. The clause $\mathsf{Closed}$ excludes
free recursion variables, and $\mathsf{Hygienic}$ forbids a binder from reusing
the name of an enclosing binder. Sibling binders may reuse a name.

For each binder $\tmu X.\theta$, $\mathsf{ClockGuarded}$ forbids a syntactic
path from $\theta$ to $X$ before the next $\tclock$. Such paths enter the right
side of a sequence only if its left side passes $\Complete$, follow one alternative
at choices and offers, and explore parallel components independently. The
completion test and future offers are defined in
Appendix~\ref{app:type-offer-functions}; semantic paths use the rules in
this subsection.

Tail recursion requires each call to return directly to its binder. To check
this condition, we determine where calls to a binder may occur in its body:
\[
\begin{aligned}
 \mathsf{Tail}_X(\theta_1\tseq \theta_2)
  &\Longleftrightarrow X\notin\mathsf{FV}(\theta_1)
       \land\mathsf{Tail}_X(\theta_2),\\
 \mathsf{Tail}_X(\theta_1\tparop_I\theta_2)
  &\Longleftrightarrow X\notin\mathsf{FV}(\theta_1)\cup\mathsf{FV}(\theta_2),\\
 \mathsf{Tail}_X(\theta_1+\theta_2)
  &\Longleftrightarrow\mathsf{Tail}_X(\theta_1)\land\mathsf{Tail}_X(\theta_2),\\
 \mathsf{Tail}_X(\theta_1\tsym{\dotplus}_\eta\theta_2)
  &\Longleftrightarrow\mathsf{Tail}_X(\theta_1)\land\mathsf{Tail}_X(\theta_2),\\
 \mathsf{Tail}_X(\toplus(\theta_0,\{a_i:\theta_i\}_{i\in J}))
  &\Longleftrightarrow\mathsf{Tail}_X(\theta_0)
       \land\bigwedge_{i\in J}\mathsf{Tail}_X(\theta_i),\\
 \mathsf{Tail}_X(\twedge(\theta_0,\{a_i:\theta_i\}_{i\in J}))
  &\Longleftrightarrow\mathsf{Tail}_X(\theta_0)
       \land X\notin\bigcup_{i\in J}\mathsf{FV}(\theta_i),\\
 \mathsf{Tail}_X(\tmu Y.\theta)
  &\Longleftrightarrow Y=X\ \lor\ \mathsf{Tail}_X(\theta).
\end{aligned}
\]
Here $\mathsf{FV}$ denotes free recursion variables, and $\mathsf{Tail}_X$
holds on atoms and variables. The clause $\mathsf{TailRec}$ requires
$\mathsf{Tail}_X(\theta)$ at every binder $\tmu X.\theta$. The sequence clause
allows calls to $X$ only in the right continuation; the parallel clause
excludes them from both children. Choices and exclusive offers check every
alternative. A conjunctive offer allows calls to $X$ only in its default;
its branches may still recurse through binders inside those branches.
At a binder $\tmu Y.\theta$, the check stops when $Y=X$; otherwise it checks
the body.

The clause $\mathsf{SameCycleLinear}$ forbids repeated uses of a message
along a finite semantic path before the next clock, starting from any signed
subterm. A send, receive, or internal $\tick{m}$ contributes $m$, with an
internal exchange counted once. For offers, it also checks the branch heads
against the future offers of every default and branch continuation:
\[
 \{\mathsf{Msg}(a_i)\mid i\in J\}
 \cap\{\mathsf{Msg}(b)\mid b\in\FutOff(\sigma)\}=\emptyset.
\]
Here $\{a_i\}_{i\in J}$ are the offer heads and $\sigma$ is any default or
branch continuation. The condition excludes all head message names from
each continuation's optimistic future menu, including offers that may not
occur on an actual path.

The clause $\mathsf{IfaceScoped}$ keeps an interface distinct from those
declared inside its children. At every $(\Omega,\theta_1\tparop_I\theta_2)$,
it requires
\[
 I\cap\mathsf{IfaceIds}(\theta_1)
 =I\cap\mathsf{IfaceIds}(\theta_2)=\emptyset.
\]
Here $\mathsf{IfaceIds}(\theta)$ collects all interfaces declared below
$\theta$. The same clause requires every signed subterm $(\Omega,\nu)$ to
satisfy $\restr{\Omega}{\mathsf{IfaceIds}(\nu)}=\emptyset$: an interface
declared inside a subterm cannot also occur in its boundary alphabet.

The clause $\mathsf{ParaMsgOwn}$ allows two parallel children to share
message names only on their interface. It compares the message names in
their $\mathsf{MsgLabel}$ sets, which collect the polarized actions written
in each body and its continuations. For a conjunctive offer,
$\mathsf{ConjMsgOwn}$ requires the message names in each branch head and
continuation to be disjoint from those of every other branch. These
conditions are checked recursively, separating the messages used by
independently executing bodies.

Finally, $\mathsf{ProgLinear}$ and $\mathsf{FairLinear}$ require each progress
identifier and each fair-choice identifier, respectively, to occur at most
once in the written syntax.

%% file: sections/appendix/types/type-offer-functions.tex
\myparagraph{Type offer functions}
\label{app:type-offer-functions}
The type offer functions are the erased versions of those in
\appref{aux-appendix}. On type syntax, their future-menu clauses are
\[
\begin{aligned}
 \FutOff(\tkw{skip})&=\FutOff(X)=\FutOff(\tclock)=\FutOff(\tmarker{t})=\emptyset,\\
 \FutOff((\Omega,a))&=\{a\}\cap\Omega,\qquad
 \FutOff(\tmu X.\theta)=\FutOff(\theta),\\
 \FutOff(\theta_1+\theta_2)&=\FutOff(\theta_1\tsym{\dotplus}_\eta\theta_2)
       =\FutOff(\theta_1)\cup\FutOff(\theta_2),\\
 \FutOff(\theta_1\tseq \theta_2)&=
 \begin{cases}
  \FutOff(\theta_1)\cup\FutOff(\theta_2),&\Complete(\theta_1),\\
  \FutOff(\theta_1),&\text{otherwise},
 \end{cases}\\
 \FutOff(\theta_1\tparop_I\theta_2)
       &=\Hide{I}{\FutOff(\theta_1)\cup\FutOff(\theta_2)},\\
 \FutOff((\Omega,\toplus(\theta_0,B)))
       &=\FutOff((\Omega,\twedge(\theta_0,B)))\\
       &=\Omega\cap\Bigl(\FutOff(\theta_0)
          \cup\bigcup_{a_i:\theta_i\in B}(\{a_i\}\cup\FutOff(\theta_i))\Bigr).
\end{aligned}
\]
The completion test has $\Complete(a)=\Complete(\tkw{skip})
=\Complete(\tmarker{t})=\mathsf{true}$ and
$\Complete(\tclock)=\Complete(X)=\mathsf{false}$. Recursion reads its body;
sequence and parallel take conjunction; ordinary and fair choice take
disjunction; either offer takes the disjunction over its default and branch
continuations.

For $\Off(E,\tau,\alpha,\tau')$, let $D_a=\{a\}\cap\Alphabet(\tau)$,
$A_B=\{a_i\mid a_i:\tau_i\in B\}$, $D_B=A_B\cap\Alphabet(\tau)$, and
$W(E,B)=\{a\in A_B\mid\Need(a)\in\Fut(E)\}$. Atomic actions publish
\[
\begin{aligned}
 \Off(E,a,a,\tkw{skip})&=\{(D_a,\emptyset)\},\\
 \Off(E,a,\epsilon,a)&=\{(D_a,\FutOff(a)\setminus D_a)\}.
\end{aligned}
\]
Both atomic clauses have an empty future field: an atomic message has no
continuation, and its current offer removes its head from the future menu.
Clock steps publish $(\emptyset,\emptyset)$. Silent clock waiting, completed
waiting, and progress steps also have empty offer pairs. Choice and recursion
publish
\[
\begin{aligned}
 \Off(E,\tau_1+\tau_2,\epsilon,\tau_i)
   &=\Off(E,\tau_1\tsym{\dotplus}_\eta\tau_2,\eta_i,\tau_i)\\
   &=\{(\emptyset,\FutOff(\tau_1)\cup\FutOff(\tau_2))\},\\
 \Off(E,\tmu X.\theta,\epsilon,\theta')
   &=\{(\emptyset,\FutOff(\theta'))\},
 \qquad\theta'=\theta[X\mapsto\tmu X.\theta].
\end{aligned}
\]
For $Z=\toplus(\tau_0,B)$ and $Y=\twedge(\tau_0,B)$, the offer clauses are
\[
\begin{aligned}
 \Off(E,Z,a_i,\tau_i)&=\{(\{a_i\},\FutOff(Z)\setminus A_B)\},\\
 \Off(E,Z,\epsilon,\tau_0)&=\{(D_B,\FutOff(Z)\setminus A_B)\},\\
 \Off(E,Z,\epsilon,Z)&=\{(W(E,B),\FutOff(Z)\setminus W(E,B))\},\\
 \Off(E,Y,A_K,C_{B,K})&=\{(D_B,\FutOff(C_{B,K})\setminus D_B)\},\\
 \Off(E,Y,\epsilon,\tau_0)&=\{(D_B,\FutOff(\tau_0)\setminus D_B)\},\\
 \Off(E,Y,\epsilon,Y)&=\{(D_B,\FutOff(Y)\setminus D_B)\}.
\end{aligned}
\]

Sequence adds the right menu when the left source passes $\Complete$, then
removes the current offer. Parallel unions component fields and hides $I$;
survivor entry hides $I$ in the surviving pair. In both cases use exactly the
component pairs and local environments in the chosen rule premises, as in
\appref{aux-appendix}. A clock step always overrides these clauses with the
empty pair.

%% file: sections/appendix/types/semantic-automaton.tex
\subsubsection{Canonical Semantic Automaton and Correspondence}
\label{app:semantic-automaton-correspondence}

We complete the canonical automaton definition from \secref{typecompat} and
prove its correspondence with type executions. We then extend this
correspondence to conforming processes.

The canonical automaton has a final state for each
completed residual. Its final states and edge alphabet are
\[
\begin{aligned}
 F_\tau&=\{\lfloor\rho\rfloor_\tau\in Q_\tau
               \mid\mathsf{Body}(\rho)=\tkw{skip}\},\\
 \Sigma_\tau&=\Env{\tau}\times2^{\mathsf{Act}_\tau}
                    \times2^{\Alphabet(\tau)}\times2^{\Alphabet(\tau)}.
\end{aligned}
\]
Here $\mathsf{Act}_\tau$ contains the message, clock, fairness, and progress
observations of $\tau$.

To state \lemref{type-automata-bisimulation} in full, we combine a semantic
step and its offer pair into a decorated step:
\[
 e\vdash\rho\xrightarrow[\langle O_C,O_F\rangle]{\alpha}\rho'.
\]
This notation requires both $\lstep{e}{\rho}{\alpha}{\rho'}$ and
$(O_C,O_F)\in\Off(e,\rho,\alpha,\rho')$.

For every well-formed regular type $\tau$,
\(\Aut_\tau\) is finite and there is a relation
\(\mathcal R_\tau\) between residuals of \(\tau\) and states of
\(\Aut_\tau\), with \((\tau,i_\tau)\in\mathcal R_\tau\), such that, for every
\((\tau_0,q)\in\mathcal R_\tau\) and \(e=(E_C,E_F)\),
\[
\begin{aligned}
 &e\vdash\tau_0\xrightarrow[\langle O_C,O_F\rangle]{\alpha}\tau_1\\
 &\quad\Longrightarrow\exists q'.\
 \atrans{q}{E_C}{E_F}{\alpha}{O_C}{O_F}{q'}\in\delta_\tau
 \land(\tau_1,q')\in\mathcal R_\tau,\\[2pt]
 &\atrans{q}{E_C}{E_F}{\alpha}{O_C}{O_F}{q'}\in\delta_\tau\\
 &\quad\Longrightarrow\exists\tau_1.\
 e\vdash\tau_0\xrightarrow[\langle O_C,O_F\rangle]{\alpha}\tau_1
 \land(\tau_1,q')\in\mathcal R_\tau.
\end{aligned}
\]
In particular, a \termref{def:type-residual}{residual} has a successor under \(e\) exactly when its related
automaton state has an edge carrying \(e\).

\begin{proof}[Proof of \lemref{type-automata-bisimulation}]
We will relate each reachable residual $\rho$ to its control
representation $\lfloor\rho\rfloor_\tau$, giving the relation
$\mathcal R_\tau$.
Recall that sharing preserves each constructor, its annotation,
and its sequence and parallel contexts.
The only repeated syntax replaced by a pointer is a recursive binder
reintroduced at a call, whose unfolding still takes an explicit silent step.

We first show that sharing preserves completion tests and future menus.
This follows by induction on their syntax-directed definitions, since
atomic and compound clauses inspect the same constructors, alphabets,
unions, and projections.
The interesting case is recursion.
Here, clock guarding prevents the pre-clock menu calculation from
following a recursive call back to its binder.
Tail position also excludes the bound variable from sequence-left
completion tests, parallel children, and conjunctive branches.
Consequently, sharing and textual substitution give the same completion
tests and menu contributions before the next clock.
More precisely, the marker-path and substitution equations of
Lemma~\ref{lem:constructor-menus} identify the contributions introduced
when a free call is replaced by its binder, while retaining the explicit
unfolding step.
These equations apply equally to the shared representation.

We will also need substitution and re-signing to commute with
representation.
For substitution, we graft the shared binder at each free tail
occurrence, consistently renaming bound variables.
This leaves the surrounding constructor contexts unchanged.
For re-signing, the incoming alphabet is passed through ordinary
children, with $\Pi_I$ applied at each parallel node.
Since the shared graph retains the annotation at each represented
occurrence, it performs the same replacements as the signed residual.
Here, sharing is by the annotated control term, including its context.

With these properties established, consider a decorated step in either
representation.
We proceed by induction on its derivation.
Direct rules inspect the same body, guard, and environment, so their
current offers agree, while the menu property identifies their future
offers.
For sequence, the same left step and completion test determine whether
to rebuild the sequence or return its right child.
Parallel similarly preserves the child steps, their chosen pairs,
local environments, and synchronization.
Fusion returns signed $\tkw{skip}$, while survivor entry uses the same
re-signing operation.
Conjunctive firing preserves the enabled set $K$, label $A_K$,
and generated product.
Finally, recursion uses the substitution property to match the
explicit unfolding step.
Thus, each case preserves the complete decoration and produces targets
with the same control representation.
This establishes step preservation in both directions, including when
distinct residuals share a control state.
The definition of $\delta_\tau$ now gives the required correspondence.
The initial pair is related, and final-state agreement follows from
preservation of the $\tkw{skip}$ constructor.

It remains to show that the control representation is finite.
By the rank argument of Appendix~\ref{app:constructor-automaton},
written constructors and generated conjunctive products give finitely
many context templates.
Moreover, recursive calls occur only in tail position.
Unfolding therefore introduces no additional pending sequence context,
parallel component, or conjunctive branch around a call, but returns
to a shared binder header.
Hence, these contexts cannot grow without bound through recursion.
Since annotations range over subsets of the fixed finite message
alphabet, there are only finitely many annotated control terms.
The environment, observation, and offer alphabets are also finite,
establishing finiteness of the automaton.
\end{proof}

\myparagraph{Process--automaton correspondence}
\label{app:process-automaton-correspondence}
The full statement of \lemref{process-automaton-correspondence} is as follows.

Suppose \(P\models\tau\) and $\tau$ is well formed. There is a relation \(\mathcal C\) between process
residuals and states of \(\Aut_\tau\), containing \((P,i_\tau)\), such that
for every \((P_0,q)\in\mathcal C\) and process step
\[
  \lstep{E}{P_0}{\alpha}{P_1},
\]
there are \(q',O_C,O_F\) for which
\[
  \atrans{q}{\erase(E)_C}{\erase(E)_F}{\erase(\alpha)}{O_C}{O_F}{q'}
  \qquad\text{and}\qquad
  (P_1,q')\in\mathcal C.
\]
At every related pair, $P_0$ has a step under $E$ exactly when $q$
has an edge under $\erase(E)$. Thus every ordinary process execution prefix has a matching automaton prefix
with the same erased environments and observations.

The same statement applies to a fixed group by taking
\(P=Q_0=P_1\parallel_{I_1}\cdots\parallel_{I_{n-1}}P_n\) and
\(\tau=\Theta_0\), provided the composite conformance
\(Q_0\models\Theta_0\) has been established. For types accepted by the
data-oblivious checker, \corref{compositional-conformance} establishes this
premise from the component checks.

\begin{proof}[Proof of \lemref{process-automaton-correspondence}]
We will construct the required witness by composing the conformance
relation $R$ with $\mathcal R_\tau$ through their shared type residual:
\[
 \mathcal C=\{(P_0,q)\mid\exists\rho.
       (P_0,\rho)\in R\land(\rho,q)\in\mathcal R_\tau\}.
\]

Let us show that $\mathcal C$ preserves both steps and availability.
For a process step, conformance supplies a matching type step under
the erased environment. Every type step admits a decoration, either
directly from its rule or from its component pairs.
The forward automaton correspondence then gives the required edge
with a related target.

Conversely, consider an automaton edge under $\erase(E)$.
The backward correspondence gives a type step, and availability
reflection supplies a process step under $E$.
Together, these establish readiness equivalence at every related pair.

We can now match finite process executions by induction.
If an execution becomes stuck, both sides are stuck under the same
erased environment and record the same erased failure entry and suffix.
Finally, the group result follows by applying this argument to the
established composite conformance $Q_0\models\Theta_0$.
\end{proof}

%% file: sections/appendix/types/constructor-automaton.tex
\subsubsection{Constructor-Based Automaton}
\label{sec:tc-appendix}
\label{app:constructor-automaton}

\begingroup
\raggedbottom
We construct the automaton by induction on the type. This construction is
used in our implementation and yields an automaton bisimilar to the type
semantics. Each constructor defines its states, entry, final states, and
transitions:
\[
 \Aut_\tau=(Q_\tau,\Sigma_\tau,\delta_\tau,i_\tau,F_\tau),\qquad
 \Sigma_\tau=\Env{\tau}\times2^{\mathsf{Act}_\tau}
                     \times2^\Omega\times2^\Omega.
\]
Here $\tau=(\Omega,\nu)$, and $\mathsf{Act}_\tau$ contains its message,
clock, fairness, and progress observations. Open fragments also use temporary
recursion-variable labels, which are removed by their binders.

An edge records its environment, observation, and current/future offer pair:
\[
 q\xrightarrow[e]{\alpha,O_C,O_F}q'.
\]
Here $e=(E_C,E_F)\in\Env{\tau}$. The empty observation is $\epsilon$,
and singleton observations appear without braces. We copy child state sets
disjointly, along with their menu references, and omit the injections.

The menu $\mathsf{OF}_\tau(q)$ approximates messages offered from $q$ before
the next clock. A dependency $(q,r)\in L_\tau$ includes the menu of $r$ in
that of $q$ without adding a transition.

\newcounter{pactautcase}
\newcommand{\pactautcase}[2][]{%
 \par\medskip\refstepcounter{pactautcase}%
 \noindent\textbf{\arabic{pactautcase}.}\quad $#2$\quad #1\par\nobreak}
\newcommand{\pactautedge}[6]{%
 \qst{#1}\xrightarrow[\aenv{#2}]{\aact{#3,\langle#4,#5\rangle}}\qst{#6}}

An unindexed $\mathsf{OF}$ denotes the menu of the containing fragment.
Completed states wait for the next clock using an idle self-loop.

\pactautcase[$Q_\tau=\{q_0\},\quad i_\tau=q_0,\quad F_\tau=\{q_0\}.$]{\nu=\tkw{skip}}
\[
\begin{aligned}
 \delta_\tau={}&
 \bigl\{
 \pactautedge{q_0}{e}{\epsilon}{\emptyset}{\mathsf{OF}(q_0)}{q_0}
 \ \bigm|\ E_C\neq\{\#\}
 \bigr\}. &&\text{(idle)}
\end{aligned}
\]

\pactautcase[$Q_\tau=\{q_0,q_f\},\quad i_\tau=q_0,\quad F_\tau=\{q_f\}.$]{\nu=a}
Let $D_a=\{a\}\cap\Omega$ be the declared head. The transitions are
\[
\begin{aligned}
 \delta_\tau={}&
 \bigl\{
 \pactautedge{q_0}{e}{a}{D_a}{\mathsf{OF}(q_f)\setminus D_a}{q_f}
 \ \bigm|\ \overline a\in E_C
 \bigr\} &&\text{(fire)}\\
 &{}\cup
 \bigl\{\, \pactautedge{q_0}{e}{\epsilon}{D_a}{\mathsf{OF}(q_0)\setminus D_a}{q_0}\ \bigm|\ \overline a\notin E_C,\ E_C\neq\{\#\}\,\bigr\} &&\text{(wait)}\\
 &{}\cup
 \bigl\{\,\pactautedge{q_f}{e}{\epsilon}{\emptyset}{\mathsf{OF}(q_f)}{q_f}
 \ \bigm|\ E_C\neq\{\#\}\,\bigr\}. &&\text{(idle)}
\end{aligned}
\]

\pactautcase[$Q_\tau=\{q_0,q_f\},\quad i_\tau=q_0,\quad F_\tau=\{q_f\}.$]{\nu=\tclock}
\[
\begin{aligned}
 \delta_\tau={}&
 \bigl\{
 \pactautedge{q_0}{e}{\#}{\emptyset}{\emptyset}{q_f}
 \ \bigm|\ E_C=\{\#\}
 \bigr\} &&\text{(clock)}\\
 &{}\cup
 \bigl\{
 \pactautedge{q_0}{e}{\epsilon}{\emptyset}{\mathsf{OF}(q_0)}{q_0}
 \ \bigm|\ E_C\neq\{\#\}
 \bigr\} &&\text{(wait)}\\
 &{}\cup
 \bigl\{\,\pactautedge{q_f}{e}{\epsilon}{\emptyset}{\mathsf{OF}(q_f)}{q_f}
 \ \bigm|\ E_C\neq\{\#\}\,\bigr\}. &&\text{(idle)}
\end{aligned}
\]
\pactautcase[$Q_\tau=\{q_0,q_f\},\quad i_\tau=q_0,\quad F_\tau=\{q_f\}.$]{\nu=\tmarker{t}}
\[
\begin{aligned}
 \delta_\tau={}&
 \bigl\{
 \pactautedge{q_0}{e}{\$_t}{\emptyset}{\mathsf{OF}(q_f)}{q_f}
 \ \bigm|\ E_C\neq\{\#\}
 \bigr\} &&\text{(progress)}\\
 &{}\cup
 \bigl\{\,\pactautedge{q_f}{e}{\epsilon}{\emptyset}{\mathsf{OF}(q_f)}{q_f}
 \ \bigm|\ E_C\neq\{\#\}\,\bigr\}. &&\text{(idle)}
\end{aligned}
\]

\pactautcase{\nu=\tau_1;\tau_2}
We redirect transitions that complete the left type to the right entry,
and update the future offers of left transitions:
\[
 Q_\tau=(Q_1\uplus Q_2)\setminus(F_1\setminus\{i_1\}),\qquad
 i_\tau=i_1,\qquad F_\tau=F_2.
\]
Let $d(q)=i_2$ for $q\in F_1$, and $d(q)=q$ otherwise.
For a left transition with source $q$, observation $\alpha$, and offers
$(O_C,O_F)$, the updated future offer is
\[
 O'_F=
 \begin{cases}
  \emptyset,&\#\in\alpha,\\
  (O_F\cup\mathsf{OF}(i_2))\setminus O_C,&\kappa_1(q),\\
  O_F\setminus O_C,&\text{otherwise}.
 \end{cases}
\]
Here $\kappa_1(q)=\Complete(R_{\tau_1}(q))$ when the signed residual
$R_{\tau_1}(q)$ is defined, and $\mathsf{false}$ otherwise.
The state correspondence below defines these residuals. The transitions are
\[
\begin{aligned}
 \delta_\tau={}&
 \bigl\{\, \pactautedge{q}{e}{\alpha}{O_C}{O'_F}{d(q')}\ \bigm|\ \pactautedge{q}{e}{\alpha}{O_C}{O_F}{q'}\in\delta_1,\
             q\in Q_\tau\,\bigr\} &&\text{(left)}\\
 &{}\cup\delta_2. &&\text{(right)}
\end{aligned}
\]

Left dependencies $(p,r)$ become $(p,d(r))$ when $p$ is retained;
right dependencies are unchanged. We also add $(q,i_2)$ to $L_\tau$ for
each retained left state $q$ with $\kappa_1(q)$.

\pactautcase{\nu=\tau_1+\tau_2}
Let $i_\tau$ be fresh, $Q_\tau=\{i_\tau\}\uplus Q_1\uplus Q_2$, and
$F_\tau=F_1\uplus F_2$. The entry takes a silent step to either child:
\[
\begin{aligned}
 \delta_\tau={}&
 \bigl\{
 \pactautedge{i_\tau}{e}{\epsilon}{\emptyset}{\mathsf{OF}(i_\tau)}{i_b}
 \ \bigm|\ b\in\{1,2\},\ E_C\neq\{\#\}
 \bigr\} &&\text{(choose)}\\
 &{}\cup\delta_1\cup\delta_2. &&\text{(inherit)}
\end{aligned}
\]
\pactautcase{\nu=\tau_1\dotplus_\eta\tau_2}
The state sets and entry are the same as for ordinary choice. Each initial
transition also records the selected outcome:
\[
\begin{aligned}
 \delta_\tau={}&
 \bigl\{
 \pactautedge{i_\tau}{e}{\eta_{\mathsf{true}}}{\emptyset}{\mathsf{OF}(i_\tau)}{i_1}
 \ \bigm|\ E_C\neq\{\#\}
 \bigr\} &&\text{(true)}\\
 &{}\cup
 \bigl\{
 \pactautedge{i_\tau}{e}{\eta_{\mathsf{false}}}{\emptyset}{\mathsf{OF}(i_\tau)}{i_2}
 \ \bigm|\ E_C\neq\{\#\}
 \bigr\} &&\text{(false)}\\
 &{}\cup\delta_1\cup\delta_2. &&\text{(inherit)}
\end{aligned}
\]

\pactautcase{\nu=\tau_1\parallel_I\tau_2}
Parallel transitions combine child transitions whose offers agree with
their local environments:
\[
\mathsf{SyncEdges}_I(\tau_1,\tau_2)=
\left\{\begin{aligned}
 &\pactautedge{(q_1,q_2)}{e}{\alpha}{O_C}{O_F}{(q'_1,q'_2)}
       \mid \Sync_I(\alpha_1,\alpha_2),\\
 &\quad\pactautedge{q_j}{e_j}{\alpha_j}{O_{Cj}}{O_{Fj}}{q'_j}
           \in\delta_j\quad(j=1,2),\\
 &\quad e_1\in\LocalEnv(e,I,\tau_1,O_{C2},O_{F2}),\\
 &\quad e_2\in\LocalEnv(e,I,\tau_2,O_{C1},O_{F1}),\\
 &\quad\alpha=\alpha_1\unionI\alpha_2,\\
 &\quad O_C=\Hide{I}{O_{C1}\cup O_{C2}},\
       O_F=\Hide{I}{O_{F1}\cup O_{F2}}
\end{aligned}\right\}.
\]
Here $I$ is the declared interface. We include all product sources, even
unreachable ones, since they may later be copied into a survivor.
We add a copy for each surviving component and a final state:
\[
\begin{aligned}
 Q_\tau={}&(Q_1\times Q_2)
   \uplus\{\mathsf{sur}_1(q)\mid q\in Q_1\}
   \uplus\{\mathsf{sur}_2(q)\mid q\in Q_2\}\uplus\{q_d\},\\
 i_\tau={}&(i_1,i_2),\qquad F_\tau=\{q_d\}.
\end{aligned}
\]
Each synchronized transition keeps its source and edge fields.
Its target is given below:
\[
\begin{array}{lll}
 \text{source}&\text{target}&\text{new target}\\[2pt]
 q_1\notin F_1,\ q_2\notin F_2
   &q'_1\notin F_1\text{ or }q'_2\notin F_2 &(q'_1,q'_2)\\
 \text{any}&q'_1\in F_1,\ q'_2\in F_2&q_d\\
 q_1\in F_1,\ q_2\notin F_2
   &q'_1\in F_1,\ q'_2\notin F_2&\mathsf{sur}_2(q'_2)\\
 q_2\in F_2,\ q_1\notin F_1
   &q'_2\in F_2,\ q'_1\notin F_1&\mathsf{sur}_1(q'_1).
\end{array}
\]
Only the listed cases contribute transitions. Survivor entry still uses
both child transitions, including the completed child's non-clock silent
transition with empty offers. A pair of final child states must still step
to $q_d$.

Subsequent survivor transitions come from $\tau_j[\Omega]$.
Re-signing preserves its state sets, constructor tags, and final sets.
The transitions are
\[
\begin{aligned}
 \delta_\tau={}&
 \left\{\begin{aligned}
  &\pactautedge{s}{e}{\alpha}{O_C}{O_F}{r}
   \ \mid\ \pactautedge{s}{e}{\alpha}{O_C}{O_F}{t}
       \in\mathsf{SyncEdges}_I(\tau_1,\tau_2),\\
  &\qquad r\text{ is the new target for }s,t\text{ in the table above}
 \end{aligned}\right\} &&\text{(synchronize)}\\
 &{}\cup\bigcup_{j\in\{1,2\}}
 \left\{\begin{aligned}
  &\pactautedge{\mathsf{sur}_j(q)}{e}{\alpha}{O_C}{O_F}{\mathsf{sur}_j(q')}
   \ \mid\ \pactautedge{q}{e}{\alpha}{O_C}{O_F}{q'}
       \in\delta_{\tau_j[\Omega]},\\
  &\qquad q'\notin F_{\tau_j[\Omega]}
 \end{aligned}\right\} &&\text{(survive)}\\
 &{}\cup\bigcup_{j\in\{1,2\}}
 \left\{\begin{aligned}
  &\pactautedge{\mathsf{sur}_j(q)}{e}{\alpha}{O_C}{O_F}{q_d}
   \ \mid\ \pactautedge{q}{e}{\alpha}{O_C}{O_F}{q'}
       \in\delta_{\tau_j[\Omega]},\\
  &\qquad q'\in F_{\tau_j[\Omega]}
 \end{aligned}\right\} &&\text{(finish)}\\
 &{}\cup
 \bigl\{\,\pactautedge{q_d}{e}{\epsilon}{\emptyset}{\mathsf{OF}(q_d)}{q_d}
   \ \bigm|\ E_C\neq\{\#\}\,\bigr\}. &&\text{(idle)}
\end{aligned}
\]
For product menus, let
$K_\tau(q_1,q_2)=\Hide{I}{\mathsf{OF}_{\tau_1}(q_1)
                                  \cup\mathsf{OF}_{\tau_2}(q_2)}$,
using the original child menus, including those at final states.
Products add no menu dependencies. Survivors inherit the dependencies and
fixed contributions of the re-signed child; non-final sources carry the
survivor tag, and final targets are redirected to $q_d$.

\pactautcase{\nu=\toplus(\tau_0,B)}
Let $B=\{a_i:\tau_i\mid i\in J_B\}$, where $J_B=\{1,\ldots,n\}$,
and let $A_B=\{a_i\mid i\in J_B\}$. We will use
$D_B=A_B\cap\Omega$ for the declared heads and
$W=\overline{E_F}\cap A_B$ for heads whose complements remain in the future
environment.

\[
 Q_\tau=\{i_\tau\}\uplus Q_0\uplus\biguplus_{i\in J_B}Q_i,\qquad
 F_\tau=F_0\uplus\biguplus_{i\in J_B}F_i.
\]
\[
\begin{aligned}
 \delta_\tau={}&
 \left\{\begin{aligned}
  &\pactautedge{i_\tau}{e}{a_i}{\{a_i\}}
       {\mathsf{OF}(i_\tau)\setminus A_B}{i_i}
    \ \mid\ i\in J_B,\ \overline{a_i}\in E_C,\\
  &\qquad E_F\cap\overline{A_B}=\emptyset
 \end{aligned}\right\} &&\text{(fire)}\\
 &{}\cup
 \bigl\{\,\pactautedge{i_\tau}{e}{\epsilon}{W}
       {\mathsf{OF}(i_\tau)\setminus W}{i_\tau}
    \ \bigm|\ E_F\cap\overline{A_B}\neq\emptyset,\ E_C\neq\{\#\}\,\bigr\} &&\text{(wait)}\\
 &{}\cup
 \left\{\begin{aligned}
  &\pactautedge{i_\tau}{e}{\epsilon}{D_B}
       {\mathsf{OF}(i_\tau)\setminus A_B}{i_0}
    \ \mid\ E_C\cap\overline{A_B}=E_F\cap\overline{A_B}=\emptyset,\\
  &\qquad E_C\neq\{\#\}
 \end{aligned}\right\} &&\text{(default)}\\
 &{}\cup\delta_0\cup\bigcup_{i\in J_B}\delta_i. &&\text{(inherit)}
\end{aligned}
\]
The root depends on the default and every branch entry:
$\{(i_\tau,i_0)\}\cup\{(i_\tau,i_i)\mid i\in J_B\}\subseteq L_\tau$.

\pactautcase{\nu=\twedge(\tau_0,B)}
Each nonempty $K\subseteq J_B$ uses a separate copy of $C_{B,K}$.
Let $(Q_K,\Sigma_K,\delta_K,i_K,F_K)$ be its automaton and
$A_K=\{a_i\mid i\in K\}$. Then
\[
 Q_\tau=\{i_\tau\}\uplus Q_0\uplus\biguplus_{\emptyset\neq K\subseteq J_B}Q_K,
 \qquad F_\tau=F_0\uplus\biguplus_{\emptyset\neq K\subseteq J_B}F_K.
\]
The default and subset copies have distinct tags, retained under re-signing
even when a subset has no incoming firing edge.
\[
\begin{aligned}
 \delta_\tau={}&
 \bigl\{\,\pactautedge{i_\tau}{e}{A_K}{D_B}
       {\mathsf{OF}(i_K)\setminus D_B}{i_K}
    \ \bigm|\ K=\{i\in J_B\mid\overline{a_i}\in E_C\}\neq\emptyset\,\bigr\} &&\text{(fire)}\\
 &{}\cup
 \left\{\begin{aligned}
  &\pactautedge{i_\tau}{e}{\epsilon}{D_B}
       {\mathsf{OF}(i_\tau)\setminus D_B}{i_\tau}
    \ \mid\ E_C\cap\overline{A_B}=\emptyset,\\
  &\qquad E_F\cap\overline{A_B}\neq\emptyset,\ E_C\neq\{\#\}
 \end{aligned}\right\} &&\text{(wait)}\\
 &{}\cup
 \left\{\begin{aligned}
  &\pactautedge{i_\tau}{e}{\epsilon}{D_B}
       {\mathsf{OF}(i_0)\setminus D_B}{i_0}
    \ \mid\ E_C\cap\overline{A_B}=E_F\cap\overline{A_B}=\emptyset,\\
  &\qquad E_C\neq\{\#\}
 \end{aligned}\right\} &&\text{(default)}\\
 &{}\cup\delta_0\cup\bigcup_{\emptyset\neq K\subseteq J_B}\delta_K.
   &&\text{(inherit)}
\end{aligned}
\]
The root depends on the default and every residual-product entry:
\[
 \{(i_\tau,i_0)\}\cup
 \{(i_\tau,i_K)\mid\emptyset\neq K\subseteq J_B\}\subseteq L_\tau.
\]

\pactautcase[$Q_\tau=\{q_0,q_1\},\quad i_\tau=q_0,\quad F_\tau=\emptyset.$]{\nu=X}
A free recursion variable uses a temporary edge to $q_1$, which has no
outgoing edge or annotation:
\[
\begin{aligned}
 \delta_\tau={}&
 \bigl\{\,\pactautedge{q_0}{e}{X}{\emptyset}{\mathsf{OF}(q_1)}{q_1}
   \ \bigm|\ E_C\neq\{\#\}\,\bigr\}. &&\text{(variable)}
\end{aligned}
\]

\pactautcase{\nu=\mu X.\theta}
The binder redirects the temporary $X$ edges to the body entry, retaining
the unfolding step. Let $i_\tau$ be fresh,
$Q_\tau=\{i_\tau\}\uplus Q_\theta$, and $F_\tau=F_\theta$.
\[
\begin{aligned}
 \delta_\tau={}&
 \bigl\{\,\pactautedge{i_\tau}{e}{\epsilon}{\emptyset}{\mathsf{OF}(i_\theta)}{i_\theta}
   \ \bigm|\ E_C\neq\{\#\}\,\bigr\} &&\text{(enter)}\\
 &{}\cup
 \bigl\{\,\pactautedge{q}{e}{\alpha}{O_C}{O_F}{q'}\in\delta_\theta
   \ \bigm|\ \alpha\neq\{X\}\,\bigr\} &&\text{(inherit)}\\
 &{}\cup
 \bigl\{\,\pactautedge{q}{e}{\epsilon}{\emptyset}{\mathsf{OF}(i_\theta)}{i_\theta}
   \ \bigm|\ \exists q'.\
   \pactautedge{q}{e}{X}{\emptyset}{\emptyset}{q'}\in\delta_\theta\,\bigr\}.
   &&\text{(recur)}
\end{aligned}
\]
By tail recursion, a variable transition cannot run in parallel with a
sibling. Closed types have no remaining variable labels.

\myparagraph{Computing future menus}
After constructing the states and guarded transitions, we compute future
menus by iteration:
\[
\begin{aligned}
 \mathsf{OF}_\tau&=\mu U.\Phi_\tau(U),\\
 \Phi_\tau(U)(q)&=\Omega\cap\Bigl(K_\tau(q)
       \cup\bigcup_{(q,r)\in L_\tau}U(r) \qquad \cup
   \bigcup_{\pactautedge{q}{e}{\alpha}{O_C}{O_F}{q'}\in\delta_\tau,
                  \ \#\notin\alpha}(O_C\cup U(q'))\Bigr).
\end{aligned}
\]
Here $U$ assigns an offer approximation to each state. The three terms
collect its fixed contributions, the menus on which it depends, and the
current offers and successor menus of its non-clock transitions.

We start with empty menus and apply $\Phi_\tau$ until none changes.
Edge guards, current offers, dependencies, and fixed contributions remain
fixed during iteration; guards use already computed child menus.
The resulting menus fill the future fields of the transitions.
All atomic menus are empty except $\mathsf{OF}_a(q_0)=D_a$.

\myparagraph{Finiteness and correspondence}
To show termination, we order recursive calls by the lexicographic rank
\[
 \mathsf{Rank}(\tau)=
  (\text{number of conjunctive offers in }\tau,\ |\mathsf{Body}(\tau)|).
\]
Written children have smaller bodies, and a generated conjunctive product
removes its enclosing offer. Re-signing preserves the rank, so survivor
calls also decrease it. Each constructor forms finite unions or products
and adds finitely many transitions. Menu iteration is monotone on a finite
lattice. Thus the construction terminates with a finite automaton.

\input{sections/appendix/types/constructor-correspondence}
\par\endgroup

%% file: sections/appendix/types/constructor-correspondence.tex
\myparagraph{Residuals represented by states}
To prove correspondence, we will relate constructed states to their semantic
residuals. Survivor states may represent residuals that differ on unused
alphabet declarations, since nested completions can re-sign them in
different orders.

Let $\rho\equiv_{\mathsf{live}}\theta$ mean that the two residuals have the
same control syntax and outer alphabet, and corresponding signed subterms
$(\Omega_1,\rho_1)$ of $\rho$ and $(\Omega_2,\rho_2)$ of $\theta$ satisfy
\[
 (\Omega_1\mathbin\triangle\Omega_2)
       \cap\mathsf{MsgLabel}(\rho_1)=\emptyset.
\]
Here $\triangle$ denotes symmetric difference, and $\mathsf{MsgLabel}$
collects the polarized actions in the body and its continuations.

Let $\mathcal C_\sigma$ relate residuals to constructed states.
A fresh root represents
$\sigma$, and an atomic final state represents signed $\tkw{skip}$.
An ordinary child copy uses its child's relation. The remaining clauses
retain the surrounding context:
\[
\begin{aligned}
 ((\Omega,\rho_1;\sigma_2),\iota_1(q_1))\in\mathcal C_{\sigma_1;\sigma_2}
 &\quad\text{if }(\rho_1,q_1)\in\mathcal C_{\sigma_1},\\
 ((\Omega,\rho_1\parallel_I\rho_2),(q_1,q_2))\in\mathcal C_\sigma
 &\quad\text{if }(\rho_j,q_j)\in\mathcal C_{\sigma_j},\\
 (\rho,\mathsf{sur}_j(q))\in\mathcal C_\sigma
 &\quad\text{if }(\theta,q)\in\mathcal C_{\sigma_j[\Omega]},\quad
              \rho\equiv_{\mathsf{live}}\theta,\\
 ((\Omega,\tkw{skip}),q_d)&\in\mathcal C_{\sigma_1\parallel_I\sigma_2},\\
 (\rho[X\mapsto S],\iota(q))\in\mathcal C_S
 &\quad\text{if }S=(\Omega,\mu X.\sigma_0),\quad
                   (\rho,q)\in\mathcal C_{\sigma_0}.
\end{aligned}
\]
Only retained left states occur in the sequence clause. The survivor clause
requires $q$ to be non-final and $\rho$ to satisfy $\mathsf{AlphaWF}$ and
$\mathsf{IfaceScoped}$. Conjunctive child $K$ uses
$\mathcal C_{C_{B,K}}$ with its subset tag. Temporary variable targets have
no related residual. Omitting the survivor variation gives the canonical
annotation $R_\sigma(q)$ used in the construction. An annotation is
undefined when a required child annotation is undefined.

The proofs use open fragments in a binder context $\Gamma$. Their free
variables lie in $\Gamma$, their binders are fresh for $\Gamma$, and each
free variable satisfies its inherited tail condition. Descending under
$\mu X$ adds $X$ to this context. At a closed type the context is empty.

\begin{lemma}[Structural invariants]
\label{lem:constructor-structural-invariants}
Every recursive compiler call is well formed in its inherited context and
has smaller rank. Re-signing preserves state tags, final sets, and completion
flags. For $(\rho,q)\in\mathcal C_\sigma$,
\[
 \Alphabet(\rho)=\Alphabet(\sigma),\qquad
 q\in F_\sigma\Longleftrightarrow\mathsf{Body}(\rho)=\tkw{skip},\qquad
 \kappa_\sigma(q)=\Complete(\rho).
\]
Represented residuals satisfy alphabet well-formedness, interface scoping,
tail recursion, and message ownership.
\end{lemma}
\begin{proof}
The rank argument above gives the required decrease. Let us first check
well-formedness. A generated $C_{B,K}$ uses each continuation once.
Branch compatibility supplies its
alphabets, conjunctive ownership separates its factors, and the
head-exclusion condition makes its retained singleton offers well formed.
The remaining requirements are inherited, including closure, guarding,
tail position, and identifier linearity.

A survivor call only re-signs a proper child. To check that this preserves
well-formedness, consider a node with interface $J$ and incoming alphabet
$A$. Interface scoping gives $A|_J=\emptyset$ and separates $J$ from
descendant interfaces. Each child receives $(\Omega_j|_J)\cup A$, preserving
its interface polarities and coverage. Removing $J$ leaves $A$ on both
sides, as required by the parallel alphabet conditions. The same argument
preserves scoping in the children; other structural conditions are unchanged.

For the state properties, we proceed by induction on the compiler rank.
Re-signing preserves constructor tags and therefore the state and final sets,
including those of every generated nonempty subset. It also leaves
$\Complete$ unchanged.

The remaining residual properties follow from the clauses of
$\mathcal C$. The root and atomic cases are immediate. Ordinary copies retain their child's
alphabet; parallel either rebuilds the parent or uses its re-signed child.
Tail recursion excludes substitution below parallel nodes and conjunctive
branches, so substituted variables retain their binder's alphabet. These
clauses preserve the structural and inherited tail conditions. Final
states represent completed children, an atomic final, or $q_d$; retained
sequence-left and product states are not final. Substitution preserves
$\tkw{skip}$, and related residuals have the same control syntax as
$R_\sigma(q)$, giving the same completion flag.
\end{proof}

\begin{lemma}[Annotation and survivor transport]
\label{lem:constructor-survivor-transport}
For an alphabet-well-formed, interface-scoped residual and
$A\subseteq\Alphabet(\rho)$,
\[
 \FutOff(\rho[A])=\FutOff(\rho)\cap A,\qquad
 \Complete(\rho[A])=\Complete(\rho).
\]
Live-equivalent residuals have equal menus and completion tests. When they
also satisfy tail recursion and message ownership, their decorated steps
match in both directions, with live-equivalent targets. Moreover,
\[
 (\rho,q)\in\mathcal C_\sigma
 \ \Longrightarrow\ \exists\theta.\
 (\theta,q)\in\mathcal C_{\sigma[A]}
       \land\rho[A]\equiv_{\mathsf{live}}\theta.
\]
\end{lemma}
\begin{proof}
We will first show the menu projection by structural induction. The ordinary
cases follow from distribution of intersection over unions and the
unchanged completion test. At parallel, re-signing retains each child's
interface actions and projects its boundary actions to $A$. Hiding the
interface therefore gives the original parent menu intersected with $A$.

Consider now live-equivalent residuals.
Their common control syntax gives the same completion test, and their
annotations agree on every action contributing to a menu, so their menus
also agree. To match their steps, we strengthen the induction to subterms with
possibly different outer alphabets and environments agreeing on complements
used by their bodies. The direct cases are immediate. Sequence combines
matching child pairs with the same completion flag and right menu. At
parallel, $\LocalEnv$ gives matching environments on used actions; we apply
the child hypotheses and hide the same interface. Ownership separates
non-interface uses, while the alphabet conditions fix interface polarities.
The same projection preserves agreement at survivor targets. Recursion
substitutes at tail positions, where the variable has its binder's alphabet.
Thus steps preserve live equivalence, the outer alphabet, and interface
scoping. Equal outer alphabets give the stated step correspondence.

It remains to establish state transport. Here, we proceed by induction on
the clauses of $\mathcal C_\sigma$. Ordinary copies and sequences follow from the child hypotheses. At a
product, the children and their represented residuals have equal outer
alphabets, so re-signing applies the same interface projection. At a
survivor, we project its live-equivalent witness using the child hypothesis
and use transitivity. The structural invariants ensure that the resulting
witness satisfies the required alphabet and scoping conditions. Recursion
uses the same tail-position substitution argument.
\end{proof}

\begin{lemma}[Recursion substitution]
\label{lem:constructor-recursion-substitution}
Let $S=(\Omega,\mu X.\sigma_0)$ be well formed in context $\Gamma$, and
$(\rho,q)\in\mathcal C_{\sigma_0}$. If $\mathsf{Body}(\rho)\neq X$,
substitution of $S$ transports steps in both directions, preserving the
environment, label, and current offer, and substituting their targets.
Future fields use the same offer clauses with substituted menu arguments.
If $\mathsf{Body}(\rho)=X$, the substituted residual is $S$ and its only
steps are non-clock silent unfoldings to $\sigma_0[X\mapsto S]$.
\end{lemma}
\begin{proof}
Let us first consider residuals whose body is not $X$.
The inherited tail condition excludes $X$ from sequence-left operands,
parallel components, and conjunctive branches. The direct cases therefore
follow by substitution in the selected target. In sequence, the left step
and completion flag stay unchanged, and substitution acts on the right
continuation.

For a nested binder $Y$, freshness gives $Y\neq X$ and
$Y\notin\mathsf{FV}(S)$, so capture-avoiding substitution gives
\[
 (\theta[Y\mapsto\mu Y.\theta])[X\mapsto S]
 =\theta[X\mapsto S][Y\mapsto(\mu Y.\theta)[X\mapsto S]].
\]
This identifies the unfolding targets. The offer clauses preserve current
offers and substitute only their menu arguments.

It remains to consider residuals with body $X$. Substitution returns $S$
itself. By induction on constructor tags,
a state representing $X$ has only temporary $X$ edges. Replacing them with
silent backedges therefore gives exactly the explicit unfolding steps.
\end{proof}

\myparagraph{Menu and step correspondence}
We will prove menu and step correspondence together by induction on the
compiler rank. At each rank, the child hypotheses match environments,
labels, current offers, and targets. These determine the menus, since menu
iteration does not read future fields. Equal menus then identify the future
fields.

\begin{lemma}[Menus agree]
\label{lem:constructor-menus}
If $(\rho,q)\in\mathcal C_\sigma$, then
$\mathsf{OF}_\sigma(q)=\FutOff(\rho)$.
\end{lemma}
\begin{proof}
Let us begin with two properties of non-clock steps and menu dependencies.

First, if the target passes the completion test, then its source also
passes the test. For execution steps, this follows from the completion
clauses and preservation under re-signing. A sequence dependency is added
only when the left residual passes the test. For conjunctive offers, if
$C_{B,K}$ passes the test, then so does a continuation of a branch in $K$,
since $K$ is nonempty. The enclosing offer therefore passes its completion
test as well. Clock-guarded unfolding preserves the same pre-clock tests.

Second, a non-clock step's current offer and the future menu of its target
are both contained in the future menu of its source. This follows by
induction.

These properties let us compare finite menu derivations with structural
future offers. The atomic and choice cases follow directly from their
definitions. In particular, an undeclared message has only an empty-current waiting loop,
whose least menu is empty. For exclusive offers, dependencies supply every
branch and default menu, including continuations of undeclared heads.
Each declared head contributes through a firing edge under its complementary
immediate environment and an empty future environment. Their union is
exactly the source menu.

For a conjunctive offer, the generated product has menu
\[
 \FutOff(C_{B,K})=\Omega\cap\Bigl(
  \bigcup_{i\in J_B}\FutOff(\tau_i)
    \cup\{a_i\mid i\in J_B\setminus K\}\Bigr).
\]

The dependency for $K=J_B$ supplies all continuation menus, while the
default dependency supplies the default menu. Under the empty environment,
the default edge supplies every declared head. These contributions give
exactly the source menu; no generated child or execution edge adds anything
outside it. If $J_B$ is empty, only the default menu is needed.

For a retained sequence-left state, the required menu is
\[
 \mathsf{OF}_{\sigma_1}(q_1)\ \cup\
 \begin{cases}
  \mathsf{OF}_{\sigma_2}(i_2),&\kappa_1(q_1),\\
  \emptyset,&\text{otherwise}.
 \end{cases}
\]

A contribution can cross into the right fragment only through an actual
completion, a dependency redirected into a final state, or the new sequence
dependency. By the completion property above, each requires
$\kappa_1(q_1)$. Conversely, copying finite left-menu derivations loses
nothing, since deleted final states have empty menus. When the flag holds,
the new dependency supplies the full right menu. This proves both
inclusions, even when the retained left entry is final.

At a product state, the fixed contribution is the hidden union of child
menus. Child step correspondence bounds each synchronized edge by this
union. Survivor transport projects its target menu to the parent alphabet
and adds nothing outside the union. A survivor inherits the re-signed
child's contributions and dependencies, with live equivalence identifying
its represented menu. The done state has an empty menu.

The remaining case is recursion. Here, we must track whether the structural
menu reaches a free $X$.
Let $b_X(\rho)$ denote this condition. It is true at $X$, false at other
atoms, and a disjunction at choices and exclusive offers. For sequence it is
\[
 b_X(\rho_1)\lor(\Complete(\rho_1)\land b_X(\rho_2)).
\]

Tail position makes it false in parallel components and conjunctive
branches; only a conjunctive default can reach $X$. Under another binder
we descend into its body, stopping at a binder of $X$.

Induction on the preceding cases identifies $b_X(\rho)$ with a finite
non-clock edge/dependency path from $q$ to an $X$ marker source. At a nested
binder $Y$, a new backedge reaches $X$ exactly when the old path reaches
$Y$ and the body entry reaches $X$. Clock guarding prevents the body entry
from reaching $Y$ again before a clock. Cutting at the first new backedge
therefore establishes the same path property for this case.

This path property now lets us compare the menus before and after binding.
Let $V=\FutOff(\sigma_0)$ and let $U$ be the copied body-entry menu after
binding $X$. Structural substitution and finite saturation derivations give
\[
\begin{aligned}
 \FutOff(\rho[X\mapsto S])
   &=\FutOff(\rho)\cup
       \begin{cases}V,&b_X(\rho),\\\emptyset,&\text{otherwise},\end{cases}\\
 \mathsf{OF}_S(\iota(q))
   &=\mathsf{OF}_{\sigma_0}(q)\cup
       \begin{cases}U,&b_X(\rho),\\\emptyset,&\text{otherwise}.\end{cases}
\end{aligned}
\]

The second equality follows by cutting a finite derivation at its first
new backedge. Conversely, a marker path followed by a body-entry derivation
supplies the added menu. Clock guarding gives
$b_X(\sigma_0)=\mathsf{false}$, so the child hypothesis gives $U=V$.
The two equations now agree at every copied state and the fresh root.
Since all derivations are finite, they identify the least fixed point.
\end{proof}

\begin{lemma}[Constructor correspondence]
\label{lem:constructor-correspondence}
For a well-formed closed type $\tau$, $\mathcal C_\tau$ relates its initial
residual and entry state and matches decorated steps in both directions,
with related targets. Final states correspond exactly to signed
$\tkw{skip}$ residuals.
\end{lemma}
\begin{proof}
We will use the simultaneous rank induction described above to match
steps in both directions. For open fragments, the comparison includes only
labels containing no free variable. The free-variable case has no such edge or semantic step; its temporary marker is handled
at its binder. Atomic and choice cases follow directly from the rules,
and ordinary copies use the child hypothesis.

For offers, conjunctive firing selects the same enabled set $K$ and enters
the related entry of $C_{B,K}$. Resolution uses the same default or waiting
rule. For sequence, the final-state invariant identifies exactly when a
child target is $\tkw{skip}$, so both sides select the same completing or
rebuilding rule.

Let us consider parallel composition. The child hypotheses preserve the
chosen pairs, local environments, and synchronization premises.
Final-state agreement therefore selects the same product, fusion, or
survivor case. At survivor entry, the finished child takes a silent step
with empty offers, so the witness gives the surviving label and hidden
pair. Lemma~\ref{lem:constructor-survivor-transport} relates the re-signed
target to its survivor copy and, together with the child hypothesis,
matches subsequent survivor steps. Structural preservation keeps the
targets within the relation's alphabet and scoping conditions.

For recursion, the root takes one silent unfolding step into the body.
Lemma~\ref{lem:constructor-recursion-substitution} matches non-variable
body steps and their targets. A variable state has only marker edges;
their replacement gives one explicit unfolding backedge.

With these steps matched, it remains to identify their future offers.
Lemma~\ref{lem:constructor-menus} supplies matching menus. Direct cases use the same source or target menus, with exclusive commitment
and default removing all branch heads. Sequence uses the same child future,
source completion flag, and right-entry menu; parallel hides the same
child futures. Clock pairs are empty. In recursion copies,
Lemma~\ref{lem:constructor-recursion-substitution} identifies the substituted
menu arguments, and Lemma~\ref{lem:constructor-menus} identifies their values.
Thus the full offer pairs agree in both directions. The relation clauses
and structural invariants give the initial and final correspondence.
\end{proof}

\myparagraph{Equivalence with the canonical automaton}
We will relate reachable states of the two automata through a common semantic
residual:
\[
 \mathcal B_\tau=\{(\lfloor\rho\rfloor_\tau,q)\mid
     \rho\text{ reachable},\ (\rho,q)\in\mathcal C_\tau,
     \ q\text{ reachable}\}.
\]
To check this relation, consider an edge of either automaton. The canonical correspondence in
Appendix~\ref{app:semantic-automaton-correspondence} and
Lemma~\ref{lem:constructor-correspondence} give a semantic step and then a
matching edge of the other automaton. Both edges have the same environment,
label, and full offer pair, and their targets remain related. The initial
states are related, and both final tests identify $\tkw{skip}$. Thus
$\mathcal B_\tau$ is a decorated bisimulation, preserving environment
availability and execution logs.

%% file: sections/appendix/types/type-inference-refinement-and-checking.tex
\subsection{Proofs and Algorithms from \secref{type-checking-inference}}
\applabel{typeinference-optimality}

We first state how abstraction preserves alphabets and how process
well-formedness depends on the inferred type. We then define valid compositions using these requirements.
These definitions are used in the conformance and strategy-transfer proofs
that follow. Finally, we give the refinement game used by the data-oblivious
checker.

The abstraction copies every written alphabet:
$\mathsf{abs}((\Omega,R))=(\Omega,\mathsf{abs}_\Omega(R))$.
Here $\mathsf{abs}_\Omega$ applies the clauses of \secref{type-checking-inference} to the body $R$.
Each atomic type introduced for a prefix or binding carries $\Omega$.
Checking compares types at this fixed alphabet, so
$\Alphabet(\tau)=\Alphabet(P)$ whenever $\checkdo{P}{\tau}$.

\phantomsection\label{app:process-well-formedness}
Process well-formedness also requires expressions to evaluate to values:
\[
 \mathsf{WFProc}(P)\ \Longleftrightarrow\
 \mathsf{DataWF}(P)\land\mathsf{TypeWF}(\mathsf{abs}(P)).
\]
Here $\mathsf{DataWF}$ requires bound data and process variables and
well-typed expressions, whose operators are total on well-typed values.
These are the process well-formedness requirements used in the main text.

\input{sections/appendix/types/valid-compositions}

\input{sections/appendix/types/type-inference-proofs}
\input{sections/appendix/types/refinement-game}

%% file: sections/appendix/types/valid-compositions.tex
\subsubsection{Valid Compositions}
\applabel{valid-compositions}

Validity specifies the well-formedness and comparison requirements for the
initial compositions in \secref{types}. A composition fixes how its
components are grouped and which messages are internal at each parallel
operator.

Let $C$ be a finite binary composition tree with leaves $1,\ldots,n$ and
an interface at each internal node. The expression $C[R_1,\ldots,R_n]$
denotes the process or type composition obtained by placing $R_i$ at leaf
$i$. A process composition is valid when it satisfies $\mathsf{WFProc}$;
a type composition is valid when it satisfies $\mathsf{TypeWF}$, with its
components drawn from the regular types of \secref{type-syntax-semantics}.

For a fixed system $C[R_1,\ldots,R_n]$, a composition is valid for that
system when it is valid, uses the same tree $C$ and its interfaces, and its
corresponding components satisfy
\[
\begin{aligned}
 \Alphabet(S_i)&=\Alphabet(R_i),\\
 \EtaSet{S_i}&=\EtaSet{R_i},\\
 \PidSet{S_i}&=\PidSet{R_i}
 \qquad (i=1,\ldots,n).
\end{aligned}
\]
Here $S_i$ is the component at leaf $i$ of the second composition.
These equalities fix the polarized boundary actions, fair-choice
identifiers, and progress identifiers at every component.
All compatibility comparisons use the same set
$T\subseteq\PidSet{C[R_1,\ldots,R_n]}$ of required progress identifiers.
Thus types quantified over in the optimality and inference-failure results
form valid compositions for the fixed process system.

%% file: sections/appendix/types/type-inference-proofs.tex
\subsubsection{Proofs from \secref{type-checking-inference}}
\label{app:sec4-type-proofs}

We will first establish the properties of abstraction needed to relate process steps to type steps. These properties will then allow us to establish conformance and transfer strategies from types to processes.

\begin{lemma}[Abstraction and substitution]
\label{lem:abstraction-substitution}
For capture-avoiding substitutions,
\[
\begin{aligned}
 \Alphabet(\mathsf{abs}(P))&=\Alphabet(P),\\
 \mathsf{abs}(P[x\mapsto v])&=\mathsf{abs}(P),\\
 \mathsf{abs}(P[X\mapsto Q])&=\mathsf{abs}(P)[X\mapsto\mathsf{abs}(Q)],\\
 \mathsf{abs}(P[\Omega])&=\mathsf{abs}(P)[\Omega],\\
 \mathsf{Body}(\mathsf{abs}(P))=\tkw{skip}
   &\Longleftrightarrow\mathsf{Body}(P)=\pkw{skip},\\
 \mathsf{abs}(C^P_{B,K,\mathbf v})&=C_{\mathsf{abs}_B(B),K}.
\end{aligned}
\]
Data well-formedness is preserved by process steps.
\end{lemma}
\begin{proof}
The abstraction properties follow by structural induction. Value substitution affects only expressions, which abstraction removes. Process substitution commutes with abstraction because both follow the same constructors, with bound variables renamed consistently under recursion. Re-signing also commutes with abstraction, since the alphabets and interfaces used by $\Pi_I$ are preserved.

The $\tkw{skip}$ equivalence follows from the abstraction clauses, including the binding clause $\tkw{skip}\tseq \mathsf{abs}(P)$.

Consider now conjunctive residuals. A consumed receive leaves $P_i[x\mapsto v_i]$, whose abstraction is $\mathsf{abs}(P_i)$. Unconsumed heads remain singleton exclusive offers, and abstraction preserves their empty-interface product. This gives the required residual equation.

It remains to establish preservation of data well-formedness. Here, value substitution preserves expression typing by induction. Applying this property at receives and bindings, together with capture-avoiding process substitution at recursive unfolding, establishes preservation across process steps. Totality of expression operators ensures expression progress.
\end{proof}

Matching process steps also requires abstraction to preserve the side conditions governing their execution.

\begin{lemma}[Erasure preserves local side conditions]
\label{lem:erasure-side-conditions}
Abstraction preserves $\Complete$ and $\FutOff$. It also preserves enabled branch indices and resolution under erased environments. Erasure commutes with $\Need$, $\Sync_I$, label combination, and $\LocalEnv$ formed from chosen current/future offer pairs.
\end{lemma}

\begin{proof}
Let us first consider completion and future offers. Both are preserved by structural induction: bindings contribute an empty $\tkw{skip}$ menu and a true completion flag, prefixes retain their heads and continuations, and the remaining constructors preserve the corresponding unions, projections, and completion tests.

The other side conditions depend only on information retained by erasure. Branch needs and resolution inspect bare actions, synchronization depends on message names and polarities, and $\LocalEnv$ uses unions, restrictions, and subtraction. These operations are therefore unchanged by erasure.

For offer pairs, we proceed by induction on their derivations. Atomic heads lose only payload values, clock pairs are empty, and conditionals retain their source union menus. Exclusive commitment and default remove the same heads, while waiting preserves the same promised heads.

Conjunctive steps use the same set $K$, with their residual menus identified by Lemma~\ref{lem:abstraction-substitution}. Sequence combines the left pair with the same right menu under the same completion flag, while parallel hides the same interface from its component pairs. Hence, every process offer pair erases to the corresponding abstract offer pair.
\end{proof}

With these properties in place, we can establish the correspondence between process steps and their abstractions.

\begin{lemma}[One-step abstraction]
\label{lem:one-step-abstraction}
If $\mathsf{DataWF}(P)$ and $\lstep{E}{P}{\alpha}{P'}$, then
\[
 \lstep{\erase(E)}{\mathsf{abs}(P)}{\erase(\alpha)}{\mathsf{abs}(P')}.
\]
Every decoration of the process step erases to a decoration of this type step.
\end{lemma}

\begin{proof}
We proceed by induction on the process-step derivation, matching each process rule with its abstract counterpart.

For the basic rules, a firing prefix abstracts to an atomic message followed by its continuation, which the completing sequence rule reaches in the same step. Waiting prefixes use the non-completing sequence rule. Bindings similarly use a silent $\tkw{skip}$ step followed by completion, with value substitution leaving the abstract target unchanged. Atomic control rules preserve their guards and erased labels, while an abstract conditional selects the branch determined by the concrete guard.

Consider next sequence and offers. For sequence, abstraction reflects whether the left target is $\tkw{skip}$, allowing the same completing or rebuilding rule. Substitution commutation identifies recursive unfolding targets. Exclusive offers preserve their enabled heads and resolution tests, while conjunctive offers communicate the same heads and have matching residuals by Lemma~\ref{lem:abstraction-substitution}.

The remaining case is parallel composition. Here, the induction hypotheses match the component steps and their chosen offer pairs. Lemma~\ref{lem:erasure-side-conditions} preserves local environments and synchronization, while reflection of $\tkw{skip}$ determines the same product or fusion rule. Survivor entry additionally matches the finished component's silent step, with re-signing commutation identifying the target.

The correspondence between decorations follows from Lemma~\ref{lem:erasure-side-conditions}.
\end{proof}

So far, we have shown that every process step can be matched by its abstraction. Conformance also requires the converse availability property: whenever an abstract step is enabled, the process must admit a corresponding step.

The difficulty here is that abstraction removes data-dependent guards. An abstract step may therefore select a branch different from the one selected by the process. The following lemma establishes availability without requiring both steps to make the same branch choice.

\begin{lemma}[Decorated steps lift]
\label{lem:decorated-steps-lift}
Let $P$ be reachable from a well-formed process. If a step of
$\mathsf{abs}(P)$ under $\erase(E)$ exposes $(C,F)$, then $P$ has a step
under $E$ exposing $(O_C,F)$ with $\erase(O_C)=C$. Its label agrees on
message and clock observations; fair-choice outcomes may differ.
\end{lemma}

\begin{proof}
Fix the abstract derivation, including its component offer pairs. We will construct a corresponding process step by restoring concrete payloads and environments while preserving the recorded offers.

We begin with the payloads. Data preservation ensures that expressions at every active process occurrence are well typed. We can therefore evaluate active sends before assigning values to receives. This is possible because no step evaluates a continuation before its binding or receive has executed. Message ownership also gives a unique sending component for each interface message and separates simultaneous conjunctive branches.

Using these payloads, we reconstruct the concrete current offers from their direct offer clauses, combining them through sequence and hidden union for parallel composition. Their erasure is the recorded abstract current set, while future offers remain unchanged.

The restored offers allow us to construct the local environments. Starting from $E$, each local environment is formed using the restored offers of the sibling component. By Lemma~\ref{lem:erasure-side-conditions}, these environments erase to those recorded in the abstract derivation. Receives obtain their values from the resulting environments, with interface receives taking the values of their corresponding sibling sends.

With the offers and environments fixed, we can now lift the derivation inductively. Atomic and offer guards agree after erasure. A conditional may select a different branch, but both choices expose the same source union menu and empty current set. Fair-choice outcomes may therefore differ without changing message or clock observations. Bindings and recursive unfolding follow from substitution and menu commutation.

Let us consider the structural rules. For sequence, the concrete left target may have a different completion status from the abstract target. We therefore select the completing or rebuilding rule according to the concrete target. Both rules preserve the recorded pair, since they combine the same child pair with the right source menu under the same source completion flag.

For parallel composition, the restored pairs satisfy both local-environment equations, and their labels still synchronize. The step fuses if both targets finish. Otherwise, it retains their product when both sources are unfinished, or enters the appropriate survivor. A finished source steps only to itself. These cases are exhaustive and expose the same hidden component pairs, even when the concrete and abstract targets differ.

Finally, consider the clock environment. Only clock, sequence, and two-component parallel rules apply here. Induction gives the same clock step and abstract target. Survivor entry is impossible because a finished component cannot step under this environment, and all clock offer pairs are empty.
\end{proof}

The two step-correspondence lemmas now give the properties needed to establish conformance between a well-formed process and its abstraction.

\DataAbsSound*
\begin{proof}
Fix a well-formed process $P_*$ and relate each reachable residual $P$ to $\mathsf{abs}(P)$. The initial pair is related.

We will show that this relation satisfies both conformance clauses. For step matching, data preservation and Lemma~\ref{lem:one-step-abstraction} match every process step with an abstract step whose target remains related.

For availability reflection, observe that every abstract step admits a decoration: direct rules provide a pair, sequence extends the left pair, and parallel requires component pairs. Lemma~\ref{lem:decorated-steps-lift} therefore supplies a corresponding process step under the concrete environment.

An inadmissible concrete environment also has an inadmissible erasure and introduces no additional abstract availability obligation.
\end{proof}

We will now use conformance to transfer a type strategy to the process. The idea is to maintain a matching type history alongside the process execution, allowing the type strategy to determine the environment choices of the process strategy.

\StrategyTransfer*
\begin{proof}
Let $R$ witness $P\models\tau$. To transfer the strategy, we first choose a concrete representative for each type environment. For every type environment $e$, let $\gamma(e)$ be a concrete environment whose erasure is $e$ and whose bare current and future sets agree with those of $e$, with
\[
\gamma((\emptyset,\emptyset))=(\emptyset,\emptyset).
\]
For each boundary send, choose any well-typed payload.

We will construct the process strategy together with a matching type history. Initially, the histories end in $(P,\tau)\in R$. Suppose a process history has prefix $h$, with selected matching type prefix $g$. We define the process strategy by following the type strategy through $\gamma$:
\[
 S_P(h)=\gamma(S_\tau(g)).
\]

To maintain the correspondence, consider any possible process step under this choice. Conformance provides a matching type step with the erased environment and label and a related target. Selecting one such step extends $g$ consistently. By induction, every generated finite process prefix therefore has a matching type history. On prefixes that the strategy does not generate, it is defined arbitrarily.

This correspondence also extends to completed executions. An infinite process execution produces an infinite type execution with the erased log. If the process instead becomes stuck under its chosen environment, availability reflection ensures that the related type is also stuck under the erased choice. Both completions record that choice with $\epsilon$ and subsequently repeat the empty-environment failure entry.

Thus, every completed process log erases to a completed type log generated by $S_\tau$. Since $\Phi$ is value-oblivious, satisfaction transfers to the process log.

It remains to establish honesty. Erasure preserves clocks, while $\gamma$ preserves bare availability and future sets. Hence, every same-cycle honesty inclusion of the matching type history also holds for the process history, proving honesty of $S_P$.
\end{proof}

Next, we will show that conformance is preserved under type refinement. This allows us to replace a type by a further abstraction without losing conformance.

\ConformanceRefinement*
\begin{proof}
Let $R_{pt}$ and $R_{tt}$ witness the two premises. We will construct the required relation by composing them through their shared type residual:
\[
 R=\{(P_0,\tau'_0)\mid\exists\tau_0.
       (P_0,\tau_0)\in R_{pt}\land(\tau_0,\tau'_0)\in R_{tt}\}.
\]

Consider first availability. A decorated step of $\tau'_0$ is matched by a step of $\tau_0$ under the same erased environment through refinement. Conformance then supplies the corresponding process step.

For step matching, conformance matches a process step with a step of $\tau_0$. Decorating this step and applying refinement's response clause gives a step of $\tau'_0$ with the same erased label and a target related through $R$.

The initial pair also belongs to $R$, establishing the required conformance.
\end{proof}

With abstraction soundness and preservation under refinement established, type-checking soundness follows directly.

\TypeCheckingSound*
\begin{proof}
For a well-formed process $P$, \lemref{data-abs-sound} gives
\[
P\models\mathsf{abs}(P).
\]
Successful checking establishes
\[
\mathsf{abs}(P)\sqsubseteq\tau.
\]
Applying \lemref{conformance-refinement} gives the result.
\end{proof}

We will now establish decidability of type checking and bound the complexity of refinement checking.

\TypeCheckingDecidable*
\begin{proof}
Abstraction visits each syntax node once and introduces a constant number of nodes. With shared subterms and alphabet annotations, it therefore takes $O(|P|)$ time. Comparing the resulting abstraction with a given finite type is decidable.

To bound the checking cost, consider the safety game of \appref{refinement-game}. Here, checking reduces to refinement of two finite automata. The game has $|Q|\,|Q'|$ pair vertices and at most $|\delta|\,|Q'|$ challenge vertices.

A matching response is determined by edges with equal environments, labels, and complete offer pairs. Indexing these signatures gives at most $|\delta|\,|\delta'|$ response edges. Availability and completion violations are identified within the same product bound.

For automaton sizes $m=|Q|+|\delta|$ and $n=|Q'|+|\delta'|$, the arena therefore has size $O(mn)$.

The losing attractor of the violation sink can be computed using a worklist. Each vertex is removed once, while response counters and predecessor lists process each edge once. This gives $O(mn)$ time and space, measured after construction of the two automata.
\end{proof}

Finally, let us consider compositional conformance. Here, the checked refinements can be lifted through parallel composition.

\myparagraph{Compositional conformance of checked types}
\label{app:compositional-conformance-proof}
\noindent\textbf{\corref{compositional-conformance} (full statement).}\enspace
Let $P_1\pparop_I P_2$ be a valid process composition and $\tau_1\tpar{I}\tau_2$ a
valid type composition for it. If
$\checkdo{P_i}{\tau_i}$ for $i=1,2$, then
\[
  P_1\pparop_I P_2\models\tau_1\tpar{I}\tau_2.
\]

\begin{proof}[Proof of \corref{compositional-conformance}]
Checking gives $\mathsf{abs}(P_i)\sqsubseteq\tau_i$. Since both compositions are well formed and corresponding alphabets agree, Lemma~\ref{app:refinement-product-lifting} gives
\[
 \mathsf{abs}(P_1\pparop_I P_2)
 =\mathsf{abs}(P_1)\tpar{I}\mathsf{abs}(P_2)
 \sqsubseteq\tau_1\tpar{I}\tau_2.
\]

Applying inferred-type soundness to the composed process, followed by conformance under refinement, establishes the result.
\end{proof}

%% file: sections/appendix/types/refinement-game.tex
\subsubsection{Refinement Game}
\applabel{refinement-game}

The checker compares two well-formed types at the same boundary alphabet.
It repeatedly removes state pairs that fail the refinement conditions. We
define this procedure on their automata and show that the initial pair
survives exactly when the types are related by refinement.

Let their finite automata be
\[
 \Aut_\tau=(Q,\Sigma,\delta,i,F),\qquad
 \Aut_{\tau'}=(Q',\Sigma,\delta',i',F').
\]
An edge signature contains its environment, label, and complete offer pair.
Let $\mathsf{Sig}(q)$ be the set of signatures of outgoing edges from $q$.
For a candidate relation $R\subseteq Q\times Q'$, define
\[
\begin{aligned}
 \mathsf{Good}_R(q,q')\ \Longleftrightarrow\quad
 &\mathsf{Sig}(q')\subseteq\mathsf{Sig}(q)\\
 &{}\land(q\in F\Longleftrightarrow q'\in F')\\
 &{}\land\ \forall s,r.\ (q\xrightarrow{s}r\in\delta\ \Longrightarrow\\
 &\hspace{14mm}\exists r'.\ q'\xrightarrow{s}r'\in\delta'
                    \land(r,r')\in R).
\end{aligned}
\]
The first clause matches each abstract signature on the precise side,
without relating its target. The second checks completion. The third
matches each precise edge by an abstract edge leading back to $R$.

Starting with $R_0=Q\times Q'$, remove pairs that fail these clauses:
\[
 R_{k+1}=\{(q,q')\in R_k\mid\mathsf{Good}_{R_k}(q,q')\}.
\]
The finite sequence stabilizes. The checker accepts exactly when $(i,i')$
remains. Equivalently, Spoiler chooses a precise edge and Duplicator answers
with an abstract edge of the same signature. Availability or completion
violations and unanswered challenges enter a losing sink. Duplicator wins
by avoiding that sink. A pair with no precise challenge may stay in place.

\begin{lemma}[Refinement game characterization]
The initial pair survives exactly when $\tau\sqsubseteq\tau'$.
\end{lemma}
\begin{proof}
The type--automaton correspondence preserves decorated steps in both
directions and identifies final states with $\tkw{skip}$. It therefore
transports the three clauses of $\mathsf{Good}$ to those of refinement.
At the fixed point, the remaining relation gives a refinement witness.
Conversely, a refinement witness transports to a post-fixed relation on
states. No iteration removes one of its pairs, so its initial pair remains.
\end{proof}

%% file: sections/appendix/types/type-compatibility-checking.tex
\subsection{Formal Type-Compatibility Construction}
\applabel{type-compatibility-construction}

To prove the game reduction of \secref{typecompat}, we first relate semantic
execution prefixes to automaton states and establish the log properties used
by the game. We then construct the game and prove its correspondence with
semantic strategies. Finally, we prove process safety and the complexity
bound.

Throughout, require $\mathsf{TypeWF}(\Theta_0)$ and fix the parenthesized composition
\[
  \Theta_0=\tau_1\tpar{I_1}\cdots\tpar{I_{n-1}}\tau_n,
  \qquad I=I_1\cup\cdots\cup I_{n-1},
\]
and a progress set \(T\subseteq\Pids(\Theta_0)\).

\subsubsection{Strategies, Executions, and Logs}

To relate semantic executions to game plays, we pair each finite execution
prefix from \secref{typecompat} with a matching automaton state.
For the canonical automaton, the current residual determines its matching
state. Represented prefixes then correspond one-to-one to the prefixes of
\secref{pact-safety}, so the strategies in
Theorem~\ref{thm:type-compatibility-soundness-completeness} are those of
\secref{pact-safety}. The constructor automaton may have several states representing a
residual, so transferring a strategy also tracks the matching state.
Let \(\widehat{\mathsf{Pref}}(\Theta_0)\) contain these represented prefixes,
with matching states supplied by the relation of
\lemref{type-automata-bisimulation}. A strategy is a function
\[
  \Strategy_\Theta:
  \widehat{\mathsf{Pref}}(\Theta_0)
  \longrightarrow
  \mathsf{Env}_I.
\]
Here $\mathsf{Env}_I$ is the finite set of well-formed environments on the
boundary of $\Theta_0$. The current matching state is finite controller
memory. After the strategy chooses an environment, the adversarial resolution
extends the semantic prefix and updates the state to any corresponding
successor allowed by the relation. The log set includes every such resolution.
For a prefix $h_k$ of $k$ steps, let $\widehat{h}_k$ denote its
represented form. An execution is generated by $\Strategy_\Theta$ when
$e_k=\Strategy_\Theta(\widehat{h}_k)$ at each ordinary step.

Completed executions follow the convention in \secref{pact-safety},
with payloads erased.
An enabled but dishonest choice executes normally; its log is rejected by
the honesty conjunct of $\Psi_T$. Let $\mathsf{Log}(h)$ denote the
sequence of environment--action entries of a completed execution $h$.
Then
\[
  \Logs^\omega_{\Strategy_\Theta}(\Theta_0)
  =
  \{\mathsf{Log}(h)\mid
    h\text{ is an infinite completed execution generated by }
    \Strategy_\Theta\}.
\]
The adjacent honesty condition also constrains every later environment in
the same cycle:
\begin{lemma}[Same-cycle honesty]
\label{lem:type-log-honesty}
For every failure-free log
\(\ell=\langle e_0,\alpha_0\rangle
       \langle e_1,\alpha_1\rangle\cdots\),
\(\Honest(\ell)\) holds exactly when, for every \(i<j\) such that none of
\(\alpha_i,\ldots,\alpha_{j-1}\) contains \(\tclock\),
\[
  \Avail(e_j)\subseteq\Fut(e_i).
\]
\end{lemma}
\begin{proof}
The right-to-left direction takes \(j=i+1\).  Conversely, each adjacent
honesty step gives both
\(\Avail(e_{k+1})\subseteq\Fut(e_k)\) and
\(\Fut(e_{k+1})\subseteq\Fut(e_k)\).  Chaining the latter inclusions and using
the final adjacent step yields
\(\Avail(e_j)\subseteq\Fut(e_i)\).
\end{proof}

\noindent\textbf{\propref{type-objective-regular} (full statement).}\enspace
For every finite composition \(\Theta_0\) and finite progress set \(T\),
\(\Psi_T\) is \(\omega\)-regular.

\begin{proof}[Proof of \propref{type-objective-regular}]
Honesty is a safety property recognized by a finite monitor that stores either
\(\star\), after a clock, or the preceding future set.  The remaining condition
is the finite disjunction of
\[
\begin{array}{ll}
\mathit{live}:&
\tclock\text{ and every }\tmarker{t}\ (t\in T)
\text{ occur infinitely often},\\[1mm]
\mathit{unfair}_{\eta,b}:&
\tclock\text{ and }\eta_b\text{ occur infinitely often, while }
\eta_{\bar b}\text{ occurs only finitely often}.
\end{array}
\]
The live case is generalized B\"uchi and each unfair case is a Rabin
condition.  There are finitely many observations and fairness identifiers, so
their finite union, intersected with the honesty monitor, is
\(\omega\)-regular.
\end{proof}

\input{sections/appendix/types/detailed-type-compatibility-game}

\subsubsection{Soundness and Complexity Proofs}
\label{app:type-compatibility-process-soundness}

A winning type strategy must also ensure \termref{def:comm-safety}{communication safety} for conforming
processes. We prove this using strategy transfer, then derive decidability
and the complexity bound from the finite Rabin game.

\noindent\textbf{Theorem~\ref{thm:type-compatibility-process-soundness} (full statement).}\enspace
Fix the process composition
\(Q_0=P_1\pparop_{I_1}\cdots\tparop_{I_{n-1}}P_n\), the matching type
composition \(\Theta_0\) with the same nesting and interfaces, and a progress
set \(T\). Require \(\Theta_0\) to be well formed and the two compositions
to have the same fair-choice and progress identifiers. If
\[
Q_0\models\Theta_0
\qquad\text{and}\qquad
\TypeCompat_{I,T}(\tau_1,\ldots,\tau_n),
\]
then the processes are compatible:
\[
\exists\Strategy_{Q_0}.\
\Honest(\Strategy_{Q_0})
\land
\ProcJoin_{\Strategy_{Q_0},T}(P_1,\ldots,P_n).
\]

\begin{proof}[Proof of Theorem~\ref{thm:type-compatibility-process-soundness}]
By \thmref{type-compatibility-soundness-completeness}, a winning type strategy
makes every completed type log satisfy $\Psi_T$, and is therefore honest.
Let $\Phi(\ell)=\Psi_T(\erase(\ell))$. This objective is value-oblivious
and agrees with $\Psi_T$ on type logs. Apply \lemref{strategy-transfer}
to $Q_0\models\Theta_0$. Every completed process log then satisfies $\Phi$,
and the transferred strategy is honest.

Erasure preserves clock, fairness, and progress observations. Equality of
the identifier sets makes process and type fairness range over the same
identifiers and gives them the same progress set $T$. The transferred
objective therefore gives process liveness. If a reachable process prefix
were stuck under its chosen environment, its failure completion would have
only $\epsilon$ observations thereafter, contradicting clock recurrence.
Thus the strategy is also sync-compatible, proving the result.
\end{proof}

\label{app:type-compatibility-complexity}
\begin{proof}[Proof of \thmref{type-compatibility-decidable}]
The construction is effective and always produces a finite Rabin game.
Winning regions in finite Rabin games are decidable, so the reduction directly
implies that type compatibility is decidable.  The stated bound follows from
standard Rabin-game algorithms~\cite{majumdarRabinGamesColourful2024}.  The
parameters \(N\) and \(M\) already include the automaton product, honesty
memory, and round-robin monitors; for fixed \(K\), the query is polynomial in
the monitored arena.
\end{proof}

%% file: sections/appendix/types/detailed-type-compatibility-game.tex
\subsubsection{Detailed Type-Compatibility Game}
\label{sec:tc-game-appendix}

We will construct the game in two steps. First, the safety arena checks
honesty and sends unresolvable choices to a losing sink. Second, we add
monitors for clock, fairness, and progress. We then prove game correspondence.

Let the automaton of the fixed parenthesized composition, retaining its
individual interfaces, be
\[
  \Aut_I=\Aut(\Theta_0)=(S,\Sigma_I,\Delta_I,q_0,F_I).
\]
For \(e=(E_C,E_F)\in\mathsf{Env}_I\), we will omit the offer pair when
referring to product edges:
\[
q\xrightarrow[e]{\alpha}q'
\quad\triangleq\quad
\exists O_C,O_F.\
\atrans{q}{E_C}{E_F}{\alpha}{O_C}{O_F}{q'}\in\Delta_I.
\]
\nparagraph{Honesty memory}
To enforce honesty, we store the preceding future set until the next clock:
\[
  \mathsf{Mem}_I=
  \{\star\}
  \cup
  \{E_F\mid(E_C,E_F)\in\mathsf{Env}_I\}.
\]
Here \(\star\) marks a fresh cycle, with no preceding future set to
constrain the next choice. A choice respects this memory when
\[
\begin{aligned}
\mathsf{Adm}(\star,e)&\quad\text{for every }e,\\
\mathsf{Adm}(F,e)&\triangleq\Avail(e)\subseteq F
   &&(F\neq\star),
\end{aligned}
\]
After a step under \(e=(E_C,E_F)\), the memory is updated by
\[
\mathsf{next}(e,\alpha)=
\begin{cases}
\star,&\tclock\in\alpha,\\
E_F,&\tclock\notin\alpha.
\end{cases}
\]
Thus the memory resets only when the fired observation contains a clock tick.

\nparagraph{Safety arena}
The arena alternates between environment choices and system resolutions.
Player~0 vertices store the automaton state and honesty memory;
Player~1 vertices also store the chosen environment:
\[
\begin{aligned}
V_0
  &=\{q^{P,F}\mid q\in S,\ F\in \mathsf{Mem}_I\},\\
V_1
  &=\{q^{F,e}\mid q\in S,\ F\in \mathsf{Mem}_I,\ e\in\mathsf{Env}_I\}
    \cup\{\bot\}.
\end{aligned}
\]
At \(q^{P,F}\), Player~0 chooses an environment:
\[
E_{\mathsf{choice}}
=
\{q^{P,F}\xrightarrow{\theta_e}q^{F,e}
  \mid q\in S,\ F\in \mathsf{Mem}_I,\ e\in\mathsf{Env}_I\}.
\]
For an honest choice, Player~1 selects a matching product transition:
\[
\begin{aligned}
E_{\mathsf{resolve}}
=
\{q^{F,e}\xrightarrow{\langle e,\alpha\rangle}
  q'^{P,\mathsf{next}(e,\alpha)}
\mid{}&
\mathsf{Adm}(F,e),\
q\xrightarrow[e]{\alpha}q'\}.
\end{aligned}
\]
Each resolution contributes \(\langle e,\alpha\rangle\) to the log,
including when \(\alpha=\epsilon\).

A dishonest choice or one with no matching transition enters the absorbing sink:
\[
\begin{aligned}
e_\bot&=(\emptyset,\emptyset),\\
E_{\mathsf{fail}}
={}&
\{q^{F,e}\xrightarrow{\langle e,\epsilon\rangle}\bot
\mid
\neg\mathsf{Adm}(F,e)
\ \text{or}\
\nexists\alpha,q'.\ q\xrightarrow[e]{\alpha}q'\}\\
&{}\cup
\{\bot\xrightarrow{\langle e_\bot,\epsilon\rangle}\bot\}.
\end{aligned}
\]
The safety arena is
\[
A_I^{\mathsf H}(\tau_1,\ldots,\tau_n)
=
(V_0,V_1,
  E_{\mathsf{choice}}\cup E_{\mathsf{resolve}}\cup E_{\mathsf{fail}},
  q_0^{P,\star}).
\]

For a play \(\pi\), let \(\mathsf{Log}(\pi)\) retain the labels on resolving
and failure edges and omit the bookkeeping labels \(\theta_e\).
The failure loop makes every log infinite. By \lemref{type-log-honesty},
every log that avoids the sink is honest.

\nparagraph{Monitor product}
To monitor the remaining clock, fairness, and progress conditions, we
index their disjuncts by
\[
  D_T=\{\mathsf{live}\}
      \cup(\Etas_I\times\{\mathsf{true},\mathsf{false}\}).
\]
Here \(\Etas_I=\Etas(\Theta_0)\) contains the fair-choice identifiers.
Each case requires some observations infinitely often and permits others
only finitely often:
\[
\begin{array}{lll}
\mathsf{Req}(\mathsf{live})
  =\{\tclock\}\cup\{\tmarker{t}\mid t\in T\},
&\qquad&
\mathsf{Forbid}(\mathsf{live})=\emptyset,\\[1mm]
\mathsf{Req}(\eta,b)
  =\{\tclock,\eta_b\},
&&
\mathsf{Forbid}(\eta,b)=\{\eta_{\bar b}\}.
\end{array}
\]
Here \(\bar b\) denotes the other Boolean outcome.

For each \(d\), let \(\mathcal M_d\) be a deterministic round-robin monitor
that waits for the observations in \(\mathsf{Req}(d)\) in a fixed order.
On a resolving label \(\langle e,\alpha\rangle\), it advances when the
awaited observation belongs to \(\alpha\) and emits a one-step completion
flag at the end of each round. An \(\epsilon\)-step advances no monitor.
The flag therefore recurs exactly when every required observation recurs.

We combine \(A_I^{\mathsf H}\) with all monitors in a synchronous product,
updating them only on resolving edges. The Rabin condition uses sets of
vertices, so we insert a record vertex on each resolving edge to store its
observations and completed rounds. For \(\xi=(u,\langle e,\alpha\rangle,v)\),
the edge becomes
\[
u\xrightarrow{\langle e,\alpha\rangle}z_{\xi}
\xrightarrow{\theta}v.
\]
Here \(z_\xi\) stores \(\alpha\) and the flags after the monitors update
on the first edge. Only that edge contributes to the log; \(\theta\) is
bookkeeping. Each record vertex belongs to Player~1 with sole successor \(v\).

Let \(\widehat A_{I,T}\) denote the resulting arena, with initial vertex
\[
  s_0=(q_0^{P,\star},\mathbf c_0).
\]
Here \(\mathbf c_0\) contains the initial states of all monitors.
To define the Rabin pairs, we collect vertices that complete a round and
those carrying forbidden observations:
\[
\begin{aligned}
F_d
  &=\{v\in\widehat A_{I,T}\mid
      v\text{ carries the completion flag of }\mathcal M_d\},\\
B_d
  &=\{v\in\widehat A_{I,T}\mid
      v\text{ carries some }x\in\mathsf{Forbid}(d)\}.
\end{aligned}
\]
No sink vertex belongs to any \(F_d\). The pair \((F_d,B_d)\) requires
infinitely many visits to \(F_d\) and only finitely many to \(B_d\).
The resulting Rabin objective is
\[
R_T=\{(F_d,B_d)\mid d\in D_T\}.
\]
The resulting game is
\[
G_{I,T}(\tau_1,\ldots,\tau_n)
=
(\widehat A_{I,T},R_T).
\]

\begin{lemma}[Objective monitor]
\label{lem:type-objective-monitor}
For every play \(\pi\) of the monitored arena,
\[
  \pi\models R_T
  \quad\Longleftrightarrow\quad
  \Psi_T(\mathsf{Log}(\pi)).
\]
\end{lemma}
\begin{proof}
Let us first consider a play that reaches \(\bot\). Its log thereafter
repeats \(\langle e_\bot,\epsilon\rangle\), so the clock cannot recur and
no monitor completes another round. Both conditions are therefore false.
For a play that avoids the sink, the safety arena ensures honesty.

It remains to compare the recurrence conditions. The live pair
\((F_{\mathsf{live}},\emptyset)\) accepts exactly when the clock and every
marker in \(T\) recur. For a fairness identifier \(\eta\), the pair for
\(\mathsf{true}\) accepts exactly when the clock and
\(\eta_{\mathsf{true}}\) recur while \(\eta_{\mathsf{false}}\) occurs
only finitely often; the other pair is symmetric. Thus the union accepts
precisely the logs with recurring clock that are live or unfair, as
required by the fairness implication in \(\Psi_T\).
\end{proof}

\nparagraph{Game correspondence}
To establish game correspondence, we relate boundary strategies to
Player~0 strategies and compare their outcomes:

\noindent\textbf{Theorem~\ref{thm:type-compatibility-soundness-completeness} (full statement).}\enspace
For every fixed finite well-formed parenthesized composition \(\Theta_0\), its interfaces
\(I_1,\ldots,I_{n-1}\), and progress set \(T\),
\[
\TypeCompat_{I,T}(\tau_1,\ldots,\tau_n)
\Longleftrightarrow
\exists\Strategy_\Theta.\
\forall\ell\in\Logs^\omega_{\Strategy_\Theta}(\Theta_0).\
\Psi_T(\ell).
\]

\begin{proof}
We will transfer strategies through \lemref{type-automata-bisimulation}.
A represented semantic prefix and its related automaton state have the
same chosen environments, fired labels, and offers. A semantic boundary
strategy therefore gives a Player~0 strategy that chooses the same
environment. Conversely, a Player~0 strategy gives a semantic controller
that stores the current arena state.

For an honest enabled choice, Player~1's resolutions are exactly the
semantic nondeterministic successors. By
\lemref{type-automata-bisimulation}, a choice has no semantic successor
exactly when the product has no resolving edge. Both executions then enter
their failure sinks after the same choice, regardless of its honesty.
In these cases, the outcome logs agree entry for entry.

An enabled dishonest choice is the only case in which the logs need not
match. The semantic execution proceeds and violates honesty, while the
arena enters \(\bot\). Both outcomes lose. Applying
\lemref{type-objective-monitor} therefore gives
\[
s_0\in
\TypeWinRegion(G_{I,T}(\tau_1,\ldots,\tau_n))
\Longleftrightarrow
\exists\Strategy_\Theta.\
\forall\ell\in\Logs^\omega_{\Strategy_\Theta}(\Theta_0).\
\Psi_T(\ell),
\]
By the definition of type compatibility, this is the required equivalence.
\end{proof}

\nparagraph{Size of the monitoring construction}
For a boundary with \(b\) bare message polarities, honesty requires at
most \(1+2^b\) non-failure memory values. The round-robin product has at most
\[
  (|T|+1)\,2^{2|\Etas_I|}
\]
counter combinations before reachable-state pruning. Both factors are
included in \(N\) and \(M\) in \thmref{type-compatibility-decidable}.
Neither adds Rabin pairs; the objective has \(K=1+2|\Etas_I|\) pairs.

%% file: sections/appendix/types/optimality.tex
\input{sections/appendix/types/type-optimality-proofs}

%% file: sections/appendix/types/type-optimality-proofs.tex
\subsection{Proofs from \secref{type-optimality}}
\applabel{type-optimality}

We will establish optimality by first showing that refinement is
preserved under alphabet restriction and parallel composition.
These properties will allow us to show that refinement preserves
compatibility, from which the optimality and inference-failure
results of \secref{type-optimality} follow.

\begin{lemma}[Common alphabet restriction]
\label{lem:refinement-alphabet-restriction}
Let $\rho\sqsubseteq\theta$ at alphabet $\Omega$, with both types closed,
alphabet well formed, interface scoped, and tail recursive. For
$A\subseteq\Omega$,
\[
 \rho[A]\sqsubseteq\theta[A].
\]
\end{lemma}
\begin{proof}
We will first show that alphabet restriction preserves decorated
steps. Let $\mathcal D_E(\rho)$ contain the tuples
$(\alpha,O_C,O_F,\sigma)$ of decorated steps of $\rho$ under $E$.
For every environment $E\in\mathsf{Env}_A$ admissible at alphabet $A$, we claim
\[
 \mathcal D_E(\rho[A])=
 \{(\alpha,O_C\cap A,O_F\cap A,\sigma[A])\mid
       (\alpha,O_C,O_F,\sigma)\in\mathcal D_E(\rho)\}.
\]

We establish this equality by simultaneous induction on step and
offer derivations. The atomic and control cases are immediate:
they inspect the same environment and retain their labels.
An undeclared action can wait with empty current offers.
Offer rules select the same enabled set, branch, default, or
waiting case, since intersection commutes with the relevant
future unions and differences.
For sequence, the completion flag is unchanged and the right
menu is projected, by
Lemma~\ref{lem:constructor-survivor-transport}.
Re-signing also preserves the $\tkw{skip}$ test of each target.

Let us consider recursion and parallel composition, where the
interaction with restriction requires more care.
For recursion, tail position gives
\[
(\rho_0[X\mapsto S])[A]=\rho_0[A][X\mapsto S[A]].
\]
Here, ordinary tail children receive $A$, and no free $X$
occurs below a parallel node or in a conjunctive branch.

At a parallel node with interface $J$, child $i$ receives
$A_i=(\Omega_i|_J)\cup A$.
Interface scoping ensures that removed boundary actions lie
outside $J$, so the projected component pairs retain exactly
their interface restrictions.
Both local-environment equations therefore use the same
environments before and after projection.
Hiding $J$ projects the published pair to $A$, while the
unchanged source and target completion tests select the same rule.

For survivor entry, the targets agree because
\[
(\sigma_i[A_i])[A]=(\sigma_i[\Omega])[A].
\]
This follows by induction down the target: both sides retain
each nested child's interface part and pass it the same incoming
boundary. Together, these cases establish both inclusions of
the displayed step equality.

We can now use this equality to transport the refinement witness.
Restrict the witness to residual pairs reachable on their
respective sides. Since the structural conditions are preserved
by steps, the equality applies at every such pair.
Project each pair by $A$.

For response, lift a projected precise step, apply the original
refinement witness, and project the matching abstract step.
For availability, lift an abstract signature, apply availability,
and project the resulting precise step.
In both cases, the projected offers agree, response targets
remain related, and the $\tkw{skip}$ tests still agree.
The projected relation therefore witnesses
$\rho[A]\sqsubseteq\theta[A]$.
\end{proof}

\begin{lemma}[Componentwise product lifting]
\label{app:refinement-product-lifting}
Fix the same nesting and interfaces on two well-formed type compositions.
If corresponding components refine each other at equal alphabets, then the
whole precise composition refines the whole abstract composition.
\end{lemma}
\begin{proof}
It suffices to establish refinement through one parallel node;
the general result then follows by induction on the nesting.

Consider the candidate relation consisting of equal-alphabet
refinement on the structural class of the preceding lemma,
together with pairs of alphabet-compatible products whose
corresponding children refine each other.
Product sources agree on being non-$\tkw{skip}$.

We will check the response and availability clauses.
For response, consider a precise product step.
Each child step has a matching abstract step by its refinement
witness, with the same local environment, label, and complete
offer pair.
The preserved pairs satisfy the same $\LocalEnv$ equations,
and the labels still synchronize.
Completion agreement selects the same parent rule.
Continuing targets remain related as products, while fused
targets are identical signed $\tkw{skip}$.
For survivor entry, common alphabet restriction gives refinement
after both child targets are re-signed by the parent alphabet.
All these rules publish the same hidden pair.

Availability requires a slightly different argument because
the precise and abstract targets need not agree.
Start with an abstract product step and use the children's
availability clauses to obtain precise steps with the same
local environments, labels, and pairs.
We then choose the parent rule according to the actual precise
targets: fuse if both finish, retain the product if both sources
are unfinished and a target continues, or enter the appropriate
survivor.
A finished source can only remain finished, so these cases
cover all possibilities.
Each rule publishes the required label and hidden pair;
availability does not require related targets.

Thus, the candidate witnesses refinement at this parallel node.
Induction on the fixed nesting completes the proof.
\end{proof}

With product lifting established, we can now show that replacing
components by more precise types preserves compatibility.

\noindent\textbf{\lemref{type-compat-refinement-monotone} (full statement).}\enspace
Let $\tau_i\sqsubseteq\tau_i'$ for every $i$. Fix the same nesting and
interfaces on both sides and require both compositions to be well formed.
Corresponding components must have the same boundary, fair-choice
identifiers, and progress identifiers. Then
\[
\TypeCompat_{I,T}(\tau_1',\ldots,\tau_n')
\Longrightarrow
\TypeCompat_{I,T}(\tau_1,\ldots,\tau_n).
\]

\begin{proof}
By product lifting, the corresponding compositions satisfy
$\Theta_0\sqsubseteq\Theta'_0$.
Since the component identifiers agree, both compositions use
the same fairness predicates, progress set, and hence $\Psi_T$.

Let $S'$ be a winning strategy for $\Theta'_0$.
We will construct a strategy for $\Theta_0$ by maintaining a
matching abstract history, as in \lemref{strategy-transfer},
but without payload concretization.
For a precise prefix $h$ with selected abstract match $g$,
choose $S(h)=S'(g)$.

Consider any precise step under this choice.
Decorating it and applying refinement's response clause gives
an abstract step with the same environment and label and a
related target.
We can therefore extend the matching history consistently.
By induction, every infinite precise execution has a matching
abstract execution with the identical log.

If a precise prefix becomes stuck under the selected environment,
the abstract prefix must also be stuck.
Otherwise, availability would supply a precise step.
Both failure completions consequently record the same failed
choice and empty-environment suffix.

Thus, every completed precise log is also a completed log
generated by $S'$ and satisfies $\Psi_T$.
The game reduction now gives compatibility of the precise tuple.
\end{proof}

We can now prove optimality of the inferred types.
The forward direction uses the inferred types themselves,
while the reverse direction follows from compatibility
monotonicity.

\noindent\textbf{\thmref{type-inference-optimality} (full statement).}\enspace
Fix a well-formed composition of $P_1,\ldots,P_n$, its interfaces $I$, and
\termref{def:progress-marker}{progress markers} $T$. Quantify over well-formed type
compositions with the same nesting and interfaces, and the same boundary,
fair-choice identifiers, and progress identifiers at each component. Then
\begin{align*}
\begin{array}{rcl}
\TypeCompat_{I,T}
 (\mathsf{abs}(P_1),\ldots,\mathsf{abs}(P_n))
 & \iff &
\exists \tau_1,\ldots,\tau_n.\\
\left(\bigwedge_i \checkdo{P_i}{\tau_i}\right)
\wedge
\TypeCompat_{I,T}(\tau_1,\ldots,\tau_n).
\end{array}
\end{align*}

\begin{proof}
We prove both directions.

First, suppose the inferred types are compatible:
\[
\TypeCompat_{I,T}(\mathsf{abs}(P_1),\ldots,\mathsf{abs}(P_n)).
\]
Choose $\tau_i=\mathsf{abs}(P_i)$ for every $i$.
Abstraction preserves alphabets, fair-choice identifiers, and
progress identifiers, while process well-formedness ensures
that the inferred composition is well formed.
Thus, this choice belongs to the quantified domain.
By reflexivity of refinement,
$\mathsf{abs}(P_i)\sqsubseteq\mathsf{abs}(P_i)$,
so $\checkdo{P_i}{\tau_i}$ follows from
\defref{data-oblivious-type-checking}.
This establishes the forward implication.

Conversely, suppose there exist types $\tau_1,\ldots,\tau_n$
such that
\[
\left(\bigwedge_i \checkdo{P_i}{\tau_i}\right)
\wedge
\TypeCompat_{I,T}(\tau_1,\ldots,\tau_n).
\]
By \defref{data-oblivious-type-checking},
$\mathsf{abs}(P_i)\sqsubseteq\tau_i$ for every $i$.
The quantified domain provides the required well-formedness,
fixed nesting and interfaces, and equality of corresponding
alphabets and identifiers.
We can therefore apply
\lemref{type-compat-refinement-monotone}, taking
$\mathsf{abs}(P_i)$ as the precise type and $\tau_i$
as the coarser type, to obtain
\[
\TypeCompat_{I,T}(\mathsf{abs}(P_1),\ldots,\mathsf{abs}(P_n)).
\]
This proves the reverse implication.
\end{proof}

Finally, the inference-failure result follows directly from
optimality: if the inferred types are incompatible, no compatible
tuple of checked types can exist.

\noindent\textbf{\corref{type-inference-failure} (full statement).}\enspace
Under the hypotheses and type domain of \thmref{type-inference-optimality},
if $\neg\TypeCompat_{I,T}
 (\mathsf{abs}(P_1),\ldots,\mathsf{abs}(P_n)),$
then there are no types $\tau_1,\ldots,\tau_n$ such that
$\left(\bigwedge_i \checkdo{P_i}{\tau_i}\right)
\wedge
\TypeCompat_{I,T}(\tau_1,\ldots,\tau_n).$

\begin{proof}
The claim is the contrapositive of the backward implication
of \thmref{type-inference-optimality}.
Indeed, any compatible tuple of checked types would imply
compatibility of the inferred tuple, contradicting the hypothesis.
\end{proof}

%% file: sections/appendix/evaluation.tex
\section{Evaluation}
\applabel{appendix-evaluation}

\input{sections/appendix/evaluation/implementation}

%% file: sections/appendix/evaluation/implementation.tex
\subsection{Implementation Optimizations}
\applabel{implementation-optimizations}

We implemented several optimizations to reduce the hardware generated by \anvil{} and improve the performance of compatibility checking.
The checker also supports parallel execution across multiple threads.

\myparagraph{Code generation}
The code generator translates each process thread into an event graph, where events represent control operations such as delays, branch and loop joins,
 and data selection for message sends. These events introduce control state in the generated hardware.
 We apply two optimizations to eliminate redundant state.

\begin{itemize}
  \item \emph{Shared delays at joins.}
  When multiple branches end with the same delay, the code generator can join the branches before applying the delay.
  This transformation is valid because delaying each input before a join is equivalent to delaying the joined result,
  provided the delays have no actions of their own and the join is their only reader. The optimization also applies
  to branches that end a loop iteration or re-enter a recursion. Since a one-cycle delay requires one flip-flop,
  delays merged into the same event can share a flip-flop, even when they originate from different conditionals.

  \item \emph{Data from the active sender.}
  When multiple program locations send the same message, the baseline implementation uses a selector register
  to remember the last sender and preserve its data throughout the message lifetime. However, if the lifetime
  is one cycle, the receiver only needs the data during the handshake cycle. In this case, the selector can
  be eliminated by taking the data directly from the active sender.
  This optimization requires each sender to drive its valid condition throughout its send and at most one sender to be active in any cycle.
  The latter is guaranteed by \anvil{}'s lifetime checker, which rejects overlapping uses of the same message.
  Thus, whenever the valid signal is asserted, the data comes from exactly one active sender.
\end{itemize}

\myparagraph{Compatibility checking}
The compatibility checker constructs the synchronous product of the participants' type automata and solves the
two-player graph game described in \secref{typecompat}. We optimize both the construction of the product and
the solution of the resulting game.

\begin{itemize}
  \item \emph{Pruned product construction.}
  A product transition combines member transitions whose guards and synchronization constraints are consistent.
  Rather than enumerating all combinations, the checker constructs them incrementally using depth-first search.
  At each step, it selects the member sharing the most messages with those already considered, allowing synchronization
  constraints to be checked earlier.

  The checker prunes a partial combination whenever it violates a constraint that cannot be repaired by adding more members.
  These constraints require agreement on clock advancement, matching communication partners, and consistent synchronization
  with members already considered. The remaining combinations undergo the full consistency check. Since only irreparable
  combinations are pruned, the resulting product is identical to that obtained by exhaustive enumeration.

  \item \emph{Integer arena and linear-time attractors.}
  The checker represents arena nodes as integers and stores edges in compressed arrays.
  Each edge records whether it advances the clock and which progress markers it carries,
  using a bit mask.

  This representation supports linear-time attractor computation. A Player~0 node enters
  the attractor when it has a winning successor, while a Player~1 node enters only when
  all its successors are winning. The checker maintains a counter of unresolved edges
  for each Player~1 node and constructs each predecessor index once.
  The counters only avoid rescanning edges, so the checker computes the same
  attractor as the standard fixpoint while examining each edge once.

  \item \emph{Generalized B\"uchi without degeneralization.}
  The checker takes $T$ to contain every participant's progress marker. Without fair-choice identifiers, the winning condition then requires both clock advancement and every participant's progress marker to occur infinitely often. This is a generalized B\"uchi objective. A conventional approach reduces it to a B\"uchi objective using a round-robin monitor, which replicates the arena for each participant.
  Instead, the checker solves the generalized objective directly. In each fixpoint iteration, it computes the regions satisfying the clock and individual progress requirements from the same candidate region. These computations run in parallel, and their results are intersected to obtain the next candidate region.
  The first candidate region contains the nodes from which Player~0 can avoid the failure sink. A node removed in an iteration cannot satisfy one of the requirements even within the current candidate region, so it cannot be winning. Conversely, once the region stops changing, Player~0 can force each requirement to recur from every node of the region while staying in it, so it can satisfy all requirements in turn, infinitely often.
  Thus, the fixpoint is the same winning region as the one obtained with the monitor, without replicating the arena.

  With fair-choice identifiers, the checker first computes this fixpoint while letting a fair node make progress through any of its fair arms.
  The resulting region over-approximates the winning region and is closed under Player~1 moves. When it contains the initial node,
  the checker solves the Rabin game of \appref{type-compatibility-construction} only within this region.
\end{itemize}